\documentclass[reprint,showpacs,superscriptaddress,english,twocolumn,notitlepage,longbibliography]{revtex4-2}
\usepackage{times}
\usepackage{graphicx}
\usepackage{dcolumn}
\usepackage{amsmath,empheq,amssymb, bm, physics, bbm, mathrsfs}
\usepackage{array}
\usepackage{booktabs, multirow}
\usepackage{hyperref}
\usepackage{float}

\usepackage{color}
\usepackage[dvipsnames]{xcolor}
\hypersetup{
    unicode=true,
    bookmarks=true, bookmarksnumbered=true, bookmarksopen=false,
    breaklinks=true, pdfborder={0 0 0}, backref=false, colorlinks=true,
    linkcolor=red,          
    citecolor=blue,        
    filecolor=magenta,      
    urlcolor=magenta           
}

\usepackage{ulem}

\usepackage{etoolbox} 

\makeatletter
\newwrite\apxtocfile

\newcommand{\startappendixtoc}{%
  \immediate\openout\apxtocfile=\jobname.apx.toc
  \let\orig@addcontentsline\addcontentsline
  \def\addcontentsline##1##2##3{%
    \orig@addcontentsline{##1}{##2}{##3}%
    \protected@write\apxtocfile{}{\string\contentsline{##1}{##2}{##3}}%
  }%
}
\newtheorem{theorem}{Theorem}
\newcommand{\stopappendixtoc}{%
  \immediate\closeout\apxtocfile
  \let\addcontentsline\orig@addcontentsline
}

\newcommand{\printappendixtoc}{%
  \section*{Appendix contents}%
  \addcontentsline{toc}{section}{Appendix contents}
  \begingroup
    \input{\jobname.apx.toc}%
  \endgroup
}
\makeatother

\def\emp{\mathrm{emp}}
\def\occ{\mathrm{occ}}
\newcommand{\ovl}{\overline}
\newcommand{\MC}{\mathscr}

\newcommand{\mbbm}{\mathbbm}

\newcommand{\defeq}{\vcentcolon=}

\newcommand{\dfdx}[2]{\frac{d #1}{d #2}}

\newcommand{\pfpx}[2]{\frac{\partial #1}{\partial #2}}

\newcommand{\ppfpx}[2]{\frac{\partial_0 #1}{\partial_0 #2}}
\newcommand{\encvert}[2]{\left. #1 \right\vert_{#2}}
\newcommand{\Expval}[2]{\langle #1 \rangle_{#2}}
\newcommand{\Bra}[1]{\langle #1 |}
\newcommand{\Ket}[1]{| #1 \rangle}

\newcommand{\yjwNmsl}[1]{\textcolor{purple}{#1}}

\begin{document}

\preprint{APS/123-QED}

\title{Strongly Correlated Superconductivity in Twisted Bilayer Graphene: a Gutzwiller Study}

\author{Matthew Shu Liang}
\thanks{These authors contributed equally to this work.}
\affiliation{Department of Physics, Hong Kong University of Science and Technology, Clear Water Bay, Hong Kong, China}

\author{Yi-Jie Wang}
\thanks{These authors contributed equally to this work.}
\affiliation{International Center for Quantum Materials, School of Physics, Peking University, Beijing 100871, China}
\affiliation{Beijing Key Laboratory of Quantum Devices, Peking University, Beijing 100871, China}

\author{Geng-Dong Zhou}
\affiliation{International Center for Quantum Materials, School of Physics, Peking University, Beijing 100871, China}
\affiliation{Beijing Key Laboratory of Quantum Devices, Peking University, Beijing 100871, China}

\author{Zhi-Da Song}
\email{songzd@pku.edu.cn}
\affiliation{International Center for Quantum Materials, School of Physics, Peking University, Beijing 100871, China}
\affiliation{Beijing Key Laboratory of Quantum Devices, Peking University, Beijing 100871, China}
\affiliation{Hefei National Laboratory, Hefei 230088, China}

\author{Xi Dai}
\email{daix@ust.hk}
\affiliation{Department of Physics, Hong Kong University of Science and Technology, Clear Water Bay, Hong Kong, China}
\affiliation{New Cornerstone Science Laboratory, Department of Physics, Hong Kong University of Science and Technology, Clear Water Bay, Hong Kong, China}

\date{\today}

\begin{abstract}
We study strongly correlated superconductivity in magic-angle twisted bilayer graphene (MATBG) using a variational Gutzwiller wavefunction $\ket{\Psi_G} = \prod_{\vb{R}} \hat{P}_{\vb{R}} \ket{\Phi_0}$, where the Gutzwiller projector $\hat{P}_{\vb{R}}$ is allowed to break charge U(1) symmetry to accommodate superconducting (SC) order. The ground state energy is evaluated via the \textit{Gutzwiller Approximation} applied to an 8-band model consisting of correlated $f$-orbitals and uncorrelated $c$-orbitals, with interactions including onsite Coulomb repulsion $U$, phonon-mediated anti-Hund's coupling $\hat{H}_{J_A}$, and intra-orbital Hund's coupling $\hat{H}_{J_H}$. 
At filling $\nu = 2.5$, we map out the phase diagram as a function of $U$ and $J_A$, 
and reveal a strongly correlated SC (SC-SC) phase dominates at large $U$, wherethe strong on-site interaction U strongly suppress the $f$-orbital charge fluctuations while maintaining finite pairing order and a sizeable quasiparticle weight Z, distinguishing it from a conventional Mott insulator. 
For a range of $J_{\rm A}$, SC-SC transitions to FL as $U$ decreases, until the weakly correlated BCS-like SC (BCS-SC) re-enters as $U \to 0$.
We further identify a novel small Fermi liquid (sFL) state with effective Fermi surface formed by $c$-orbitals, which is essentially different with the normal Fermi liquid. 
Interestingly, in the intermediate- ($U \lesssim 40$ meV) and large-$U$ ($U \gtrsim 40$ meV) regimes, the conventional FL and the sFL are the lowest-energy normal phases, respectively, potentially serve as the parent states of the SC-SC phase.
These results illuminate the interplay between strong correlations and unconventional pairing in MATBG, and establish a versatile Gutzwiller framework applicable to other strongly correlated superconductors.
\end{abstract}

\maketitle 
\emph{Introduction.}
The groundbreaking discovery of superconductivity (SC) in magic-angle twisted bilayer graphene (MATBG) \cite{Cao2018_nature, Oh2021_nature, Cao2021_science, Lin2022_science, Liu2021_science, Stepanov2020_nature, Yankowitz2019_science, Cao2020_PRL, Kerelsky2019_nature, Lu2019_nature, Tanaka2025_nature, Liu2022_nature, Yuan2021_nature, Arora2020_nature, Saito2020_nature, Nuckolls2023_nature} has established a new frontier in condensed matter physics. The phenomenology of MATBG exhibits striking parallels with high-$T_c$ cuprates—most notably the emergence of the SC phase upon doping a correlated insulator and the linear-$T$ resistivity \cite{Cao2020_PRL, polshyn_large_2019, jaoui_quantum_2022} of the normal state—pointing toward a strongly correlated pairing mechanism. 
Although early studies suggested that screening the Coulomb interaction enhances SC \cite{Stepanov2020_nature, Liu2021_science}, a recent experiment showed the opposite: By an {\it in situ} tuning of dielectric constant of the SrTiO$_3$ substrate, $T_c$ of the SC in MATBG monotonously decreases with as $U$ is reduced \cite{gao_double-edged_2024}.
Furthermore, the V-shape \cite{Oh2021_nature} and ``two-gap'' structures \cite{park_experimental_2025,kim_resolving_2026} in tunneling spectra, nematicity \cite{Cao2021_science} of Coulomb-driven origin \cite{NP_Zhang2026}, small coherence length \cite{Cao2018_nature,Lu2019_nature} imply an unconventional pairing mechanism.
These features have inspired a plethora of theoretical proposals \cite{YJWANG2024_PRL, wang_spin-valley_2026, Wu2018_PRL, BLian2019_PRL, FCWu2019_PRB, Sharma2020_PRR, zhao2025RVBsFL, Isobe2018_PRX, Chichinadze2020_PRB, Kennes2018_PRB, Gonzalez2019_PRL, You2019_npj, Huang2022_PRL, Khalaf2021_science, LEDWITH2021168646}, ranging from electron-phonon coupling \cite{Wu2018_PRL, BLian2019_PRL, FCWu2019_PRB}, collective-mode-mediated pairing \cite{Sharma2020_PRR} and instability from Fermi surface nesting \cite{Isobe2018_PRX, Chichinadze2020_PRB}, Kohn-Luttinger mechanisms \cite{Gonzalez2019_PRL,You2019_npj} to anti-Hund's coupling \cite{angeli_valley_2019,YJWang2025_PRB,HShi2025_PRB} driven molecular pairing \cite{YJWANG2024_PRL} or resonating-valence-bond theory \cite{zhao2025RVBsFL}, {\it etc.}  
Despite the intense theoretical effort, a quantitative self-consistent study of the SC state within a realistic lattice model remains a critical missing piece. 

In this Letter, we address this problem by applying a variational Gutzwiller theory \cite{PENG2022108348, DENG2009, Nicola2012, PhysRevB.76.165110, Gutzwiller1963_PRL, Metzner1987_PRL, Metzner1988_PRB, Metzner1989_PRL, Buneman1998_PRB} to a heavy-fermion–like \cite{song_magic-angle_2022} model containing two strongly correlated $f$-orbitals (per spin valley) and extended uncorrelated $c$-orbitals  \cite{haule_mott-semiconducting_2019, Bascones2020_PRB}.
Such models have been widely adopted to describe the normal-state properties of MATBG \cite{chou_kondo_2023, hu_symmetric_2023, Bascones2023_NC, zhou_kondo_2024, rai_dynamical_2024, PhysRevX.15.021028}. 
Accumulating experimental evidence has supported the heavy-fermion picture in which $f$- and $c$-electrons coexist \cite{saito_isospin_2021, rozen_entropic_2021, merino_interplay_2025, ghosh_thermopower_2025, QYHu2025_NC, xiao_interacting_2025}, as well as its relevance to the SC phase \cite{kim_resolving_2026}.
In addition, a recent angle-resolved photoemission spectroscopy study \cite{chen_strong_2024} revealed strong coupling between the inter-valley phonon [\onlinecite{liu_electron_2024}, \onlinecite{HShi2025_PRB}, \onlinecite{YJWang2025_PRB}] and flat band electrons.
As discussed in our previous work \cite{HShi2025_PRB}, a particular moir\'e phonon mode derived from the inter-valley in-plane transversal optical (iTO) branches of graphene couples strongly to the $f$-electrons. 
Since the energy of this mode ($\sim$150 meV) is much larger than the bandwidth of the moir\'e bands, the induced interaction can be regarded as a non-retarded inter-valley anti-Hund’s coupling among the $f$-electrons within the same moir\'e unit cell \cite{YJWANG2024_PRL}.

This anti-Hund’s coupling reveals a close analogy between MATBG and alkali-metal-doped fullerene superconductors \yjwNmsl{\cite{A3C60_RMP, Yue2023_PRB}}, where a similar mechanism arises from localized intramolecular phonons. 
A crucial difference, however, lies in the nontrivial band topology of MATBG originating from the hybridization between localized $f$-orbitals and extended $c$-orbitals, which has no counterpart in fullerene systems. 
Using the variational Gutzwiller method, we determine the superconducting phase diagram of this model and obtain three main results.
First, we identify a strongly correlated superconducting (SC-SC) state that occupies a large region of the phase diagram and is qualitatively distinct from BCS superconductivity (BCS-SC) in weakly correlated systems. 
Second, the competition between anti-Hund’s and conventional Hund’s couplings in the $f$-orbitals produces a mixed $s$+$d$ pairing state that breaks rotational symmetry, giving rise to a nematic superconducting phase. 
Third, we find that at large $U$ ($\gtrsim40$meV), a novel small Fermi liquid (sFL) replaces the conventional FL as the normal phase energetically approximate to the SC-SC ground state, whereas the conventional FL remains closer at intermediate $U$.


\emph{Model.}
To study strongly correlated superconductivity, we adopt an eight-band tight-binding model [\onlinecite{Po2019_PRB}, \onlinecite{Bascones2020_PRB}, \onlinecite{HShi2025_PRB}] (marked as $\hat{H}_0$) which encodes key features of correlation physics in MATBG such as the resets of compressibility with electron filling \cite{Bascones2023_NC, QYHu2025_NC}. 
The model consists of two correlated $f$-orbitals and six uncorrelated $c$-orbitals in each spin valley. 
$f$-electrons are denoted by $\hat{f}_{\mathbf{R};\alpha \eta s}$, where $\mathbf{R}$ labels moir\'e unit cells, $s=\uparrow\downarrow$ denotes spins, $\eta=\pm$ denotes valleys, and $(-1)^{\alpha-1} \eta$ (for $\alpha=1,2$) denotes the orbital angular momentum modulo 3.
The $f$-orbitals constitute most of the flat bands except for a small region near mori\'e $\Gamma$ point, making it faithful for describing low energy strongly correlated physics. 
We include three terms in our interaction Hamiltonian: 
The onsite Coulomb repulsion $\hat{H}_U$ leads to strong correlation; 
An anti-Hund's coupling $\hat{H}_{J_A}$ \cite{YJWANG2024_PRL} (mediated by coupling of flat band electrons and mori\'e optical phonons \cite{HShi2025_PRB, YJWang2025_PRB}) provides attractive interaction for inter-valley spin-singlet pairing channels; 
An intra-orbital Hund's couling $\hat{H}_{J_H}$ lets $d$-wave pairing energetically favorable. 
To ensure a reasonable filling of $f$-orbitals at $\nu=\pm(2+\delta \nu)$ at the presence of large $U$, phenomenological Hartree terms between $f$- and $c$-electrons ($\hat{H}_W$) and among $c$-electrons ($\hat{H}_V$) are added to our Hamiltonian to adjust the relative energy levels of $f$- and $c$-orbitals \cite{QYHu2025_NC, PhysRevX.15.021028, Bascones2023_NC}. 
The full Hamiltonian reads (for details, refer to Appendix \ref{appdx:TBG_implementation}): 
\begin{gather}
    \hat{H} = \hat{H}_0 + \hat{H}_U + \hat{H}_{J_A} + \hat{H}_{J_H} + \hat{H}_{W} + \hat{H}_{V} - \mu \hat{N} \ . 
\end{gather}
where $\mu$ is the chemical potential, and $\hat{N}$ is the total electron number. 
We tune $\mu$ to enforce filling $\nu=2.5$ in this work, which is around the optimal doping for SC in experiments \cite{Cao2018_nature, Oh2021_nature, Cao2021_science}. Since the main focus of this work is to investigate the competition between $J_A$ and $U$, we fix $J_H=1.5\mathrm{meV}$ in this work.

If $U=0$, $\hat{H}_{J_A} + \hat{H}_{J_H}$ hosts bare attractive channels. 
For $\frac{2}{3} J_H < J_A < 2J_H$, the leading attractive channel is two-fold degenerate, $\hat{\Delta}_{d,\alpha} = \hat{f}_{\alpha + \uparrow} \hat{f}_{\bar{\alpha} - \downarrow} - (\uparrow \leftrightarrow \downarrow)$ with $\alpha=1,2$; and for $J_A > 2J_H$, non-degenerate, $\hat{\Delta}_s = \sum_{\alpha} \hat{f}_{\alpha+\uparrow} \hat{f}_{\alpha - \downarrow} - (\uparrow\leftrightarrow\downarrow)$. 
$\hat{\Delta}_{d,1} \pm \hat{\Delta}_{d,2}$ transform as $d_{x^2-y^2}$ and $d_{xy}$ ``orbitals'' under the $D_6$ group, respectively, while $\hat{\Delta}_s$ forms the trivial representation. 
We dub pairing channels with the corresponding symmetry characters as $d$-wave and $s$-wave, respectively. 

\emph{Gutzwiller Theory.}
The variational ground state wavefunction is constructed by applying the Gutzwiller projector on $f$-orbitals $\hat{P}_{\vb{R}}=\sum_{I I'} \Lambda_{I I'} \ket{\vb{R}; I} \bra{\vb{R}; I'}$ to a Bardeen-Cooper-Schrieffer (BCS) wavefunction $\ket{\Phi_0}$: $\ket{\Psi_G} = \prod_{\vb{R}} \hat{P}_{\vb{R}} \ket{\Phi_0}$, where only intra-site correlation is considered. 
Here, $\mathbf{R}$ represents the sites of $f$-orbitals, $I$ and $I'$ denote the $2^8$ local many-body states, and 
$\Lambda_{I I'}$ are the many-body variational parameters to be optimized to minimize the ground state energy.
A crucial advance of our approach is that $\Lambda_{I I'}$ are allowed to break charge-U(1) symmetry, thereby enabling a description of SC states within the Gutzwiller framework \cite{PhysRevB.76.165110}.
For the model studied in this work, we classify $\Lambda_{II'}$ into irreducible representation (irrep) blocks according to both discrete symmetries ($\mathcal{T}, C_{2z}$, and $C_{2x}$) and continuous symmetries (spin-SU(2) and valley-U(1)), leaving 513 independent real parameters. 
In particular, $C_{2z}\mathcal{T}$ symmetry will enforce the condensed $d$-wave order parameter to take the form of $\langle e^{-i2\theta} \hat{\Delta}_{d,1} + e^{i 2\theta} \hat{\Delta}_{d,2} \rangle_G$, where $\theta$ parametrizes the spatial orientation. Such order spontaneously breaks $C_{3z}$. 
As the total energy anisotropy associated with $\theta$ is expected to be small, we thus utilize $C_{2x}$ to pin $\theta=0$ (up to $C_{3z}$ rotations), in order to further reduce the number of variational parameters.
The actual nematic direction in experiments is likely to be selected by extrinsic effects, such as strain in the sample. We restrict our analysis to singlet pairings in the present study, as $J_A$ energetically favors spin-singlets states. All phases we study in this work with different symmetries are summarized in Table \ref{table:phases}.

The ground state energy is obtained by minimizing the expectation value $\langle \Psi_G | \hat{H} | \Psi_G \rangle$ (abbreviated as $\langle \hat{H} \rangle_G$) with respect to $\Lambda_{II'}$ and $\ket{\Phi_0}$ under the \textit{Gutzwiller Approximation} \cite{Gutzwiller1963_PRL}, which becomes exact in the limit of infinite coordination number \cite{Metzner1987_PRL, Metzner1988_PRB, Metzner1989_PRL}. 
To handle the large set of variational parameters, we adopt the variational strategy from Refs.~\cite{PhysRevB.76.165110,Nicola2012, PENG2022108348, QYHu2025_NC}.
In this approach, the onset uncorrelated (Nambu) reduced density matrix 
\begin{equation}
    \bm{\varrho}^0 = \begin{pmatrix}
        \bm{\rho}^0 & \bm{\Delta}^0 \\
        (\bm{\Delta}^0)^\dagger & \mathbbm{1} - (\bm{\rho}^0)^\intercal
    \end{pmatrix},\  
    \begin{matrix} \bm{\rho}^0_{\alpha \eta, \beta\eta} = \langle \hat{f}^\dagger_{\vb{R}; \beta \eta \uparrow } \hat{f}_{\vb{R}; \alpha \eta \uparrow} \rangle_0 \\ 
    \bm{\Delta}^0_{\alpha \eta, \beta\bar{\eta}} = \langle \hat{f}_{\vb{R}; \beta \bar{\eta}\downarrow} \hat{f}_{\vb{R}; \alpha\eta\uparrow} \rangle_0
    \end{matrix} \label{def:rho0_ff}
\end{equation} 
is treated as an independent set of variational parameters in addition to $\Lambda_{II'}$ and $\ket{\Phi_0}$ where we have assumed translational invariance.
A valid parametrization of $\bm{\varrho}^0$ must have eigenvalues restricted in $[0,1]$, and we present two such parametrizations in [Appendix.~\ref{appndx:prmtrz_rho0}]. 
The grand energy functional reads
\begin{gather}
    \mathcal{L}[\mu, \bm{\varrho}^0, \ket{\Phi_0}, \bm{\Lambda}, \bm{\lambda}^{F}, \bm{\lambda}^B] = \langle \hat{H} \rangle_G + \sum_{l} \lambda^F_{l} \!\cdot\! g^F_{l}\left(\ket{\Phi_0}, \bm{\varrho}^0 \right) \nonumber \\
    + \sum_{l} \lambda^B_{l} \!\cdot\! g^B_{l}\left( \bm{\Lambda}, \bm{\varrho}^0 \right) - \mu  \langle \hat{N} \rangle_G \ .
\end{gather}
Here, $\bm{\lambda}^{F(B)} \cdot \bm{g}^{F(B)}$ are Lagrange multipliers for Gutzwiller constraints \cite{PENG2022108348, Nicola2012}.
The ground state for a given chemical potential $\mu$ is found by solving:
\begin{gather}
    \pfpx{}{x} \mathcal{L} \Big|_{\mu} = 0,\quad x \in (\bm{\varrho}^0, \ket{\Phi_0}, \bm{\Lambda}, \bm{\lambda}^{F}, \bm{\lambda}^B)
\end{gather}


Once the ground state is found, we can obtain the quasi-particle excitations as $|\Psi^{h(p)}_{\vb{k}, \alpha \eta s}\rangle = \hat{q}^{(\dagger)}_{\vb{k}, \alpha \eta s} \hat{P}_G |\Phi_0\rangle$ where $\hat{q}^{(\dagger)}_{\vb{k}, \alpha \eta s} = \hat{P}_G \hat{f}^{(\dagger)}_{\vb{k}, \alpha \eta s} \hat{P}^{-1}_G$, and the quasi-particle spectrum equals to the spectrum of effective Hamiltonian $\hat{H}^F$ that solves $\ket{\Phi_0}$: $\hat{H}^F\ket{\Phi_0} = E^F \ket{\Phi_0}$\cite{PhysRevB.76.165110, Bünemann2003_PRB}.
The normal and anomalous quasi-particle renormalization factors $\MC{R}$ and $\MC{Q}$ then arise when projecting operators from physical to effective space: ($\bar{\eta}$ and $\bar{s}$ denote the opposite valley and spin indices)
\begin{gather}
    \hat{P}^\dagger_{\vb{R}} \hat{f}^\dagger_{\vb{R} \alpha \eta s} \hat{P}_{\vb{R}} \rightarrow \sum_{ \beta} \hat{f}^\dagger_{\vb{R} \beta \eta s} \MC{R}_{\beta \alpha}  + (-1)^{\delta_{\bar{s} \uparrow}} \hat{f}_{\vb{R} \beta \bar{\eta} \bar{s}}  \MC{Q}_{\beta\alpha }
\end{gather}

\begin{table}[h]
\centering
\begin{tabular}{|c|c|c|c|c|c|c|c|}
\hline
\multicolumn{2}{|c|}{Phases} & $n_d$ & $\Delta_s$ & $\Delta_d$ & $C_{3z}$ & $\hat{P}_{\vb{R}}$ & $\ket{\Phi_0}$ \\
\hline
 \multirow{3}{*}{SC} & $s$-wave & 0 & $\mathbb{R}$ & 0 & yes & \multirow{3}{*}{$[\hat{P}_{\vb{R}}, \hat{N}_{\vb{R}}] \neq 0$} & \multirow{3}{*}{BCS} \\
\cline{2-6}
 & $d$-wave & $\mathbb{R}$ & 0 & $\mathbb{R}$ & no & & \\
 \cline{2-6}
 & $s$+$d$-wave & $\mathbb{R}$ & $\mathbb{R}$ & $\mathbb{R}$ & no & &  \\
\hline
 \multirow{2}{*}{FL} & symmetric & 0 & 0 & 0 & yes & \multirow{2}{*}{$[\hat{P}_{\vb{R}}, \hat{N}_{\vb{R}}] = 0$} & \multirow{2}{*}{SD} \\
 \cline{2-6}
 & nematic & $\mathbb{R}$ & 0 & 0 &no & & \\
\hline
 \multirow{2}{*}{sFL} & symmetric & 0 & 0 & 0 & yes & \multirow{2}{*}{$[\hat{P}_{\vb{R}}, \hat{N}_{\vb{R}}] = 2 \hat{P}_{\vb{R}} $} & \multirow{2}{*}{SD} \\
 \cline{2-6}
 & nematic & $\mathbb{R}$ & 0 & 0 & no & & \\
 \hline
\end{tabular}
\caption{
    Signatures for all the phases studied in this work, including local order parameters, $C_{3z}$ symmetry, and different ansatz for $\hat{P}_{\vb{R}}$ and $\ket{\Phi_0}$. 
    $n_d = \frac{1}{4} \sum_{\alpha\eta s} \langle f^\dagger_{\mathbf{R};1 \eta s} f_{\mathbf{R}; 2 \eta s} \rangle_G$, $\Delta_s = \frac{1}{4} \langle \hat{\Delta}_{s} \rangle_G$, and $\Delta_d = \frac{1}{4} \langle \hat{\Delta}_{d,1} + \hat{\Delta}_{d,2} \rangle_G$
    We have fixed the gauge of global charge U(1) breaking and utilized the $C_{2z}T$ symmetry, so that $n_d$, $\Delta_s$ and $\Delta_d$ are real [Appendix~.\ref{appdx:TBG_implementation}]. SD stands for Slater Determinant. 
}
\label{table:phases}
\end{table}

One crucial difference from normal Gutzwiller state is that our ansatz for $\ket{\Phi_0}$ and $\hat{P}_{\vb{R}}$ for SC phase permits a gauge degrees of freedom for the $f$-orbitals besides the charge-U(1) phase factor: $\ket{\Psi_G} = \prod_{\vb{R}} \hat{P}_{\vb{R}} \hat{\mathcal{U}}^\dagger \hat{\mathcal{U}} \ket{\Phi_0}$, with $\hat{\mathcal{U}}$ being a unitary operator mixing $f_{\alpha \eta s}$ with $f_{\alpha' \bar\eta \bar s}^\dagger$.
Since $\hat{\mathcal{U}}$ commutes with all the symmetries except charge-U(1), it can absorb the local pairing components of $\ket{\Phi_0}$ and equivalently encode them into the local operator $\hat{P}_{\mathbf R}$.  We prove in the [Appendix.~\ref{appdx:gauge}] that one can always find a $\hat{\mathcal{U}}_{\Delta^0=0}$ such that $\bm{\Delta}^0_{\vb{R}}=0$ for $\bm{\varrho}^0_{\vb{R}}$ defined in Eq.~\ref{def:rho0_ff}, which is a general statement of the natural basis defined in \cite{PhysRevB.76.165110}. 
Under this gauge, the anomalous renormalization factor $\MC{Q}$ becomes non-zero only when the system turns into SC phases, while stays strictly zero for normal phases. 
The overall quasiparticle weight $Z_{\alpha \beta} = (\MC{R}^\dagger \MC{R})_{\alpha \beta} + (\MC{Q}^\dagger \MC{Q})_{\alpha \beta}$ [Appendix.~\ref{appdx:gauge}] is however a $\hat{\mathcal{U}}$-independent quantity and shall still serves as a measure of hybridization strength between the $f$ and $c$-orbitals just as in normal states.

\begin{figure}[t]
\centering
\includegraphics[width=0.48\textwidth]{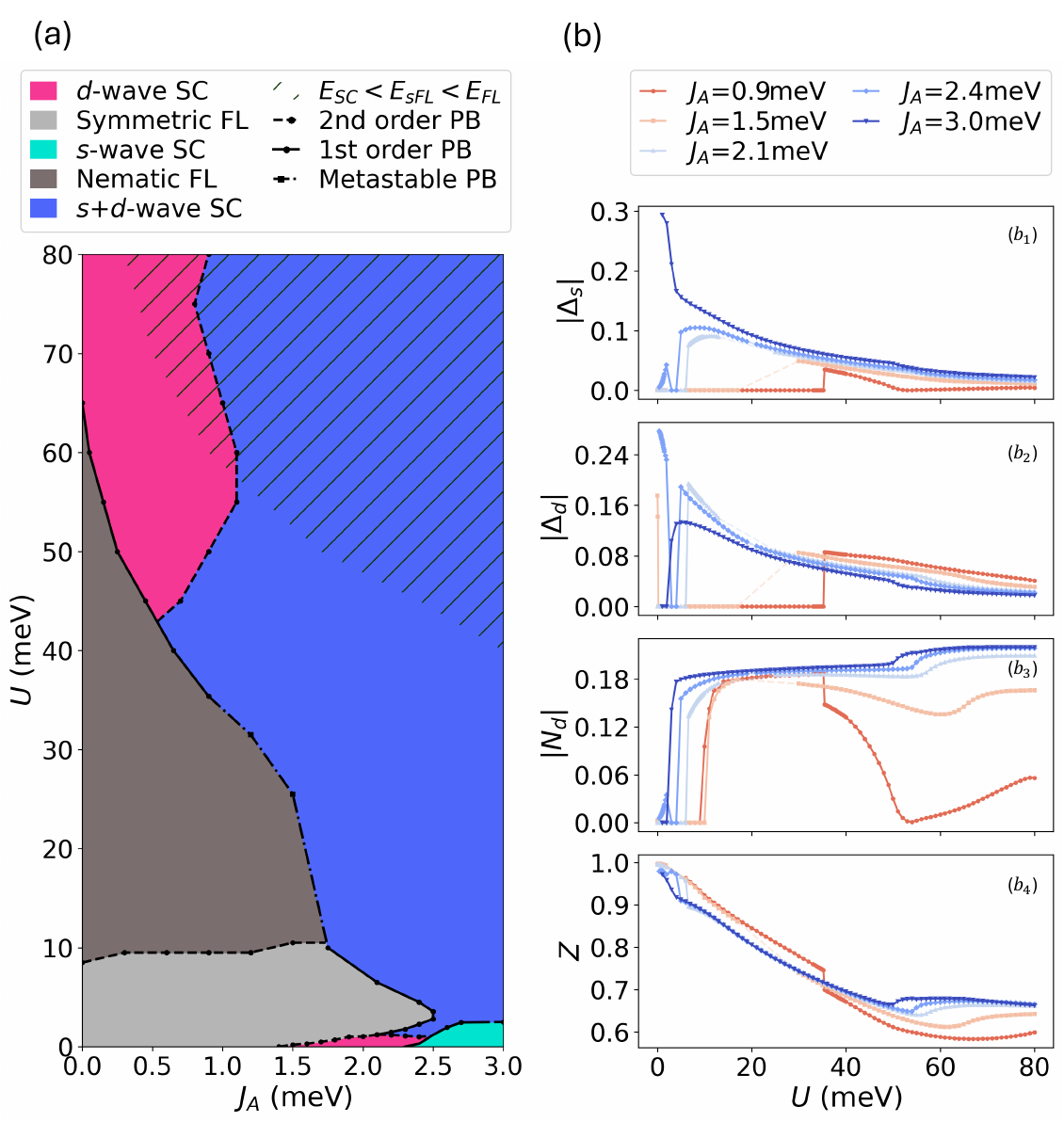}
\caption{(a) Phase diagram varying $U$ and $J_A$ at $\nu=2.5$ and $J_H=1.5\, \mathrm{meV}$. The different phases and phase boundaries (PB) are dubbed by different colors defined in the legend above. The phase boundaries are marked by black dots connected by guide lines where the solid line and dash line marks first and second order phase transitions respectively \footnote{We identify a second order transition from $d$-wave to $s$+$d$-wave when $\abs{\Delta_s} / \abs{\Delta_d} > 0.05$}. The dash-dot line marks metastable PB \footnote{The metastable PB denotes the transition from nematic FL to $d$-wave SC because the $s$+$d$-wave does not converge in the nearby region.}. The Fermi liquid regime dubbed by gray and dark gray colors separates the BCS and SCS phases. The shaded region marks the metastable normal state is sFL instead of conventional FL ($E_{SC}<E_{sFL}< E_{FL}$). (b) Plots of absolute values of intervalley local pairing order parameters $\abs{\Delta_s}$ and $\abs{\Delta_d}$, intravalley local interorbital order parameter $\abs{N_d}$ and quasi-particle weight $Z$ versus $U$ at $J_A=0.9,1.5,2.1,2.4,3.0 \, \mathrm{meV}$. The dented regions in ($\text{b}_1$) and ($\text{b}_2$) where both $\abs{\Delta_s}$ and $\abs{\Delta_d}$ are zero marks the FL phase which shrinks as $J_A$ increases. The discontinuities in ($\text{b}_4$) indciates first order phase transitions between SC and FL phases where FL tends to have larger $Z$ than SC states at same $J_A$.}
\label{fig:phase_diagram}
\end{figure}

\emph{Results}: [Fig.~\ref{fig:phase_diagram}] shows the phase diagram at filling $\nu=2.5$ as a function of $U$ and $J_A$ for fixed $J_H=1.5\,\mathrm{meV}$, highlighting the competition between onsite repulsion and phonon-mediated attraction. 
Owing to an approximate particle--hole symmetry of $\hat{H}$ \cite{song_magic-angle_2022, YJWang2025_PRB}, the results at $\nu=-2.5$ should differ only quantitatively. 
We observe a tilted dome-shaped FL phase with a tip at $(U, J_A)\approx(5\mathrm{meV}, 2.5\mathrm{meV})$ cuts through the SC phase, separating it into a weakly correlated BCS-like SC (BCS-SC) at small $U$ and a strongly correlated SC (SC-SC) at large $U$, similar to the Coulomb interaction driven superconductivity studied in \cite{Capone2004_PRL}. 
At $U=0$, a BCS-like $d$-wave SC state characterized by $\Delta_d=\frac{1}{4}\langle \hat{\Delta}_{d,1} + \hat{\Delta}_{d,2} \rangle_G$ develops out of the symmetric FL for $1.4\,\mathrm{meV} \lesssim J_{A} \lesssim 2.3\,\mathrm{meV}$, aided by $\hat{H}_{J_H}$ which favors inter-orbital pairing. 
In this regime, the condensation energy scales as $\sim J_A \Delta_d^2$ and is small for $J_A<2.0\,\mathrm{meV}$, so a modest $U$ drives a second-order transition back to the symmetric FL. 
Increasing $J_A$ at small $U$ raises the critical interaction $U_c$, and the transition becomes weakly first order for $J_A\geq 2.0\,\mathrm{meV}$. 
As the coupling strength further increases to $J_A > 2.4\,\mathrm{meV}$ at $U=0$, the pairing order with lowest energy shifts to $s$-wave from $d$-wave, where the variational wavefunction differs markedly from the metastable symmetric FL at the same $(U,J_A)$: the $f$-orbital occupation distribution $\mathcal{P}(N_f)$ Fig.~\ref{fig:Occ_Atomic_prob}($a_{31}$) indicates enhanced pairing, in accord with the large $\Delta_s$ for $J_A=3.0\,\mathrm{meV}$ in Fig.~\ref{fig:phase_diagram}(b), while the spin-valley-orbital fluctuations are strongly quenched as shown by Fig.~\ref{fig:Occ_Atomic_prob}($b_{31}$), leaving $\Ket{\psi_s}$ as the only configuration with significant weight. Its strong-coupling character also enables a smooth crossover into the SC-SC regime at large $U$ without going into the FL phase [Fig.~\ref{fig:phase_diagram}($a$)].

\begin{figure}[t]
\centering
\includegraphics[width=0.48\textwidth]{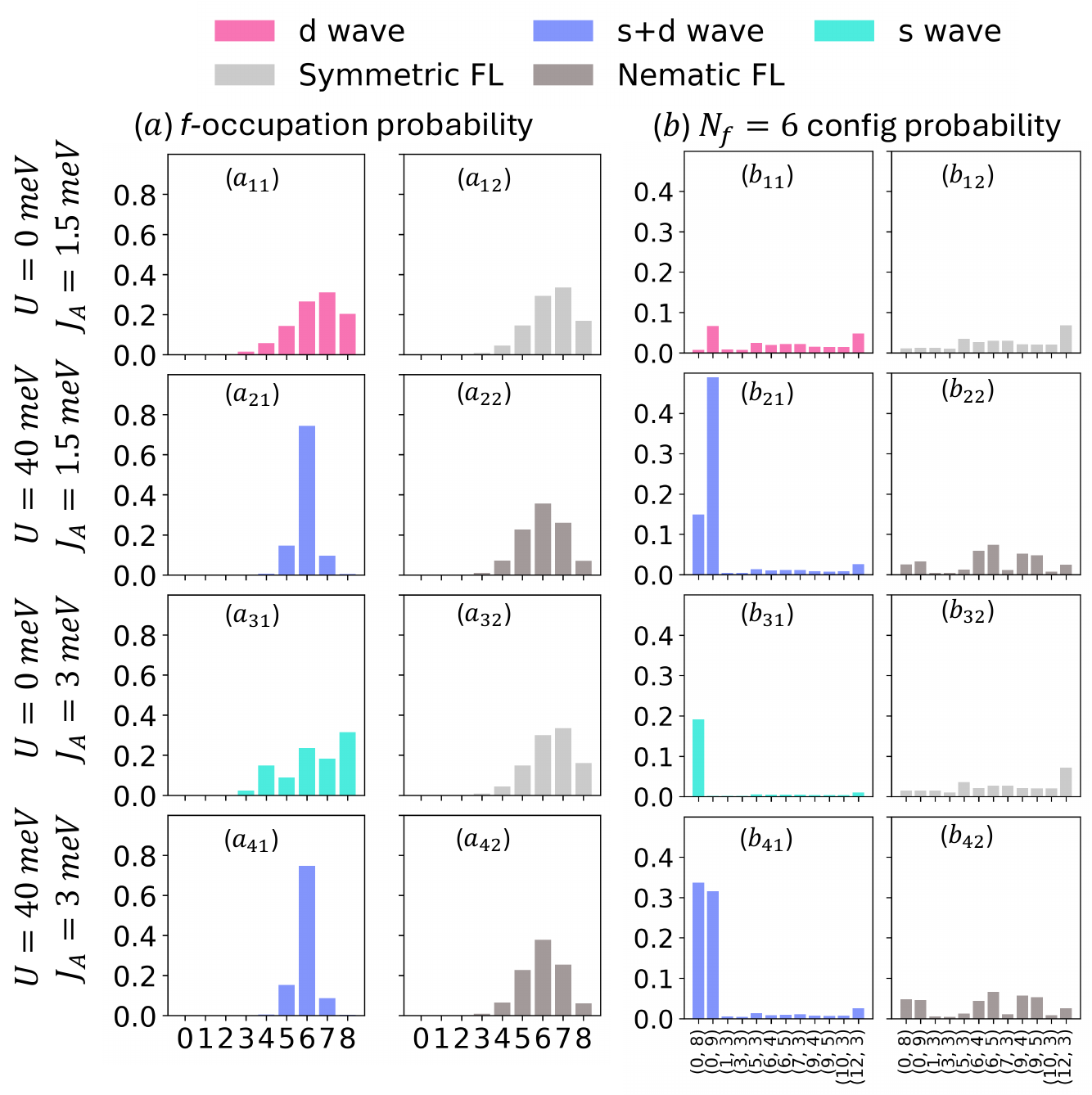}
\caption{(a) Occupation probability for $f$-orbital $\mathcal{P}(N_f)$ and (b) atomic configuration probability $\mathcal{P}_{N_f}(\Gamma)$ \footnote{$\Gamma$ are the irrep states and $\mathcal{P}_{N_f}(\Gamma)$ is multiplied by the dimension of the irrep.}
at $N_f=6$ sector at different sets of $(U, J_A)$ for each row. The left and right columns in (a) and (b) correspond to stable SC and meta-stable FL phases respectively. The x-axis in (a) labels $f$-orbital charge number $N_f$. The x-axis labels in (b): $(r, d)$ represents the $d$-th state in the $r$-th irreducible representation block defined in [Appendix.~\ref{appdx:TBG_implementation}]. $\Gamma_{(0, 8)}\equiv\ket{\psi_s}$ and $\Gamma_{(0,9)}\equiv\ket{\psi_d}$ are the $s$-wave and $d$-wave wavefunctions respectively.
}
\label{fig:Occ_Atomic_prob}
\end{figure}

Strong-correlation effects become prominent once $U$ exceeds the energy scale set by $J_A$, where a simple mean-field decoupling of the $f$-orbitals interactions fails to stabilize any superconducting solutions. Such a state is qualitatively distinct from a BCS mean-field state, in which Cooper pairs are formed in momentum space. Here, by contrast, electron pairs are formed locally in real space among the spin, orbital, and valley degrees of freedom within a single moir\'e unit cell. These ``preformed'' pairs acquire partial itinerancy through the residual charge fluctuations induced by doping away from the integer filling $\nu=2$, thereby driving the system into a superconducting phase. This mechanism can be captured by a charge-U(1)-breaking Gutzwiller wavefunction, with projector
$
\hat{P}^{\text{SC-SC}}_G \simeq (\Lambda^{\text{SC-SC}}_{6_s8} \ket{\psi_s} + \Lambda^{\text{SC-SC}}_{6_d8} \ket{\psi_d}) \bra{8}
+ \sum_{|I|=|I'|} \Lambda^{\text{SC-SC}}_{I I'} \Ket{I} \Bra{I'},
$ 
and an uncorrelated state $\Ket{\Phi^{\text{SC-SC}}_0}$ with nearly filled $f$ orbitals,
$
7 < \Bra{\Phi^{\text{SC-SC}}_0} \hat{N}_{\vb{R}, f} \Ket{\Phi^{\text{SC-SC}}_0} \lesssim 8.
$
The ``preformed'' pairs originate from the first term in $\hat{P}^{\text{SC-SC}}_G$, while their mobility is enabled by both the second term in $\hat{P}^{\text{SC-SC}}_G$ and the hybridization with $c$ electrons, as reflected in the large quasiparticle weight $Z$ shown in [Fig.~\ref{fig:phase_diagram}($b_4$)]. The SC-SC state is stabilized relative to the FL state by its smaller Coulomb energy $E_U$, achieved through a strong suppression of $f$-orbital charge fluctuations. This is evident from the sharply peaked occupation distribution $\mathcal{P}(N_f)$ at $N_f=6$ for $U=40\,\mathrm{meV}$ [Fig.~\ref{fig:Occ_Atomic_prob}($a_{21}$) and ($a_{41}$)], compared with the broader distribution in the FL state [Fig.~\ref{fig:Occ_Atomic_prob}($a_{22}$) and ($a_{42}$)]. The main difference between the cases $J_A = 1.5\,\mathrm{meV}$ and $J_A = 3\,\mathrm{meV}$ at $U=40\,\mathrm{meV}$ lies in the relative configuration weights of $\Ket{\psi_s}$ and $\Ket{\psi_d}$ [Fig.~\ref{fig:Occ_Atomic_prob}($b_{21}$) and ($b_{41}$)]. The large separation of scales between $U$ and $J_A$ therefore suggests that charge fluctuations are governed primarily by $U$, whereas the spin-valley-orbital structure of the condensate is controlled mainly by $J_A$, consistent with the mechanism proposed in Refs.~\cite{YJWANG2024_PRL,Capone2004_PRL}, where the effective $\tilde{U}$ is substantially renormalized downward while $\tilde{J}_A$ remains largely intact.

The $s$+$d$-wave state obtained here can host nodal points with appropriate interaction parameters[Fig.~\ref{fig:node_dos}(a)] due to the interplay between the $s$ and $d$ wave order parameters, when they are comparable as illustrated in [Fig.~\ref{fig:phase_diagram}(b)]. 
While $C_{2z}\mathcal{T}$ symmetry stabilizes each individual nodal point \cite{YJWANG2024_PRL}, the two nodal points with opposite winding numbers can annihilate each other \cite{ahn_failure_2019}, thereby opening a gap. 
Nonetheless, the gap structure in the momentum space remains highly anisotropic [Fig.~\ref{fig:node_dos}(a)], and the BdG quasi-particle density of states [Fig.~\ref{fig:node_dos}(b)] remains sharply distinct from the conventional $s$-wave superconductors.

\begin{figure}[t]
\centering
\includegraphics[width=0.48\textwidth]{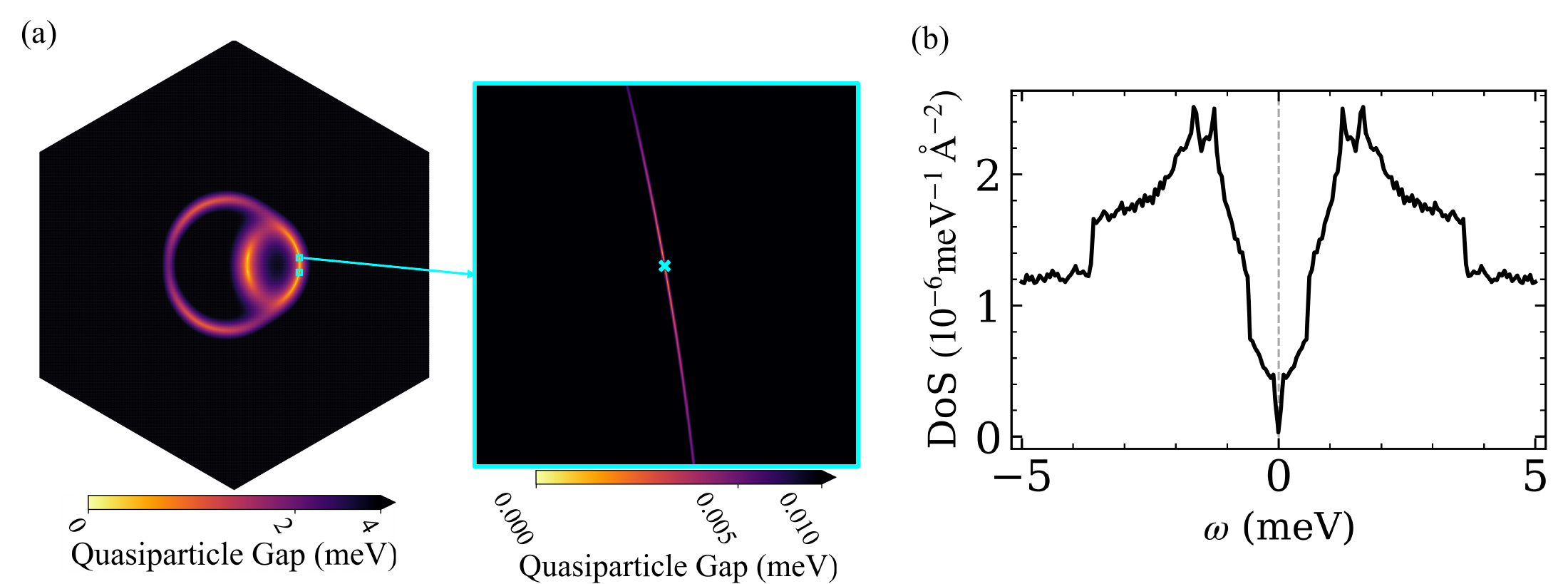}
\caption{Nodal $s+d$-wave at $U=80\mathrm{meV}$, $J_A=2.8\mathrm{meV}$: (a) Anisotropic quasiparticle gap structure near Fermi surface where a pair of nodes are found, as shown in the zoom-in gap plot with higher energy resolution; (b) Quasiparticle density of state plotted with a Lorentzian spread of $0.001\mathrm{meV}$, which has a V-shape structure at the vicinity of $\omega=0$.}
\label{fig:node_dos}
\end{figure}

Regardless of nematicity, the conventional FL remains the unique normal ground state across the phase diagram when SC does not develop (Fig.~\ref{fig:phase_diagram}).
However, we identify a metastable small Fermi liquid (sFL) phase that also becomes energetically proximate to the SC-SC state at large $U$. 
Unlike the conventional FL where $[\hat{P}_{\vb{R}}, \hat{N}_{\vb{R}}] = 0$, this novel metallic state is characterized by $[\hat{P}_{\vb{R}}, \hat{N}_{\vb{R}}] = 2 \hat{P}_{\vb{R}}$. 
As a consequence, the wavefunction $\ket{\Psi^{sFL}_G} = \hat{P}^{sFL}_G \ket{\Phi^{sFL}_0}$ is still an eigenstate of $\hat{N}$, but it has 2 less electrons per moir\'e unit cell than $|\Phi^{sFL}_0\rangle$ before the Gutzwiller projection[Appendix~.\ref{appdx:FS_sFL}]. Since the Gutzwiller projector $\hat{P}^{sFL}_G$ is not invertible by having rank $<2^8$, the usual Landau-Gutzwiller quasi-particle $\hat{q}^{(\dagger)}_{\vb{k}, \alpha \eta s} = \hat{P}_G \hat{f}^{(\dagger)}_{\vb{k}, \alpha \eta s} \hat{P}^{-1}_G$ \cite{Bünemann2003_PRB} cannot be defined which suggests sFL is not a heavy Fermi liquid and we shall not view the spectrum of $\hat{H}^F_{sFL}$ as the quasi-particle band. Instead, it is just an auxiliary Hamiltonian taking us to the non-correlated wave function before the projection.
For doping $\nu=2.5$ where we have $\nu_f \lesssim 2$ ($N_f \lesssim 6$) for the sFL state, $\hat{P}^{sFL}_G$ consists mostly of $(\Lambda_{6_s8} \ket{\psi_s} + \Lambda_{6_d8} \ket{\psi_d}) \bra{8}$ and the effective $f$-orbitals in $\ket{\Phi^{sFL}_0}$ are nearly fully occupied similar to SC-SC state at $U \gtrsim 60 \mathrm{meV}$. As a consequence, most of the mobile charge carriers are $c$-electrons which forms an effective small Fermi surface as shown in [Fig.~\ref{fig:Ecomponent}(c)].
The remaining $\lesssim 6$ electrons in the f-orbitals after the projection form a typical local RVB states \cite{Anderson1987_science,PWAnderson_2004}, which are completely localized and have no contribution to the formation of Fermi surface.
These signatures suggest that the sFL state can lower $\hat{H}_U$ and $\hat{H}_{J_A}$ much more effectively than FL state, with a cost of losing kinetic energy due to highly suppressed valence fluctuations of $f$-orbitals, which are confirmed by our numerical results in [Fig.~\ref{fig:Ecomponent}(b)]. The sFL becomes nematic if it has mixing components of $\ket{\psi_s}$ and $\ket{\psi_d}$ resulting finite $n_d$, which happens for $J_A \geq 2.3\, \mathrm{meV}$ at $U=60\, \mathrm{meV}$ [Appendix.~\ref{appdx:Uvar_JAvar}].

\begin{figure}[t]
\centering
\includegraphics[width=0.48\textwidth]{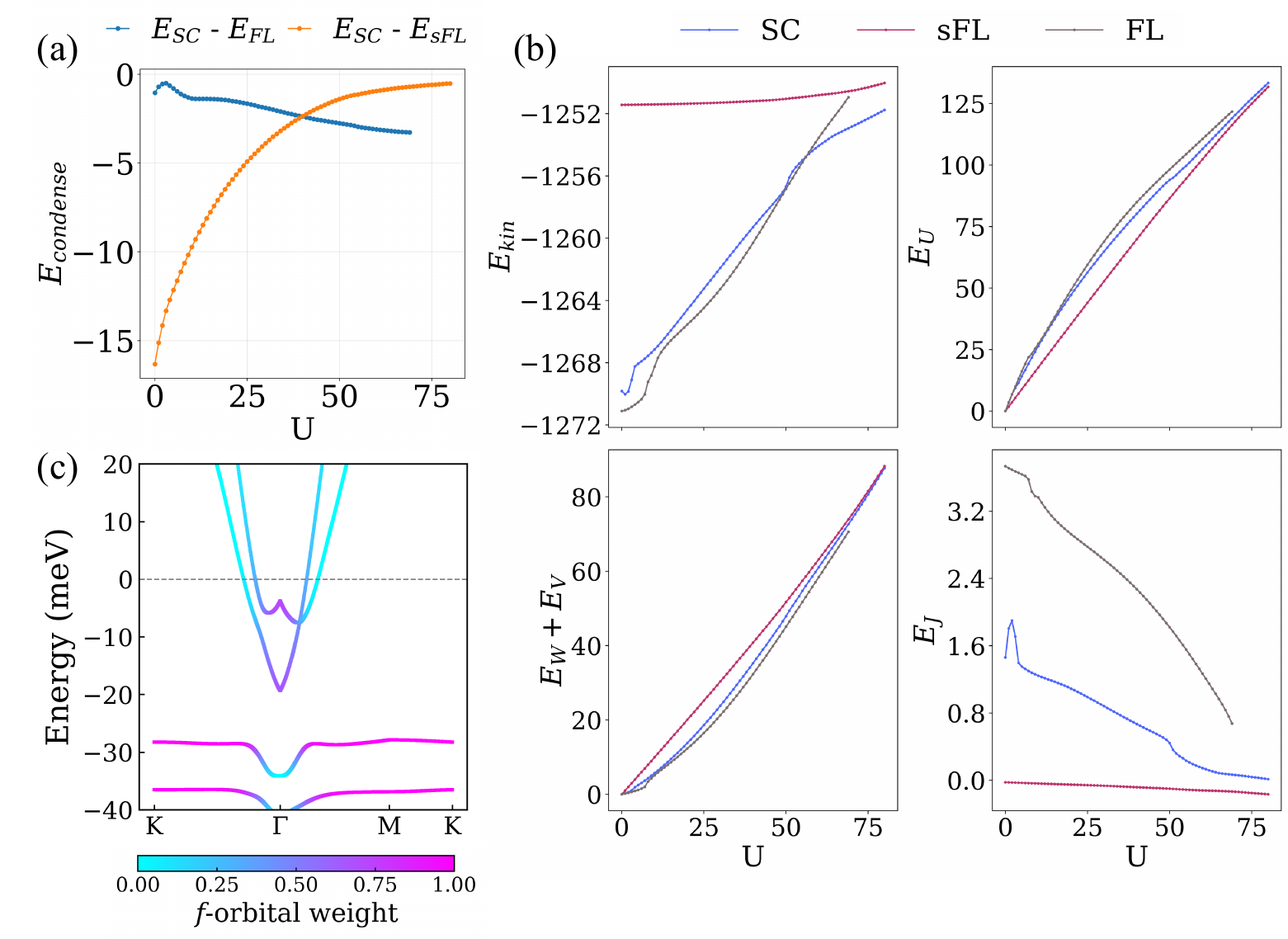}
\caption{(a) and (b) are Energy comparisons at $J_A=3\mathrm{meV}$ and different $U$, where all units are in meV: (a) Condensation energy of the SC state against the FL and sFL states; (b) Kinetic energy $E_{kin}$, Coulomb energy among $f$-orbitals $E_U$, Hartree energy between $f$- and $c$-orbitals plus Hartree energy among $c$-orbitals $E_V + E_W$ and Hund's coupling plus anti-Hund's coupling energy $E_{J_H} + E_{J_A}$ for SC, FL and sFL states respectively; (c) Effective single-particle band for nematic sFL state at $U=60\, \mathrm{meV}$ and $J_A = 3 \mathrm{meV}$, where the gray line shows the Fermi level.}
\label{fig:Ecomponent}
\end{figure}

We further compare the energies of the metastable FL and sFL phases within the SC-SC regime of the phase diagram [Fig.~\ref{fig:Ecomponent}(a)]. 
As $U$ increases, the condensation energy relative to the sFL state, \(E_{SC} - E_{sFL}\), diminishes in magnitude, while the energy relative to the FL state, \(E_{SC} - E_{FL}\), grows. 
Around $U\approx40$meV, an energy inversion occurs, after which the sFL becomes more proximate to the SC-SC than the conventional FL.
This energetic inversion may suggest that the sFL state, rather than the conventional FL, acts as the parent state of the SC-SC phase in the large $U$ limit. 
However, this identification remains subtle. 
As the SC-SC order melts upon heating, the system selects the normal phase with the lower free energy. The conventional FL, characterized by heavy flat bands near the Fermi level, contributes significant entropy that provides an entropic advantage over the sFL \cite{Saito_nature_2021}.

\emph{Conclusion}: 
We developed a variational Gutzwiller framework with charge-U(1)-breaking local projectors and applied it to an eight-band model for MATBG. The resulting phase diagram reveals a weak-coupling BCS-SC regime, a strongly correlated SC-SC regime stabilized at large $U$, and an intervening Fermi-liquid dome. In the SC-SC regime, superconductivity wins over FL state by superior suppression on $f$-orbital charge fluctuations, often with nematic $s$+$d$ character, and is accompanied by the emergence of a small Fermi liquid that becomes energetically closer to the superconducting state than the conventional Fermi liquid. These results support a strong-correlation route to pairing in MATBG and establish the present Gutzwiller approach as a useful framework for moiré superconductors with intertwined localization and pairing.

\emph{Acknowledgment}:
We thank Hao Shi for discussions about 8-band model of MATBG. We also thank Shiyu Peng and Zhanghuan Li for discussions about Gutzwiller theory. M. S. Liang and X. Dai are supported by the New Cornerstone Foundation and a fellowship award from Hong Kong Research Grant Council (Project No. SRFS2324-6S01). Z.-
D. S., Y.-J. Wang, and G.-D. Zhou. are supported by Na-
tional Natural Science Foundation of China (General Pro-
gram No. 12274005), National Key Research and Develop-
ment Program of China (No. 2021YFA1401900), and Quan-
tum Science and Technology-National Science and Technol-
ogy Major Project (No. 2021ZD0302403).

\bibliography{TBG_GASC.bib}
\onecolumngrid

\clearpage
\tableofcontents
\clearpage

\appendix

\section{Variational Gutzwiller Approximation Method for Superconductivity}
Here we derive our GA formalism for superconducting problem with original basis. While `natural basis' is commonly used in the literature for simplicity, we find that analytical derivatives are more conveniently derived in the original basis as we bypass the difficulty of taking derivatives of unitary transformation in the natural basis. We assume spin-SU(2) and time-reversal symmetry in our derivation since we are only interested in singlet pairing case for TBG. Such symmetry assumptions make our derivations compact and simple, while they can be easily generalized.

We start with a Periodic Anderson Model where $\{ \hat{f}_{s \alpha}, \hat{f}_{s' \beta}, \dots \}$ denotes the correlated electron operator and $\{ \hat{c}_{s a}, \hat{c}_{s' b}, \dots \}$ denotes the itinerant electron operator. $s, s', \dots$ are spin indices,  $\alpha, \beta, \dots$ and $a, b, \dots$ are orbital indices for $f$- and $c-$orbitals respectively. $\vb{R}$ is the lattice site index, $\vb{k}$ is the reciprocal lattice vector within the Brillouin zone.
\begin{gather}
    \hat{H}^0 = \sum_{\vb{R}_i \vb{R}_j} \sum_{s\alpha, s' \beta} t_{\vb{R}_i \vb{R}_j;\alpha \beta} \hat{f}^\dagger_{\vb{R}_i s \alpha} \hat{f}_{\vb{R}_j s' \beta} + \frac{1}{\sqrt{N}} \sum_{\vb{k}} \sum_{\vb{R}} \sum_{s \alpha, s' a} \left(e^{i \vb{k} \cdot \vb{R}} V_{\vb{k}; \alpha a} \hat{f}^\dagger_{\vb{R} s \alpha} \hat{c}_{\vb{k} s' a} + h.c.\right) + \sum_{\vb{k}} \sum_{s a} \epsilon^c_{\vb{k} a} \hat{c}^\dagger_{\vb{k} s a} \hat{c}_{\vb{k} s a} \label{H0:PAM}
\end{gather}
We consider local interactions amongst $f$-orbital which can take the form of any Hermitian 4-fermionic operator:
\begin{gather}
    \hat{H}^{\text{int}} = \sum_{\vb{R}}  \sum_{\substack{s_1 s_2 s_3 s_4 \\ \alpha \beta \gamma \delta}} U_{\alpha \beta \gamma \delta} \left( \hat{f}^\dagger_{\vb{R} s_1 \alpha} \hat{f}^\dagger_{\vb{R} s_2 \beta} \hat{f}_{\vb{R} s_3 \gamma} \hat{f}_{\vb{R} s_4 \delta} + h.c. \right)
\end{gather}
Our goal is to evaluate:
\begin{gather}
    E_G = \bra{\Psi_G} \hat{H}^0 + \hat{H}^{\text{int}} \ket{\Psi_G} / \bra{\Psi_G} \ket{\Psi_G}
\end{gather}
We first evaluate the kinetic energy $E_{\text{kin}}$. It is convenient to separate the onsite contribution of the $f$-orbitals from $\hat{H}^0$:
\begin{gather}
    \hat{H}^0_{\text{onsite}} = \sum_{\vb{R}} \sum_{\alpha \beta} \sum_{s} t_{\vb{R} \vb{R}; \alpha \beta} \hat{f}^\dagger_{\vb{R} s \alpha} \hat{f}_{\vb{R} s \beta} \\
    E_{\text{kin}} = \bra{\Psi_G} \hat{H}^0 - \hat{H}^0_{\text{onsite}} \ket{\Psi_G} / \bra{\Psi_G} \ket{\Psi_G} \label{def:Ekin}
\end{gather}
The main difficulty of evaluating $E_{\text{kin}}$ arises from the following terms:
\begin{gather}
    \bra{\Psi_G} \hat{f}^\dagger_{\vb{R}_i s \alpha} \hat{f}_{\vb{R}_j s \beta} \ket{\Psi_G} / \bra{\Psi_G} \ket{\Psi_G} \ ,\ \vb{R}_i \neq \vb{R}_j \label{expval:GffG} \\
    \bra{\Psi_G} \hat{f}^\dagger_{\vb{R} s \alpha} \hat{c}_{\vb{k} s a} \ket{\Psi_G} / \bra{\Psi_G} \ket{\Psi_G}
\end{gather}
The numerator in \eqref{expval:GffG} can be written as:
\begin{gather} 
    \bra{\Psi_G} \hat{f}^\dagger_{\vb{R}_i s \alpha} \hat{f}_{\vb{R}_j s \beta} \ket{\Psi_G} = \bra{\Phi_0} \left( \prod_{\vb{R} \neq \vb{R}_i, \vb{R}_j} \hat{P}^\dagger_{\vb{R}} \hat{P}_{\vb{R}} \right) \hat{P}^\dagger_{\vb{R}_i} \hat{f}^\dagger_{\vb{R}_i s \alpha} \hat{P}_{\vb{R}_i} \hat{P}^\dagger_{\vb{R}_j} \hat{f}_{\vb{R}_j s \beta} \hat{P}_{\vb{R}_j} \ket{\Phi_0} \label{expval:numerator_GffG}
\end{gather}
Use Wick's theorem, the above expectation value can be expanded to sums of products of contractions. The contraction $\bra{\Phi_0}\hat{f}^{(\dagger)}_{\vb{R}_1 s \alpha}  \hat{f}^{(\dagger)}_{\vb{R}_2 s' \beta}\ket{\Phi_0}$ is called a line connecting two sites $\vb{R}_1$ and $\vb{R}_2$ if $\vb{R}_1 \neq \vb{R}_2$, or is called a self-HF(Hartree Fock) Bubble if $\vb{R}_1 = \vb{R}_2$. The Gutzwiller approximation states that any products of contractions containing 2 sites connected by 3 lines(paths) or more are neglected. To cancel the self-HF bubbles, the \textbf{Gutzwiller Constraints} are introduced:
\begin{equation}\label{eq:gutzwiller_constraints}
\begin{aligned}
\bra{\Phi_0} \hat{P}_{\vb{R}}^\dagger \hat{P}_{\vb{R}} \hat{f}^\dagger_{\vb{R}s \alpha} \hat{f}_{\vb{R}s' \beta} \ket{\Phi_0}
    &= \bra{\Phi_0} \hat{f}^\dagger_{\vb{R} s \alpha} \hat{f}_{\vb{R} s\beta} \ket{\Phi_0}, \\
\bra{\Phi_0} \hat{P}_{\vb{R}}^\dagger \hat{P}_{\vb{R}} \hat{f}^\dagger_{\vb{R} s \alpha} \hat{f}^\dagger_{\vb{R} \bar{s} \beta} \ket{\Phi_0}
    &= \bra{\Phi_0} \hat{f}^\dagger_{\vb{R} s \alpha} \hat{f}^\dagger_{\vb{R} \bar{s} \beta} \ket{\Phi_0}, \\
\bra{\Phi_0} \hat{P}_{\vb{R}}^\dagger \hat{P}_{\vb{R}} \hat{f}_{\vb{R} s \alpha} \hat{f}_{\vb{R} \bar{s} \beta} \ket{\Phi_0}
    &= \bra{\Phi_0} \hat{f}_{\vb{R} s \alpha} \hat{f}_{\vb{R} \bar{s} \beta} \ket{\Phi_0},
\end{aligned}
\end{equation}
such that any 2 sites connected by 2 lines will also vanish because such cases must contain self-HF Bubbles.
Bestowed by GA and Gutzwiller Constraints(G.C.), we may proceed to calculate \eqref{expval:numerator_GffG}:
\begin{align}
    \bra{\Psi_G} \hat{f}^\dagger_{\vb{R}_i s \alpha} \hat{f}_{\vb{R}_j s \beta} \ket{\Psi_G} \overset{\text{GA \& G.C.}}{\approx}&  \left[\prod_{\vb{R} \neq \vb{R}_i, \vb{R}_j} \bra{\Phi_0} \hat{P}^\dagger_{\vb{R}} \hat{P}_{\vb{R}} \ket{\Phi_0} \right]\bra{\Phi_0} \hat{P}^\dagger_{\vb{R}_i} \hat{f}^\dagger_{\vb{R}_i s \alpha} \hat{P}_{\vb{R}_i} \hat{P}^\dagger_{\vb{R}_j} \hat{f}_{\vb{R}_j s \beta} \hat{P}_{\vb{R}_j} \ket{\Phi_0} \\
    \overset{\text{G.C.}}{=}& \bra{\Phi_0} \hat{P}^\dagger_{\vb{R}_i} \hat{f}^\dagger_{\vb{R}_i s \alpha} \hat{P}_{\vb{R}_i} \hat{P}^\dagger_{\vb{R}_j} \hat{f}_{\vb{R}_j s \beta} \hat{P}_{\vb{R}_j} \ket{\Phi_0},  \label{GA:EkinBlocks}
\end{align}
And $\ket{\Psi_G}$ is normalised to 1:
\begin{gather}
    \bra{\Psi_G} \ket{\Psi_G} \overset{\text{GA \& G.C.}}{\approx} \prod_{\vb{R}} \bra{\Phi_0} \hat{P}^\dagger_{\vb{R}} \hat{P}_{\vb{R}} \ket{\Phi_0} \overset{\text{G.C.}}{=} 1 \label{GA:normalization}
\end{gather}
The onsite energy consists of calculation of local observable $\hat{\mathcal{O}}_{\vb{R}}$:
\begin{align}
    \bra{\Psi_G} \hat{\mathcal{O}}_{\vb{R}} \ket{\Psi_G} \overset{\text{GA \& G.C.}}{\approx}&\left[\prod_{\vb{R}' \neq \vb{R}} \bra{\Phi_0} \hat{P}^\dagger_{\vb{R}'} \hat{P}_{\vb{R}'} \ket{\Phi_0} \right] \bra{\Phi_0} \hat{P}^\dagger_{\vb{R}} \hat{\mathcal{O}}_{\vb{R}} \hat{P}_{\vb{R}} \ket{\Phi_0} \\
    \overset{\text{G.C.}}{=}& \bra{\Phi_0} \hat{P}^\dagger_{\vb{R}} \hat{\mathcal{O}}_{\vb{R}} \hat{P}_{\vb{R}} \ket{\Phi_0} \label{eq:onsite_observable}
\end{align}
What remains is to compute the approximate expectation values, which can be done relative easily using either the natural-basis or mixed-basis representation of $\hat{P}_{\vb{R}}$. Here, however, we take a brief detour and work in the original basis. This choice is motivated by the need to derive analytical expressions for the derivatives of the energy functional, as will become clear in later sections B. We abbreviate the expectation value $ \langle \dots \rangle_G \defeq \bra{\Psi_G} \dots \ket{\Psi_G}$ and $ \langle \dots \rangle_0 \defeq \bra{\Phi_0} \dots \ket{\Phi_0}$.\\

\subsection{Energy Functional $\mathcal{L}$}
As discussed in the main text, additional Lagrange multipliers are introduced to $\mathcal{L}$ to enable $\bm{\varrho}^0$ as an extra degrees of variational freedom. The full energy functional (averaged per unit cell) reads:
\begin{gather}
    \mathcal{L}[\mu, \bm{\varrho}^0, \ket{\Phi_0}, \bm{\Lambda}, \bm{\lambda}^{F(B)}] = \langle \hat{H} \rangle_G + \sum_{l} \lambda^F_{l} g^F_{l}(\ket{\Phi_0}, \bm{\varrho}^0) + \sum_{l} \lambda^B_{l} g^B_{l}(\bm{\Lambda}, \bm{\varrho}^0) - \mu \sum_{l} \langle \hat{n}_{\vb{R}, l} \rangle_G, \label{eq:energy_functional}
\end{gather}
where $\langle \hat{H} \rangle_G = E_{\text{kin}} + E_{\text{atom}}$. We shall first calculate $E_{\text{kin}}$. The building blocks involving $f$-orbital for $E_{\text{kin}}$ are of the form: ($ i \neq j$)
\begin{gather}
    \bra{\Psi_G} \hat{f}^\dagger_{\vb{R}_i s \alpha} \hat{f}_{\vb{R}_j s' \beta} \ket{\Psi_G}, \quad \bra{\Psi_G} \hat{f}^\dagger_{\vb{R}_i s \alpha} \hat{f}^\dagger_{\vb{R}_j s' \beta} \ket{\Psi_G}, \quad \bra{\Psi_G} \hat{f}^\dagger_{\vb{R}_i s \alpha} \hat{c}_{\vb{k} s' a} \ket{\Psi_G} , \quad \bra{\Psi_G} \hat{f}_{\vb{R}_i s \alpha} \hat{c}_{\vb{k} s' a} \ket{\Psi_G} \label{Ekin:building_blocks}
\end{gather}
and their Hermitian conjugates. Under GA, they can be evaluated using \eqref{GA:EkinBlocks} and a convenient substitution can be made \cite{PhysRevB.76.165110, Nicola2012}:
\begin{gather}
    \hat{P}^\dagger_{\vb{R}_i} \hat{f}^\dagger_{\vb{R}_i s \alpha } \hat{P}_{\vb{R}_i} \rightarrow \sum_{s' \beta} \hat{f}^\dagger_{\vb{R}_i s' \beta} \MC{R}_{s' \beta; s \alpha}  + \hat{f}_{\vb{R}_i s' \beta}  \MC{Q}_{s' \beta ;s \alpha } 
\end{gather}
$\MC{R}, \MC{Q}$ are called the \textbf{Renormalisation Factors} for Bogliubov-Landau-Gutzwiller quasiparticle. With spin SU(2) symmetry on $\hat{P}_i$, the above substitution reduces to:
\begin{gather}
    \hat{P}^\dagger_{\vb{R}_{i}} \hat{f}^\dagger_{\vb{R}_{i} s \alpha } \hat{P}_{\vb{R}_{i}} \rightarrow \sum_{s' \beta} \delta_{s s'}\hat{f}^\dagger_{\vb{R}_{i} s \beta} \MC{R}_{s \beta; s \alpha}  + \delta_{s' \bar{s}} \hat{f}_{\vb{R}_{i} \bar{s} \beta}  \MC{Q}_{\bar{s} \beta ;s \alpha } \label{eq:renrml_substitution}
\end{gather}
The rotation along spin y-axis gives: $\hat{i \sigma_y}\hat{f}^\dagger_{\uparrow}\hat{i \sigma}^\dagger_y = \hat{f}^\dagger_{\uparrow}$, $\hat{i \sigma_y}\hat{f}^\dagger_{\downarrow}\hat{i \sigma}^\dagger_y = -\hat{f}^\dagger_{\downarrow}$. Apply it to both sides of the substitution above, we have:
\begin{gather}
    \MC{R} \defeq \MC{R}_{\uparrow \uparrow} = \MC{R}_{\downarrow \downarrow}, \quad \MC{Q} \defeq \MC{Q}_{\downarrow \uparrow} = - \MC{Q}_{\uparrow \downarrow},
\end{gather}
Now we only need to solve \eqref{eq:renrml_substitution} for spin up case by evaluating the following expectation values:
\begin{align}
    \mathcal{K}_{\uparrow \uparrow; \beta \alpha} \defeq \bra{\Phi_0} \hat{P}^\dagger_{\vb{R}_{i}} \hat{f}^\dagger_{\vb{R}_{i} \uparrow \alpha } \hat{P}_{\vb{R}_{i}} \hat{f}_{\vb{R}_{i} \uparrow \beta} \ket{\Phi_0} =& \sum_{\gamma} \langle \hat{f}^\dagger_{\vb{R}_{i} \uparrow \gamma} \hat{f}_{\vb{R}_{i} \uparrow \beta} \rangle_0 \MC{R}_{\gamma \alpha} + \langle \hat{f}_{\vb{R}_{i} \downarrow \gamma} \hat{f}_{\vb{R}_{i} \uparrow \beta} \rangle_0 \MC{Q}_{\gamma \alpha} \\
    \mathcal{K}_{\downarrow \uparrow; \beta \alpha} \defeq \bra{\Phi_0} \hat{P}^\dagger_{\vb{R}_{i}} \hat{f}^\dagger_{\vb{R}_{i} \uparrow \alpha } \hat{P}_{\vb{R}_{i}} \hat{f}^\dagger_{\vb{R}_{i} \downarrow \beta} \ket{\Phi_0} =& \sum_{\gamma} \langle \hat{f}^\dagger_{\vb{R}_{i} \uparrow \gamma} \hat{f}^\dagger_{\vb{R}_{i} \downarrow \beta} \rangle_0 \MC{R}_{\gamma \alpha} + \langle \hat{f}_{\vb{R}_{i} \downarrow \gamma} \hat{f}^\dagger_{\vb{R}_{i} \downarrow \beta} \rangle_0 \MC{Q}_{\gamma \alpha}
\end{align}
where the expectation values $\langle \rangle_0$ at the RHS are just the variational Nambu reduced local density matrix elements defined in \eqref{def:rho0_ff}. Thus, $\MC{R}$ and $\MC{Q}$ are solved by the following matrix equation:
\begin{gather}
    \begin{pmatrix}
        \MC{R} \\ \MC{Q}
    \end{pmatrix} = 
    (\bm{\varrho}^0)^{-1} 
    \begin{pmatrix}
        \mathcal{K}_{\uparrow \uparrow} \\ \mathcal{K}_{\downarrow \uparrow}
    \end{pmatrix},  \quad
    \begin{pmatrix}
        \MC{R}^\dagger & \MC{Q}^\dagger
    \end{pmatrix} = 
    \begin{pmatrix}
        \mathcal{K}^\dagger_{\uparrow \uparrow} & \mathcal{K}^\dagger_{\downarrow \uparrow}
    \end{pmatrix} (\bm{\varrho}^0)^{-1}
     \label{eq:RQ_solution}
\end{gather}
The inverse on $\bm{\varrho}^0$ means its eigenvalues must be non-zero, and spin-SU(2) symmetry guarantees that the spectrum of $\bm{\varrho}^0$ comes in pairs $\{d_n, 1-d_n\}$, thus the eigenvalues must be strictly within $(0,1)$. The legal parametrizations of $\bm{\varrho}^0$ are discussed in appendix?. The $\mathcal{K}$ matrices are calculated by expanding $\hat{P}_{\vb{R}_i}$ (site index $\vb{R}_i$ omitted hereafter):
\begin{align}
    \mathcal{K}_{\uparrow \uparrow; \beta \alpha} =& \bra{\Phi_0} \hat{P}^\dagger \hat{f}^\dagger_{\uparrow \alpha} \hat{P} \hat{f}_{\uparrow \beta} \ket{\Phi_0} \nonumber\\
    =& \sum_{I_1 \dots I_4} (\Lambda^\dagger)_{I_2 I_1} \bra{\Phi_0} \ket{I_2} \bra{I_1} \hat{f}^\dagger_{\uparrow \alpha} \ket{I_3} \bra{I_4} \hat{f}_{\uparrow \beta} \ket{\Phi_0} \Lambda_{I_3 I_4} \nonumber \\
    =& \sum_{I_1 \dots I_5}  (\Lambda^\dagger)_{I_2 I_1} \bra{\Phi_0} \ket{I_2} \bra{I_1} \hat{f}^\dagger_{\uparrow \alpha} \ket{I_3} \bra{I_4} \hat{f}_{\uparrow \beta}  \ket{I_5} \bra{I_5}  \ket{\Phi_0} \Lambda_{I_3 I_4} \nonumber \\
    =& \Tr[\bm{\Lambda}^\dagger \bm{S}^\dagger_{\uparrow \alpha} \bm{\Lambda} \bm{S}_{\uparrow \beta} \bm{m}^0], \quad (\bm{S}^{(\dagger)}_{\uparrow \alpha})_{I I'} \defeq \bra{I} \hat{f}^{(\dagger)}_{\uparrow \alpha} \ket{I'} \\[1em]
    \mathcal{K}^\dagger_{\uparrow \uparrow; \beta \alpha} =& \Tr[\bm{m}^0 \bm{S}^\dagger_{\uparrow \alpha} \bm{\Lambda}^\dagger \bm{S}_{\uparrow \beta} \bm{\Lambda}] \\[1em]
    \mathcal{K}_{\downarrow \uparrow; \beta \alpha} =& \Tr[\bm{\Lambda}^\dagger \bm{S}^\dagger_{\uparrow \alpha} \bm{\Lambda} \bm{S}^\dagger_{\downarrow \beta} \bm{m}^0] \\[1em]
    \mathcal{K}^\dagger_{\downarrow \uparrow; \beta \alpha} =& \Tr[\bm{m}^0 \bm{S}_{\downarrow \alpha} \bm{\Lambda}^\dagger \bm{S}_{\uparrow \beta} \bm{\Lambda}]
\end{align}
where we define the uncorrelated reduced many-body density matrix $\bm{m}^0$:
\begin{gather}
    \bm{m}^0_{I I'} \defeq \bra{\Phi_0} \ket{I'} \bra{I} \ket{\Phi_0} = (-1)^{\abs{\overline{(I \cup I')}}}\det(\mathcal{M}^{I' I}_{\rho}) \label{eq:m0_cal}, 
\end{gather}
where (we abbreviate $\sigma$ as a combined spin-orbital index $s \alpha$, $\sigma'_i \in I'$ and $\sigma_j \in I$)
\begin{gather}
\mathcal{M}^{I' I}_{\varrho^0} = \begin{pmatrix}
\mathcal{P}^{I' I}_{\varrho^0} & \mathcal{P}_{\varrho^0}^{I' \overline{(I \cup I')}} \\
\mathcal{P}_{\varrho^0}^{\overline{(I \cup I')} I} & \mathcal{P}_{\varrho^0}^{\overline{(I \cup I')} \overline{(I \cup I')}}  - \mbbm{1}
\end{pmatrix}
, \quad 
\mathcal{P}^{I' I}_{\varrho^0} = \begin{pmatrix}
\bm{\varrho}^0_{\bm{\sigma'}_0 \bm{\sigma}_0} & \bm{\varrho}^0_{\bm{\sigma'}_0 \bm{\sigma}_1} & \dots & \bm{\varrho}^0_{\bm{\sigma'}_0 \bm{\sigma}_{\abs{I}}} \\
\bm{\varrho}^0_{\bm{\sigma'}_1 \bm{\sigma}_0} & \bm{\varrho}^0_{\bm{\sigma'}_1 \bm{\sigma}_1} & \dots & \vdots \\
\vdots & \vdots & \ddots & \vdots \\
\bm{\varrho}^0_{\bm{\sigma'}_{\abs{I'}} \bm{\sigma}_0} & \dots & \dots & \bm{\varrho}^0_{\bm{\sigma'}_{\abs{I'}} \bm{\sigma}_{\abs{I}}}
\end{pmatrix} \label{def:M_and_P}
\end{gather}

The building blocks for $E_{\text{kin}}$ \eqref{Ekin:building_blocks} can now be explicitly expressed as functions of $\ket{\Phi_0}$, $\bm{\Lambda}$ and $\bm{\varrho}^0$:
\begin{equation}
\begin{aligned}
    \bra{\Psi_G} \hat{f}^\dagger_{\vb{R}_i \uparrow \alpha} \hat{f}_{\vb{R}_j \uparrow \beta} \ket{\Psi_G} =& 
        \sum_{\alpha' \beta'} \begin{aligned}[t] & \MC{R}^\dagger_{\beta \beta'} \Expval{\hat{f}^\dagger_{\vb{R}_i \uparrow \alpha'} \hat{f}_{\vb{R}_j \uparrow \beta'}}{0} \MC{R}_{\alpha' \alpha} + \MC{Q}^\dagger_{\beta \beta'} \Expval{\hat{f}^\dagger_{\vb{R}_i \uparrow \alpha'} \hat{f}^\dagger_{\vb{R}_j \downarrow \beta'}}{0} \MC{R}_{\alpha' \alpha} \\
        &+ \MC{R}^\dagger_{\beta \beta'} \Expval{\hat{f}_{\vb{R}_i \downarrow \alpha'} \hat{f}_{\vb{R}_j \uparrow \beta'}}{0} \MC{Q}_{\alpha' \alpha} + \MC{Q}^\dagger_{\beta \beta'} \Expval{\hat{f}_{\vb{R}_i \downarrow \alpha'} \hat{f}^\dagger_{\vb{R}_j \downarrow \beta'}}{0} \MC{Q}_{\alpha' \alpha}
        \end{aligned} \\
        =& \frac{1}{N_k} \sum_{\vb{k}} e^{i \vb{k} \cdot (\vb{R}_j - \vb{R}_i)} \sum_{\alpha' \beta'} \begin{aligned}[t] & \MC{R}^\dagger_{\beta \beta'} \Expval{\hat{f}^\dagger_{\vb{k} \uparrow \alpha'} \hat{f}_{\vb{k} \uparrow \beta'}}{0} \MC{R}_{\alpha' \alpha} + \MC{Q}^\dagger_{\beta \beta'} \Expval{\hat{f}^\dagger_{\vb{k} \uparrow \alpha'} \hat{f}^\dagger_{-\vb{k} \downarrow \beta'}}{0} \MC{R}_{\alpha' \alpha} \\
        &+ \MC{R}^\dagger_{\beta \beta'} \Expval{\hat{f}_{- \vb{k} \downarrow \alpha'} \hat{f}_{\vb{k} \uparrow \beta'}}{0} \MC{Q}_{\alpha' \alpha} + \MC{Q}^\dagger_{\beta \beta'} \Expval{\hat{f}_{-\vb{k} \downarrow \alpha'} \hat{f}^\dagger_{-\vb{k} \downarrow \beta'}}{0} \MC{Q}_{\alpha' \alpha}
        \end{aligned} 
    \\[1em]
    \bra{\Psi_G} \hat{f}^\dagger_{\vb{R}_i \uparrow \alpha} \hat{f}^\dagger_{\vb{R}_j \downarrow \beta} \ket{\Psi_G}
        =& \frac{1}{N_k} \sum_{\vb{k}} e^{i \vb{k} \cdot (\vb{R}_j - \vb{R}_i)} \sum_{\alpha' \beta'} \begin{aligned}[t] & \MC{R}^\intercal_{\beta \beta'} \Expval{\hat{f}^\dagger_{\vb{k} \uparrow \alpha'} \hat{f}^\dagger_{-\vb{k} \downarrow \beta'}}{0} \MC{R}_{\alpha' \alpha} - \MC{Q}^\intercal_{\beta \beta'} \Expval{\hat{f}^\dagger_{\vb{k} \uparrow \alpha'} \hat{f}_{\vb{k} \uparrow \beta'}}{0} \MC{R}_{\alpha' \alpha} \\
        &+ \MC{R}^\intercal_{\beta \beta'} \Expval{\hat{f}_{- \vb{k} \downarrow \alpha'} \hat{f}^\dagger_{-\vb{k} \downarrow \beta'}}{0} \MC{Q}_{\alpha' \alpha} - \MC{Q}^\intercal_{\beta \beta'} \Expval{\hat{f}_{- \vb{k} \downarrow \alpha'} \hat{f}_{\vb{k} \uparrow \beta'}}{0} \MC{Q}_{\alpha' \alpha}
        \end{aligned} 
    \\[1em]
    \bra{\Psi_G} \hat{f}^\dagger_{\vb{R}_i \uparrow \alpha} \hat{c}_{\vb{k} \uparrow a} \ket{\Psi_G} =&
        \frac{1}{\sqrt{N_k}} e^{-i \vb{k} \cdot \vb{R}_i} \sum_{\alpha'} \Expval{\hat{f}^\dagger_{\vb{k} \uparrow \alpha} \hat{c}_{\vb{k} \uparrow a}}{0} \MC{R}_{\alpha' \alpha} + \Expval{\hat{f}_{-\vb{k} \downarrow \alpha'} \hat{c}_{\vb{k} \uparrow a}}{0} \MC{Q}_{\alpha' \alpha}\\[1em]
    \bra{\Psi_G} \hat{f}_{\vb{R}_i \downarrow \alpha} \hat{c}_{\vb{k} \uparrow a} \ket{\Psi_G} =& \frac{1}{\sqrt{N_k}} e^{- i \vb{k} \cdot \vb{R}_i} \sum_{\alpha'} \Expval{\hat{f}_{-\vb{k} \downarrow \alpha'} \hat{c}_{\vb{k} \uparrow a}}{0} \MC{R}_{\alpha' \alpha} - \Expval{\hat{f}^\dagger_{\vb{k} \uparrow \alpha'} \hat{c}_{\vb{k} \uparrow a}}{0} \MC{Q}_{\alpha' \alpha}
\end{aligned} \label{cal:non-local pairing and hopping}
\end{equation}
Where we've used the translational invariance for $\ket{\Phi_0}$ and the Fourier transform is defined as $\hat{f}_{\vb{k}} = \frac{1}{\sqrt{N_k}} \sum_{i} e^{-i\vb{k} \cdot \vb{R}_i} \hat{f}_{\vb{R}_i}$. To have a more compact expression, we define the Nambu basis in momentum space:
\begin{equation}
\begin{aligned}
    \bm{\Psi}^\dagger_{\vb{k}} =& \begin{pmatrix}
        \bm{\psi}^\dagger_{\vb{k} \uparrow} & \bm{\psi}_{-\vb{k} \downarrow}
    \end{pmatrix} \\
    \bm{\psi}^\dagger_{\vb{k} \uparrow} =& \begin{pmatrix}
        \hat{f}^\dagger_{\vb{k} \uparrow \alpha_1} & \dots & \hat{f}^\dagger_{\vb{k} \uparrow \alpha_{\mathcal{N}_f}}  & \hat{c}^\dagger_{\vb{k} \uparrow a_0}& \dots  & \hat{c}^\dagger_{\vb{k} \uparrow a_{\mathcal{N}_c}} 
    \end{pmatrix} \\
    \bm{\psi}_{-\vb{k} \downarrow} =& \begin{pmatrix}
        \hat{f}_{-\vb{k} \downarrow \alpha_1} & \dots & \hat{f}_{-\vb{k} \downarrow \alpha_{\mathcal{N}_f}} & \hat{c}_{-\vb{k} \downarrow a_0} & \dots & \hat{c}_{-\vb{k} \downarrow a_{\mathcal{N}_c}}
    \end{pmatrix}
\end{aligned} \label{def:Nambu_basis_k}
\end{equation}
The uncorrelated Nambu reduced density matrix in momentum space is defined as: (The superscript $F$ marks expecation value is taken with respect to $\ket{\Phi_0}$)
\begin{gather}
    (\bm{\varrho}^F_{\vb{k}})_{s \alpha, s' \beta} = \bra{\Phi_0} \bm{\Psi}^\dagger_{\vb{k}; s' \beta} \bm{\Psi}_{\vb{k}; s \alpha} \ket{\Phi_0} \label{def:varrho0_k}
\end{gather}
One can check that $E_{\text{kin}}$ can be written as:
\begin{gather}
    E_{\text{kin}} = 2 \frac{1}{N_k} \sum_{\vb{k}} \Tr[
        \mathcal{H}_{\vb{k}}
        \begin{pmatrix}
            \mathcal{R}^\dagger & \mathcal{Q}^\dagger
        \end{pmatrix} \bm{\varrho}^F_{\vb{k}} 
        \begin{pmatrix}
            \mathcal{R} \\ \mathcal{Q}
        \end{pmatrix}
    ], \\
    \mathcal{H}_{\vb{k}} \defeq \begin{pmatrix}
            t_{\vb{k}} &  V_{\vb{k}}^\dagger \\
            V_{\vb{k}} & \epsilon^c_{\vb{k}}
        \end{pmatrix}, \quad t_{\vb{k}; \alpha \beta} \defeq \sum_{i \neq j} e^{i \vb{k} \cdot (\vb{R}_j - \vb{R}_i)} t_{ij; \alpha \beta}
\end{gather}
where the factor 2 comes from spin degeneracy. $\mathcal{R}$ and $\mathcal{Q}$ are both $\mathcal{N}_{f+c} \times \mathcal{N}_{f+c}$ size matrices promoted from $\MC{R}$ and $\MC{Q}$: 
\begin{gather}
    \mathcal{R} = \begin{pmatrix*}[l]
        \MC{R} & 0 \\ 0 & \mathbb{I}_{\scriptscriptstyle \mathcal{N}_c \times \mathcal{N}_c}
    \end{pmatrix*}, \quad
    \mathcal{Q} = \begin{pmatrix*}[l]
        \MC{Q} & 0 \\ 0 & 0_{\scriptscriptstyle \mathcal{N}_c \times \mathcal{N}_c}
    \end{pmatrix*}
\end{gather}
Up to a constant shift (The constant is actually \[\sum_{\vb{k}} \Tr[(\mathcal{R}^\dagger \mathcal{R} + \mathcal{Q}^\dagger \mathcal{Q}) \mathcal{H}_{-\vb{k}}]=\sum_{\vb{k}} \Tr[(\MC{R}^\dagger \MC{R} + \MC{Q}^\dagger \MC{Q})t_{\vb{k}}] + \sum_{\vb{k} a} \epsilon^c_{\vb{k} a}\], but since onsite terms among $f$-orbital are absent in $t_{\vb{k}}$, we have $\sum_{\vb{k}} t_{\vb{k}} = 0$) $\sum_{\vb{k} a} \epsilon^c_{\vb{k} a}$, $E_{\text{kin}}$ can also be expressed as:
\begin{gather}
    E_{\text{kin}} = \frac{1}{N_k}\sum_{\vb{k}}\Tr[
        \begin{pmatrix}
            \mathcal{H}_{\vb{k}} & 0 \\ 0 & -\mathcal{H}^*_{-\vb{k}}
        \end{pmatrix}
        \begin{pmatrix}
            \mathcal{R}^\dagger & \mathcal{Q}^\dagger \\
            -\mathcal{Q}^\intercal & \mathcal{R}^\intercal
        \end{pmatrix} \bm{\varrho}^F_{\vb{k}} 
        \begin{pmatrix}
            \mathcal{R} & - \mathcal{Q}^* \\ \mathcal{Q} & \mathcal{R}^*
        \end{pmatrix}
    ], \label{res:Ekin}
\end{gather}
where the last three terms in trace can be identified as physical Nambu reduced density matrix $\bm{\varrho}^G_{\vb{k}}$ without onsite contributions from the $f$-orbital:
\begin{gather}
    \bm{\varrho}^G_{\vb{k}} \defeq        
         \begin{pmatrix}
            \mathcal{R}^\dagger & \mathcal{Q}^\dagger \\
            -\mathcal{Q}^\intercal & \mathcal{R}^\intercal
        \end{pmatrix} \bm{\varrho}^F_{\vb{k}} 
        \begin{pmatrix}
            \mathcal{R} & - \mathcal{Q}^* \\ \mathcal{Q} & \mathcal{R}^*
        \end{pmatrix} + \bm{\varrho}^G_{ff; \text{onsite}} 
\end{gather}
$E_{\text{atom}}$ can be directly calculated from \eqref{eq:onsite_observable}:
\begin{align}
    E_{\text{atom}} =& \Tr[\bm{\Lambda}^\dagger \bm{H}_{\text{atom}} \bm{\Lambda} \bm{m}^0] \\
    (\bm{H}_{\text{atom}})_{I I'} =& \bra{I} \hat{H}^{\text{int}} + \hat{H}^0_{\text{onsite}} \ket{I'}
\end{align}
We've worked out all the energy components in $\mathcal{L}$, now we let's look at the constraints $\bm{g}^F$ and $\bm{g}^B$ in \eqref{eq:energy_functional}. We denote $l=0$ as the normalization constraint for $\ket{\Phi_0}$ and \eqref{GA:normalization}:
\begin{equation}
    \begin{aligned}
        \lambda^F_0 g^F_0 =& E^F \left(1 - \bra{\Phi_0} \ket{\Phi_0}\right), \quad \lambda^F_0 \defeq E^F \\
        \lambda^B_0 g^B_0 =& E^B \left(1 - \Tr[\bm{\Lambda}^\dagger \bm{\Lambda} \bm{m}^0] \right), \quad \lambda^B_0 \defeq E^B
    \end{aligned}
\end{equation}
As mentioned in the main text, we promoted the $\bm{\varrho}^0$ to be an extra variational degree of freedom in addition to $\ket{\Phi_0}$ and $\hat{P}$. For consistency, we need to make sure the uncorrelated density matrix calculated from $\ket{\Phi_0}$:
\begin{gather}
    \bm{\varrho}^F = \frac{1}{N_k} \sum_{\vb{k}} (\varrho^0_{\vb{k}})_{ff}, \quad \text{We take $f$-orbital part from $\bm{\varrho}^F_{\vb{k}}$}
\end{gather}
with $\varrho^0_{\vb{k}}$ defined in \eqref{def:varrho0_k} equals to $\bm{\varrho}^0$ during the variational process. The similar constraints for $\hat{P}$ are just the Gutzwiller constraints \eqref{eq:gutzwiller_constraints} where we can define the left hand side as $\bm{\varrho}^B$:
\begin{equation}
    \begin{gathered}
        \bm{\varrho}^B = \begin{pmatrix}
            \bm{\rho}^B & \bm{\Delta}^B \\
            (\bm{\Delta}^B)^\dagger & \mathbb{I} - (\bm{\rho}^B)^\intercal
        \end{pmatrix} \\
        \bm{\rho}^B_{\alpha \beta} \defeq \bra{\Phi_0} \hat{P}^\dagger \hat{P} \hat{f}^\dagger_{\uparrow \beta} \hat{f}_{\uparrow \alpha} \ket{\Phi_0} = \Tr[\bm{\Lambda}^\dagger \bm{\Lambda} \bm{S}^\dagger_{\uparrow \beta} \bm{S}_{\uparrow \alpha} \bm{m}^0] \\
        \bm{\Delta}^B_{\alpha \beta} \defeq \bra{\Phi_0} \hat{P}^\dagger \hat{P} \hat{f}_{\downarrow \beta} \hat{f}_{\uparrow \alpha} \ket{\Phi_0} = \Tr[\bm{\Lambda}^\dagger \bm{\Lambda} \bm{S}_{\downarrow \beta} \bm{S}_{\uparrow \alpha} \bm{m}^0]
    \end{gathered}
\end{equation}
In $\bm{g}^{F(B)}$, the $l \geq 1$ constraints are just to enforce $\bm{\varrho}^{F(B)} = \bm{\varrho}^0$ where the indcies $l$ denote the Hermitian matrix basis $\{ \mathcal{O}_l \}$ to expand $\bm{\varrho}^0$, $\bm{\varrho}^{F(B)}$ and $\bm{\lambda}^{F(B)}$. The basis sets $\{ \mathcal{O}_l \}$ should commute with all symmetry operations imposed on $\hat{P}$ and $\ket{\Phi_0}$ and they're chosen to be orthonormal for convenience:
\begin{gather}
    \Tr[\mathcal{O}_l \mathcal{O}_{l'}] = \delta_{l l'} \label{def:singlebodybasis}
\end{gather}
Thus, the constraints read:
\begin{equation}
    \begin{aligned}
        \Tr[\bm{\lambda}^{F(B)} \left(\bm{\varrho}^{F(B)} - \bm{\varrho}^0\right)] =& \Tr[\sum_{l} \lambda^{F(B)}_l \mathcal{O}_l \left(\sum_{l'} (\varrho^{F(B)}_{l'} - \varrho^0_{l'}) \mathcal{O}_{l'} \right)] \\
        =& \sum_{l} \lambda^{F(B)}_l \left(\varrho^{F(B)}_{l} - \varrho^0_{l} \right)
    \end{aligned} \label{eq:varrho_constraints}
\end{equation}
Where 
\begin{gather}
    \varrho^F_l = \Tr[\mathcal{O}_l \bm{\varrho}^F] = \bra{\Phi_0} \hat{\varrho}^F_l \ket{\Phi_0}, \quad \varrho^B_l = \Tr[\mathcal{O}_l \bm{\varrho}^B] = (\text{The bilinear form is given below}) \\
    \hat{\varrho}^F_l = \sum_{\alpha \beta} (\mathcal{O}_l)_{\uparrow \alpha; \uparrow \beta} \hat{f}^\dagger_{\uparrow \alpha} \hat{f}_{\uparrow \beta} + (\mathcal{O}_l)_{\uparrow \alpha; \downarrow \beta} \hat{f}^\dagger_{\uparrow \alpha} \hat{f}^\dagger_{\downarrow \beta} + (\mathcal{O}_l)_{\downarrow \alpha; \uparrow \beta} \hat{f}_{\downarrow \alpha} \hat{f}_{\uparrow \beta} + (\mathcal{O}_l)_{\downarrow \alpha; \downarrow \beta} \hat{f}_{\downarrow \alpha} \hat{f}^\dagger_{\downarrow \beta} = \big(\bm{\Psi}^\dagger_{\vb{k}} \big)_f \mathcal{O}_l \big(\bm{\Psi}_{\vb{k}}\big)_f \label{def:operator_rhoF}
\end{gather}
For the sake of convenience, we can also expand $\bm{\Lambda}$ in terms of many-body basis constrained by the system's symmetry:
\begin{gather}
    \bm{\Lambda} = \sum_{\nu} a_{\nu} \bm{\Gamma}_{\nu} \label{def:many-body basis}
\end{gather}
where $\bm{a}$ is the vector of expansion coefficients (MVP: many-body variational parameters) and $\{ \bm{\Gamma}_{\nu} \}$ is the basis matrix. Therefore, all relevant quantities: $\MC{R} (\MC{R}^\dagger),\MC{Q} (\MC{Q}^\dagger), E_{\text{atom}}, \Tr[\bm{\Lambda}^\dagger \bm{\Lambda} \bm{m}^0], \varrho_l^B$ in $\mathcal{L}$ can be expressed as a bilinear function of $\bm{a}$, whose kernel is denoted by a hollow matrix: $\mbbm{R}(\mbbm{R}^\dagger)$, $\mbbm{Q}(\mbbm{Q}^\dagger)$, $\mbbm{H}_{\text{atom}}$, $\mbbm{F}, \mbbm{P}_l$: (all traces operate on Fock space inices)
\begin{alignat}{2}
    \begin{pmatrix}
        \MC{R} \\ \MC{Q}
    \end{pmatrix} =& 
    \bm{a}^\dagger \begin{pmatrix}
        \mbbm{R} \\ \mbbm{Q}
    \end{pmatrix} \bm{a}, &\quad
    \begin{pmatrix}
        \mbbm{R}_{\nu \nu'} \\ \mbbm{Q}_{\nu \nu'}
    \end{pmatrix} =& \sum_{\alpha \beta}
    \begin{pmatrix} 
        [(\bm{\varrho}^0)^{-1}]_{\uparrow \alpha; \uparrow \beta} \Tr[\bm{\Gamma}^\dagger_{\nu} \bm{S}^\dagger_{\uparrow \alpha} \bm{\Gamma}_{\nu'} \bm{S}_{\uparrow \beta} \bm{m}^0] + [(\bm{\varrho}^0)^{-1}]_{\uparrow \alpha; \downarrow \beta} \Tr[\bm{\Gamma}^\dagger_{\nu} \bm{S}^\dagger_{\uparrow \alpha} \bm{\Gamma}_{\nu'} \bm{S}^\dagger_{\downarrow \beta} \bm{m}^0] \\
        [(\bm{\varrho}^0)^{-1}]_{\downarrow \alpha; \uparrow \beta}  \Tr[\bm{\Gamma}^\dagger_{\nu} \bm{S}^\dagger_{\uparrow \alpha} \bm{\Gamma}_{\nu'} \bm{S}_{\uparrow \beta} \bm{m}^0] + [(\bm{\varrho}^0)^{-1}]_{\downarrow \alpha; \downarrow \beta} \Tr[\bm{\Gamma}^\dagger_{\nu} \bm{S}^\dagger_{\uparrow \alpha} \bm{\Gamma}_{\nu'} \bm{S}^\dagger_{\downarrow \beta} \bm{m}^0]
    \end{pmatrix} \nonumber \\
    & \ & =& (\bm{\varrho}^0)^{-1} \begin{pmatrix}
        (\mbbm{K}_{\uparrow \uparrow})_{\nu \nu'} \\ (\mbbm{K}_{\downarrow \uparrow})_{\nu \nu'}
    \end{pmatrix}
    \label{def:HM_RQ} \\[1.5em] 
    \begin{pmatrix}
        \MC{R}^\dagger \\ \MC{Q}^\dagger
    \end{pmatrix} =& \bm{a}^\dagger \begin{pmatrix}
        \mbbm{R}^\dagger \\ \mbbm{Q}^\dagger
    \end{pmatrix} \bm{a}, &\quad
    \begin{pmatrix}
        \mbbm{R}^\dagger_{\nu \nu'} \\ \mbbm{Q}^\dagger_{\nu \nu'}
    \end{pmatrix} =& \sum_{\alpha \beta} 
    \begin{pmatrix}
        \Tr[\bm{m}^0 \bm{\Gamma}^\dagger_{\nu} \bm{S}^\dagger_{\uparrow \alpha} \bm{\Gamma}_{\nu'} \bm{S}_{\uparrow \beta}] [(\bm{\varrho}^0)^{-1}]_{\uparrow \alpha; \uparrow \beta} + \Tr[\bm{m}^0 \bm{\Gamma}^\dagger_{\nu} \bm{S}_{\uparrow \alpha} \bm{\Gamma}_{\nu'} \bm{S}_{\downarrow \beta} [(\bm{\varrho}^0)^{-1}]]_{\downarrow \alpha; \uparrow \beta} \\
        \Tr[\bm{m}^0 \bm{\Gamma}^\dagger_{\nu} \bm{S}^\dagger_{\uparrow \alpha} \bm{\Gamma}_{\nu'} \bm{S}_{\uparrow \beta}] [(\bm{\varrho}^0)^{-1}]_{\uparrow \alpha; \downarrow \beta} + \Tr[\bm{m}^0 \bm{\Gamma}^\dagger_{\nu} \bm{S}_{\uparrow \alpha} \bm{\Gamma}_{\nu'} \bm{S}_{\downarrow \beta} [(\bm{\varrho}^0)^{-1}]]_{\downarrow \alpha; \downarrow \beta}
    \end{pmatrix} \\
    & & =& \begin{pmatrix}
        (\mbbm{K}^\dagger_{\uparrow \uparrow})_{\nu \nu'} & (\mbbm{K}^\dagger_{\downarrow \uparrow})_{\nu \nu'}
    \end{pmatrix} (\bm{\varrho}^0)^{-1}
    \label{def:HM_RQdgr} \\[1.5em] 
    E_{\text{atom}} =& N_k \bm{a}^\dagger \mbbm{H}_{\text{atom}} \bm{a}, &\quad \mbbm{H}_{\text{atom}} =& \Tr[\bm{\Gamma}^\dagger_{\nu} (\bm{H}_{\text{atom}})_{\nu \nu'} \bm{\Gamma}_{\nu'} \bm{m}^0] \label{def:HM_atom} \\[2em] 
    \Tr[\bm{\Lambda}^\dagger \bm{\Lambda} \bm{m}^0] =& \bm{a}^\dagger \mbbm{F} \bm{a}, &\quad \mbbm{F}_{\nu \nu'} =& \Tr[\bm{\Gamma}^\dagger_{\nu} \bm{\Gamma}_{\nu'} \bm{m}^0] \label{def:HM_F} \\[2em]
    \varrho^B_l =& \bm{a}^\dagger \mbbm{P}_l \bm{a}, &\quad (\mbbm{P}_l)_{\nu \nu'} =& \sum_{\alpha; \beta} (\mathcal{O}_{l})_{\uparrow \alpha; \uparrow \beta} \Tr[\bm{\Gamma}^\dagger_{\nu} \bm{\Gamma}_{\nu'} \bm{S}^\dagger_{\uparrow \alpha} \bm{S}_{\uparrow \beta} \bm{m}^0] + (\mathcal{O}_l)_{\uparrow \alpha; \downarrow \beta} \Tr[\bm{m}^0 \bm{S}^\dagger_{\uparrow \alpha} \bm{S}^\dagger_{\downarrow \beta} \bm{\Gamma}^\dagger_{\nu} \bm{\Gamma}_{\nu'}] \nonumber \\
    & & &+ (\mathcal{O}_l)_{\downarrow \alpha; \uparrow \beta} \Tr[\bm{\Gamma}^\dagger_{\nu} \bm{\Gamma}_{\nu'} \bm{S}_{\downarrow \alpha} \bm{S}_{\uparrow \beta} \bm{m}^0] + (\mathcal{O}_l)_{\downarrow \alpha; \downarrow \beta} \left(\delta_{\alpha \beta} - \Tr[\bm{\Gamma}^\dagger_{\nu} \bm{\Gamma}_{\nu'} \bm{S}^\dagger_{\uparrow \beta} \bm{S}_{\uparrow \alpha} \bm{m}^0] \right) \label{def:HM_P}
\end{alignat}

\subsection{Variational Gutzwiller Equations}
We have worked out the explicit form of $\mathcal{L}$ in the last section, now we can derive the variational equations by fixing $\bm{\varrho}^0$:
\begin{align}
    \encvert{\pfpx{\mathcal{L}}{\bra{\Phi_0}}}{\bm{\varrho}^0} = 0 \Rightarrow& \hat{H}^F \ket{\Phi_0} = E^F \ket{\Phi_0} \label{eq:Fermi}\\
    \encvert{\pfpx{\mathcal{L}}{\bra{\bm{a}}}}{\bm{\varrho}^0} = 0 \Rightarrow& \mbbm{H}^B \ket{\bm{a}} = E^B \mbbm{F} \ket{\bm{a}} \label{eq:Bose}
\end{align}
Both $\ket{\Phi}$ and $\ket{\bm{a}}$ are the ground state of the respective Hamiltonians. These two are co-dependent linear equations because $\hat{H}^F$ depends on $\bm{\varrho}^0$ and $\ket{\bm{a}}$ while $\mbbm{H}^B$ depends on $\bm{\varrho}^0$ and $\ket{\Phi_0}$, thus they need to be solved self-consistently. $\hat{H}^F$ can be derived referring to \eqref{def:varrho0_k}, \eqref{res:Ekin} and \eqref{eq:varrho_constraints}:
\begin{gather}
    \hat{H}^F = \sum_{\vb{k}} \bm{\Psi}^\dagger_{\vb{k}} \Bigg[ 
    \begin{pmatrix}
        \mathcal{R} & - \mathcal{Q}^* \\ \mathcal{Q} & \mathcal{R}^*
    \end{pmatrix}
    \begin{pmatrix}
        \mathcal{H}_{\vb{k}} & 0 \\ 0 & -\mathcal{H}^*_{-\vb{k}}
    \end{pmatrix}
    \begin{pmatrix}
        \mathcal{R}^\dagger & \mathcal{Q}^\dagger \\
        -\mathcal{Q}^\intercal & \mathcal{R}^\intercal
    \end{pmatrix} 
    + \sum_{l} \lambda^F_l \mathcal{O}^F_l \Bigg]
    \bm{\Psi}_{\vb{k}}, \label{def:HF} \\
    (\mathcal{O}^F_l)_{ff} = \mathcal{O}_l, \quad (\mathcal{O}^F_l)_{fc} = (\mathcal{O}^F_l)_{cf} = (\mathcal{O}^F_l)_{cc} = 0
\end{gather}
Similarly, $\mbbm{H}^B$ can be derived by referring to \eqref{res:Ekin} and Eqs.~\eqref{def:HM_RQ}--\eqref{def:HM_P}. First we calculate $\ket{\bm{a}}$ dependent part in $E_{\text{kin}}$ (connecting bath to the impurity):
\begin{align}
    \chi^\dagger_{\alpha \beta} \defeq& \encvert{\pfpx{E_{\text{kin}}}{\MC{R}_{\beta \alpha}}}{\ket{\Phi_0}, \bm{\varrho}^0} = \frac{1}{N_k}\sum_{\vb{k}} \Bigg\{ \left[ \begin{pmatrix}
            \mathcal{H}_{\vb{k}} & 0 \\ 0 & -\mathcal{H}^*_{-\vb{k}}
        \end{pmatrix}
        \begin{pmatrix}
            \mathcal{R}^\dagger & \mathcal{Q}^\dagger \\
            -\mathcal{Q}^\intercal & \mathcal{R}^\intercal
        \end{pmatrix} \bm{\varrho}^0_{\vb{k}} 
 \right]_{\uparrow \alpha; \uparrow \beta} + \left[\bm{\varrho}^0_{\vb{k}}  \begin{pmatrix}
        \mathcal{R} & - \mathcal{Q}^* \\ \mathcal{Q} & \mathcal{R}^*
    \end{pmatrix} \begin{pmatrix}
            \mathcal{H}_{\vb{k}} & 0 \\ 0 & -\mathcal{H}^*_{-\vb{k}}
        \end{pmatrix} \right]_{\downarrow \beta; \downarrow \alpha} \Bigg\}\\
    \chi_{\alpha \beta} \defeq& \encvert{\pfpx{E_{\text{kin}}}{\MC{R}^\dagger_{\beta \alpha}}}{\ket{\Phi_0}, \bm{\varrho}^0} = \frac{1}{N_k} \sum_{\vb{k}} \Bigg\{ \left[ \begin{pmatrix}
            \mathcal{H}_{\vb{k}} & 0 \\ 0 & -\mathcal{H}^*_{-\vb{k}}
        \end{pmatrix}
        \begin{pmatrix}
            \mathcal{R}^\dagger & \mathcal{Q}^\dagger \\
            -\mathcal{Q}^\intercal & \mathcal{R}^\intercal
        \end{pmatrix} \bm{\varrho}^0_{\vb{k}} 
 \right]_{\downarrow \beta; \downarrow \alpha} + \left[\bm{\varrho}^0_{\vb{k}}  \begin{pmatrix}
        \mathcal{R} & - \mathcal{Q}^* \\ \mathcal{Q} & \mathcal{R}^*
    \end{pmatrix} \begin{pmatrix}
            \mathcal{H}_{\vb{k}} & 0 \\ 0 & -\mathcal{H}^*_{-\vb{k}}
        \end{pmatrix} \right]_{\uparrow \alpha; \uparrow \beta} \Bigg\}\\
    \Upsilon^\dagger_{\alpha \beta} \defeq& \encvert{\pfpx{E_{\text{kin}}}{\MC{Q}_{\beta \alpha}}}{\ket{\Phi_0}, \bm{\varrho}^0} = \frac{1}{N_k} \sum_{\vb{k}} \Bigg\{ \left[ \begin{pmatrix}
            \mathcal{H}_{\vb{k}} & 0 \\ 0 & -\mathcal{H}^*_{-\vb{k}}
        \end{pmatrix}
        \begin{pmatrix}
            \mathcal{R}^\dagger & \mathcal{Q}^\dagger \\
            -\mathcal{Q}^\intercal & \mathcal{R}^\intercal
        \end{pmatrix} \bm{\varrho}^0_{\vb{k}} 
 \right]_{\uparrow \alpha; \downarrow \beta} + \left[\bm{\varrho}^0_{\vb{k}}  \begin{pmatrix}
        \mathcal{R} & - \mathcal{Q}^* \\ \mathcal{Q} & \mathcal{R}^*
    \end{pmatrix} \begin{pmatrix}
            \mathcal{H}_{\vb{k}} & 0 \\ 0 & -\mathcal{H}^*_{-\vb{k}}
        \end{pmatrix} \right]_{\downarrow \beta; \uparrow \alpha} \Bigg\} \\
    \Upsilon_{\alpha \beta} \defeq& \encvert{\pfpx{E_{\text{kin}}}{\MC{Q}^\dagger_{\beta \alpha}}}{\ket{\Phi_0}, \bm{\varrho}^0} = \frac{1}{N_k} \sum_{\vb{k}} \Bigg\{ \left[ \begin{pmatrix}
            \mathcal{H}_{\vb{k}} & 0 \\ 0 & -\mathcal{H}^*_{-\vb{k}}
        \end{pmatrix}
        \begin{pmatrix}
            \mathcal{R}^\dagger & \mathcal{Q}^\dagger \\
            -\mathcal{Q}^\intercal & \mathcal{R}^\intercal
        \end{pmatrix} \bm{\varrho}^0_{\vb{k}} 
 \right]_{\downarrow \beta; \uparrow \alpha} + \left[\bm{\varrho}^0_{\vb{k}}  \begin{pmatrix}
        \mathcal{R} & - \mathcal{Q}^* \\ \mathcal{Q} & \mathcal{R}^*
    \end{pmatrix} \begin{pmatrix}
            \mathcal{H}_{\vb{k}} & 0 \\ 0 & -\mathcal{H}^*_{-\vb{k}}
        \end{pmatrix} \right]_{\uparrow \alpha; \downarrow \beta} \Bigg\}
\end{align}
Then we can write down $\mbbm{H}^B$ explicitly: (The trace operates on orbital indices of $f$-orbital: $\alpha, \beta, \dots$)
\begin{gather}
    \mbbm{H}^B = \Tr(\chi^\dagger \mbbm{R} + \chi \mbbm{R}^\dagger + \Upsilon^\dagger \mbbm{Q} + \Upsilon \mbbm{Q}^\dagger) + \mbbm{H}_{\text{atom}} + \sum_{l} \lambda^B_l \mbbm{P}_l
\end{gather}
Notice that one can always symmetrize $\mbbm{H}^B \rightarrow \frac{1}{2} \left[\mbbm{H}^B + (\mbbm{H}^B)^\dagger\right]$ since its eigenvalue $E^B$ is real. We relabel all double font matrices to be their symmetrized ones with respect to many-body basis indices \eqref{def:many-body basis}: 
\begin{gather}
  (\mbbm{R}_{\alpha \beta})_{\nu \nu'} \rightarrow \frac{1}{2} \left[ (\mbbm{R}_{\alpha \beta})_{\nu \nu'} + (\mbbm{R}_{\alpha \beta})^*_{\nu' \nu}\right], \quad (\mbbm{Q}_{\alpha \beta})_{\nu \nu'} \rightarrow \frac{1}{2} \left[ (\mbbm{Q}_{\alpha \beta})_{\nu \nu'} + (\mbbm{Q}_{\alpha \beta})^*_{\nu' \nu}\right], \quad \dots \label{def:symm_double_font_matrix}
\end{gather}
Having these explicit expression, we can set out to solve the self-consistent equations Eqs.~\eqref{eq:Fermi}-\eqref{eq:Bose} while searching for the correct lagrange multipliers $\lambda^{F(B)}_l$ to satisfy the constraints \eqref{eq:varrho_constraints}=0. One usually start the SCF loop by taking an initial guess $\MC{R}_{\text{init}}$ and $\MC{Q}_{\text{init}}$ for Fermi part and finding $\bm{\lambda}^F$ \eqref{eq:Fermi}, then pass $\chi^{(\dagger)}$ and $\Upsilon^{(\dagger)}$ to Bose part \eqref{eq:Bose} to find $\bm{\lambda}^B$ which generates a new $\MC{R}$ and $\MC{Q}$ for the Fermi part. Once the process converges (usually the quasi-particle weight $Z = \MC{R}^\dagger \MC{R} + \MC{Q}^\dagger \MC{Q}$ is used as convergence indicator, i.e. when $\abs{Z_{n+1} - Z_{n}} < \varepsilon$), we get 
\begin{gather}
    \mathcal{L}_{\text{SCF}}[\mu, \bm{\varrho}^0] = \min_{\ket{\Phi_0}, \ket{\tilde{\bm{a}}}, \bm{\lambda}^F, \bm{\lambda}^B} \mathcal{L}[\mu, \bm{\varrho}^0, \ket{\Phi_0}, \ket{\tilde{\bm{a}}}, \bm{\lambda}^F, \bm{\lambda}^B].
\end{gather}
The next step is then find an optimal $\bm{\varrho}^0$ that minimizes $\mathcal{L}_{\text{SCF}}[\mu, \bm{\varrho}^0]$ by using gradient descent or other optimization algorithms. Therefore, the process can be significantly speed up by analytical derivative $\encvert{\dfdx{\mathcal{L}_{\text{SCF}}}{\bm{\varrho}^0}}{\mu}$ which we will derive in Appendix.~\ref{appdx:dLdrho0}. (The safest way to run the algorithm is starting from the uncorrelated limit (e.g. small Hubbard U) where the initial normal state guess can be made as: $\MC{R}=\bm{1}_{N_f \times N_f}$ and $\MC{Q} = \bm{0}_{N_f \times N_f}$, from which we can tune up the correlation strength `adiabatically' to find the solution for desired interactions.)

In the following subsection, we will show how to find $\lambda^{F(B)}$ using Newton's method, especially in how to calculate the Jacobian matrix $\mathcal{J}^{F(B)}_{l l'} = \pfpx{g^{F(B)}_l}{\lambda^{F(B)}_{l'}}$ which is the key step for Newton's method.

\subsubsection{Solve Fermi part: find $\lambda^F$ with analytical Jacobian matrix} \label{sec:FermiPart}
Newton or quasi-Newton method is efficient for finding roots of a non-linear, multi-dimensional vector function $\bm{g}^{F(B)} (\bm{\lambda}^{F(B)}) = 0$. The basic updating scheme is:
\begin{gather}
    \bm{\lambda}^{F(B)}_{n+1} = \bm{\lambda}^{F(B)}_{n} - \left[\mathcal{J}^{F(B)}(\bm{\lambda}^{F(B)}_{n})\right]^{-1} \bm{g}^{F(B)}(\bm{\lambda}^{F(B)}_{n}),
\end{gather}
which leads to the core step of calculating the Jacobian matrix $\mathcal{J}^{F(B)}(\bm{\lambda}^{F(B)})$. The co-dependent nature of $\hat{H}^F$ and $\mbbm{H}^B$ poses complication for defining $\mathcal{J}^{F(B)}$, and we shall deal with it in a `closed-box' manner. For example, when we solve $\bm{g}^F = 0$, only $\hat{H}^F$ is inside the box and both $\bm{\varrho}^0$ and $\mbbm{H}^B$ are outside the box, thus we can treat $\bm{\varrho}^0$ and $\ket{\bm{a}}$ as constants. Therefore, we can define $\mathcal{J}^{F(B)}$ without ambiguity and introduce an abbreviation called `partial partial derivative' \cite{PENG2022108348}: 
\begin{gather}
    \mathcal{J}^F_{l l'} = \encvert{\pfpx{g^F_l}{\lambda^F_{l'}}}{\bm{\varrho}^0, \ket{\bm{a}}} = \ppfpx{g^F_l}{\lambda^F_{l'}}, \quad \mathcal{J}^B_{l l'} = \encvert{\pfpx{g^B_l}{\lambda^B_{l'}}}{\bm{\varrho}^0, \ket{\Phi_0}} = \ppfpx{g^B_l}{\lambda^B_{l'}}
\end{gather}
For pratical implementation, we must obtain the analytical expression for $\mathcal{J}^{F(B)}$ for speed and accuracy. They're derived using second-order perturbation theory by treating $ \delta \bm{\lambda}^{F(B)}$ as a small perturbation to $\hat{H}^F$ and $\mbbm{H}^B$ respectively. We first solve $\mathcal{J}^{F}$ by giving a general expression (including the case with degeneracy) for derivative of a local observable $\hat{\mathcal{A}}_{\vb{R}}$ against a translational invariant parameter $\nu$ in $\hat{H}^F$. For convenience, we may rename the Nambu basis \eqref{def:Nambu_basis_k} to: $\bm{\psi}^\dagger_{\vb{k} \uparrow} \rightarrow \phi^\dagger_{\vb{k} \uparrow}$, $\bm{\psi}_{-\vb{k} \downarrow} \rightarrow \bm{\phi}^\dagger_{\vb{k} \downarrow}$ such that $\hat{H}^F$ becomes a normal charge-U(1) preserving Hamiltonian under the new basis: $\bm{\Phi}^\dagger_{\vb{k}} = \begin{pmatrix}\bm{\phi}^\dagger_{\vb{k} \uparrow} & \bm{\phi}^\dagger_{\vb{k} \downarrow}\end{pmatrix}$ and we can thus utilize the Dirac braket notation: 
\begin{gather}
    \hat{H}^F = \sum_{\vb{k} \xi \xi'} H^F_{\vb{k; \xi \xi'}} \ket{\vb{k} \xi} \bra{\vb{k} \xi'}, \quad \ket{\vb{k} \xi} \leftrightarrow (\bm{\Phi}^\dagger_{\vb{k}})_{\xi}, \\
    \text{The spectrum representation:} \quad \hat{H}^F = \sum_{\vb{k} n} \epsilon_{\vb{k} n} \ket{\vb{k} n} \bra{\vb{k} n}
\end{gather}

The expectation value for $\hat{\mathcal{A}}_{\vb{R}}$ is:
\begin{align}
\expval{\hat{\mathcal{A}}_{\vb{R}}}_0 =& \frac{1}{N} \sum_{\vb{k} n} f_{\vb{k} n} \bra{\vb{k} n} \hat{\mathcal{A}}_{\vb{k}} \ket{\vb{k} n} \\
=&  \frac{1}{N} \sum_{\vb{k} n} \sum_{\alpha \beta} f_{\vb{k} n} \bra{\vb{k} n} \ket{\vb{k} \alpha} \bra{\vb{k} \alpha} \hat{\mathcal{A}}_{\vb{k}} \ket{\vb{k} \beta} \bra{\vb{k} \beta} \ket{\vb{k} n} \\
=& \frac{1}{N} \sum_{\vb{k}} \sum_{\alpha \beta} \mathcal{A}_{\vb{k}; \alpha \beta} \bra{\vb{k} \beta} \left( \sum_n f_{\vb{k} n} \ket{\vb{k} n} \bra{\vb{k} n} \right) \ket{\vb{k} \alpha} \\
f_{\vb{k} n} =& \frac{1}{e^{\beta \epsilon_{\vb{k} n}} + 1}
\end{align}
Where $f_{\vb{k} n}$ is the Fermi-Dirac distribution with Fermi level set to 0 as required by the BCS ground state $\ket{\Phi_0}$ and $1/\beta$ is the smearing temperature for numerical stability. Its derivative reduces to calculate:
\begin{gather}
\partial_{\nu} \left( \sum_n f_{\vb{k} n} \ket{\vb{k} n} \bra{\vb{k} n} \right)
\end{gather}
For non-degenerate $\{ \epsilon_{\vb{k} n} \}$, the above expression is calculated as:
\begin{gather}
\partial_{\nu} \left( \sum_n f_{\vb{k} n} \ket{\vb{k} n} \bra{\vb{k} n} \right) = \sum_{mn} \mathcal{F}_{\vb{k}; mn} \bra{\vb{k} m} \partial_{\nu} \hat{H}_{\vb{k}} \ket{\vb{k} n} \\
\mathcal{F}_{\vb{k}; mn} \defeq \delta_{mn} \dfdx{f_{\vb{k} n}}{\epsilon_{\vb{k} n}} + (1 - \delta_{mn}) \frac{f_{\vb{k} n} - f_{\vb{k} m}}{\epsilon_{\vb{k} n} - \epsilon_{\vb{k} m}}
\end{gather}
Suppose there's only one degeneracy subspace $D$ with energy $\epsilon_{\vb{k}D}$ and eigenvectors $\{ \ket{\vb{k}m^{(0)}} \}$:
\begin{gather}
\sum_n f_{\vb{k} n} \ket{\vb{k} n} \bra{\vb{k} n} =  f_{\vb{k} D} \left(\sum_{m \in D} \ket{\vb{k} m^{(0)}} \bra{\vb{k} m^{(0)}} \right) + \sum_{n \notin D}  f_{\vb{k} n} \ket{\vb{k} n} \bra{\vb{k} n} \label{ketfknbra}
\end{gather}
We focus on the derivative on the first term. There's a gauge degree of freedom on the degenerate subspace. If we consider a perturbation $\lambda \partial_{\nu} \hat{H}_{\vb{k}}$ where $\lambda$ is a real small number, the perturbed wavefunction of $\{ \ket{\vb{k} m^{(0)}} \}$ is given by degenerate perturbation theory and it'll choose a specific gauge denoted by $\{ \ket{\vb{k} l^{(0)}} \}$ which diagonalises $\partial_{\nu} \hat{H}_{\vb{k}}$:
\begin{gather}
\sum_{m \in D} \ket{\vb{k} m^{(0)}} \bra{\vb{k} m^{(0)}} = \sum_{l \in D} \ket{\vb{k} l^{(0)}} \bra{\vb{k} l^{(0)}} \\
\partial_{\nu} \hat{H}_{\vb{k}} \ket{\vb{k} l^{(0)}} = \nu_{\vb{k} l} \ket{\vb{k} l^{(0)}} 
\end{gather}
The perturbed wavefunction are denoted as $\{ \ket{\vb{k} l} \}$:
\begin{gather}
\left(\hat{H}_{\vb{k}} + \lambda \partial_{\nu} \hat{H}_{\vb{k}} \right) \ket{\vb{k} l} = E_{\vb{k}} (\lambda) \ket{\vb{k} l} \\
\ket{\vb{k} l} \xrightarrow{\lambda \rightarrow 0} \ket{\vb{k} l^{(0)}}
\end{gather}
Therefore, a legit definition of derivative for degenerate subspace fixes our gauge to be $\{ \ket{\vb{k} l^{(0)}} \}$
\begin{gather}
\partial_{\nu} \ket{\vb{k} l^{(0)}} \defeq \partial_{\lambda} \ket{\vb{k} l} \vert_{\lambda = 0} = \lim_{\lambda \rightarrow 0} \frac{\ket{\vb{k} l} - \ket{\vb{k} l^{(0)}}}{\lambda} = \lim_{\lambda \rightarrow 0} \frac{\ket{\vb{k} l^{(1)}}}{\lambda}\\
\partial_{\nu} \epsilon_{\vb{k} l} \defeq \partial_{\lambda} \epsilon_{\vb{k} l} \vert_{\lambda = 0} = \bra{\vb{k} l^{(0)}} \partial_{\nu} \hat{H}_{\vb{k}} \ket{\vb{k} l^{(0)}} = \nu_{\vb{k} l}
\end{gather}
We define the projector to degenerate subspace as $\hat{\mathcal{P}}_0$
\begin{gather}
\hat{\mathcal{P}}_0 = \sum_{l \in D} \ket{\vb{k} l^{(0)}} \bra{\vb{k} l^{(0)}} \\
\hat{\mathcal{P}}_1 = \hat{\bm{1}} - \hat{\mathcal{P}}_0
\end{gather}
The degenerate perturbation theory tells us:
\begin{gather}
\hat{\mathcal{P}}_0 \ket{\vb{k} l_i^{(1)}} = \lambda \sum_{j \neq i} \frac{\hat{\mathcal{P}}_0 \ket{\vb{k} l^{(0)}_j}}{\nu_{\vb{k} l_i} - \nu_{\vb{k} l_j}} \bra{\vb{k} l^{(0)}_j} \hat{\mathcal{V}} \ket{\vb{k} l^{(0)}_i} \\
\hat{\mathcal{V}} = \left( \partial_{\nu} \hat{H}_{\vb{k}}  \right) \hat{\mathcal{P}}_1 \frac{1}{E^{(0)}_{\vb{k} D} - \hat{H}_{\vb{k}}} \hat{\mathcal{P}}_1 \left( \partial_{\nu} \hat{H}_{\vb{k}}  \right) \\
\hat{\mathcal{P}}_1 \ket{\vb{k} l_i^{(1)}} = \lambda \sum_{n \notin D} \frac{\ket{\vb{k} n} \bra{\vb{k} n} \partial_{\nu} \hat{H}_{\vb{k}} \ket{\vb{k} l^{(0)}_i}}{E^{(0)}_{\vb{k} D} - E^{(0)}_{\vb{k} n}}
\end{gather}
Where $l_i, l_j, \dots$ labels different eignvectors in the degenerate subspace $D$. The derivative of first term in \eqref{ketfknbra} is:
\begin{align}
\partial_{\nu} \left( f_{\vb{k} D} \sum_{l_i \in D} \ket{\vb{k} l^{(0)}_i} \bra{\vb{k} l^{(0)}_i} \right) =& \dfdx{f_{\vb{k} D}}{\epsilon_{\vb{k} D}} \sum_{l_i \in D} \nu_{\vb{k} l_i} \ket{\vb{k} l^{(0)}_i} \bra{\vb{k} l^{(0)}_i} \nonumber \\
&+ f_{\vb{k} D}  \left[ \sum_{l_i} \left( \partial_{\nu} \ket{\vb{k} l^{(0)}_i} \right) \bra{\vb{k} l^{(0)}_i} + \ket{\vb{k} l^{(0)}_i} \left( \partial_{\nu} \bra{\vb{k} l^{(0)}_i} \right) \right] \label{eq:drvtDeg0}
\end{align}
We focus on the second term grouped by $[\  ]$ can be decomposed by $\hat{P}_0$ and $\hat{P}_1$:
\begin{align}
[*] =&  \sum_{l_i} \hat{P}_0 \left( \partial_{\nu} \ket{\vb{k} l^{(0)}_i} \right) \bra{\vb{k} l^{(0)}_i} + \ket{\vb{k} l^{(0)}_i} \left( \partial_{\nu} \bra{\vb{k} l^{(0)}_i} \right) \hat{P}_0 \nonumber \\
&+ \hat{P}_1 \left( \partial_{\nu} \ket{\vb{k} l^{(0)}_i} \right) \bra{\vb{k} l^{(0)}_i} + \ket{\vb{k} l^{(0)}_i} \left( \partial_{\nu} \bra{\vb{k} l^{(0)}_i} \right) \hat{P}_1
\end{align}
Where the terms from the first line can be proven to be 0:
\begin{align}
\hat{P}_0 [*] \hat{P}_0 =& \sum_{i \neq j \in D} \ket{\vb{k} l^{(0)}_j} \frac{\bra{\vb{k} l^{(0)}_j } \hat{\mathcal{V}} \ket{\vb{k} l^{(0)}_i }}{\nu_{\vb{k} l_i} - \nu_{\vb{k} l_j}} \bra{\vb{k} l^{(0)}_i} + \ket{\vb{k} l^{(0)}_i} \frac{\bra{\vb{k} l^{(0)}_i } \hat{\mathcal{V}} \ket{\vb{k} l^{(0)}_j }}{\nu_{\vb{k} l_i} - \nu_{\vb{k} l_j}} \bra{\vb{k} l^{(0)}_j} \nonumber \\
=& \sum_{i \neq j \in D} \ket{\vb{k} l^{(0)}_j} \frac{\bra{\vb{k} l^{(0)}_j } \hat{\mathcal{V}} \ket{\vb{k} l^{(0)}_i }}{\nu_{\vb{k} l_i} - \nu_{\vb{k} l_j}} \bra{\vb{k} l^{(0)}_i} - \ket{\vb{k} l^{(0)}_j} \frac{\bra{\vb{k} l^{(0)}_j } \hat{\mathcal{V}} \ket{\vb{k} l^{(0)}_i }}{\nu_{\vb{k} l_i} - \nu_{\vb{k} l_j}} \bra{\vb{k} l^{(0)}_i} \nonumber \\
=& 0
\end{align}
Thus, \eqref{eq:drvtDeg0} can be simplified to:
\begin{align}
    \partial_{\nu} \left( f_{\vb{k} D} \sum_{l_i \in D} \ket{\vb{k} l^{(0)}_i} \bra{\vb{k} l^{(0)}_i} \right) =& \dfdx{f_{\vb{k} D}}{\epsilon_{\vb{k} D}} \sum_{l_i \in D} \nu_{\vb{k} l_i} \ket{\vb{k} l^{(0)}_i} \bra{\vb{k} l^{(0)}_i} \nonumber \\
    &+ f_{\vb{k} D} \sum_{l_i \in D} \sum_{n \notin D} \left[ \frac{\ket{\vb{k} n} \bra{\vb{k} n} \partial_{\nu} \hat{H}_{\vb{k}} \ket{\vb{k} l^{(0)}_i} \bra{\vb{k} l^{(0)}_i}}{\epsilon_{\vb{k} D} - \epsilon_{\vb{k} n}} + h.c. \right]
\end{align}
For a general set of levels with multiple degeneracy subsets, the above results can be straightforwardly extended as:
\begin{align}
    \bra{\vb{k} l^i_{D_n}} \partial_{\nu} \left( \sum_{p \in \text{all levels}} f_{\vb{k} p} \ket{\vb{k} p} \bra{\vb{k} p} \right) \ket{\vb{k} l^j_{D_m}} =& (1 - \delta_{mn})\left(\frac{f_{\vb{k} D_n} - f_{\vb{k} D_m}}{\epsilon_{\vb{k} D_n} - \epsilon_{\vb{k} D_m}} \right) \bra{\vb{k} l^i_{D_n}} \partial_{\nu} \hat{H}_{\vb{k}} \ket{\vb{k} l^j_{D_m}} + \delta_{mn} \delta_{ij} \dfdx{f_{\vb{k} D_n}}{\epsilon_{\vb{k} D_n}} \nu_{\vb{k} l^i_{D_n}} \label{formula:Drvt_2ndPert}
\end{align}
Where $D_n, D_m, \dots$ are the degenerate subspaces and $|l^i_{D_n}\rangle, |l^j_{D_n}\rangle, \dots$ are the eigenvectors within $D_n$ subspace. The degeneracy gauge is fixed by diagonalizing $\partial_{\nu} \hat{H}_{\vb{k}}$ within each $D_n$ subspace:
\begin{gather}
    \partial_{\nu} \hat{H}_{\vb{k}} \ket{l^i_{D_n}} = \nu_{\vb{k} l^i_{D_n}} \ket{l^i_{D_n}}
\end{gather}
The derivative formula is:
\begin{gather}
    \partial_{\nu} \expval{\hat{\mathcal{A}}_{\vb{R}}}_0 = \frac{1}{N_k} \sum_{\vb{k}} \sum_{\substack{n \\ i \in D_n}} \sum_{\substack{m \\ j \in D_m}} \bra{\vb{k} l^j_{D_m}}\hat{\mathcal{A}}_{\vb{k}} \ket{\vb{k} l^i_{D_n}} \bra{\vb{k} l^i_{D_n}} \partial_{\nu} \left( \sum_{p\in \text{all levels}} f_{\vb{k} } \ket{\vb{k} p} \bra{\vb{k} p} \right) \ket{\vb{k} l^j_{D_m}}
\end{gather}
The $\mathcal{J}^F_{l l'}$ is just a special case of the above formula by replacing $\partial_{\nu} \hat{H}_{\vb{k}}$ with $\hat{\varrho}^F_{l'}$ \eqref{def:operator_rhoF} and 
$\hat{\mathcal{A}}_{\vb{R}}$ with $\hat{\varrho}^F_{l}$.

\subsubsection{Solve Bose part: find $\lambda^B$ with analytical Jacobian matrix}
First question we need to address is whether $\mbbm{F}$ is positive definite for the generalised eigenvalue equation \eqref{eq:Bose} which would guarantee the solution $E^B$ to be real. We now prove that $\mbbm{F}$ is indeed positive definite as a result of positive definiteness of $\bm{m}^0$.

By Schor's lemma, both $\bm{\Lambda}$ and $\bm{m}^0$ can be brought into same block diagonal form (the number of blocks equals number of irreps of the symmetry group satisfied by the system):
\begin{gather}
\begin{pmatrix}
p^V_1 & \dots & 0 \\
\vdots & \ddots & \vdots \\
0 & \dots & p^V_r
\end{pmatrix}
\end{gather}
Where each block $p^V_k$ consists of its multiplicity blocks:
\begin{gather}
p^V_k = \begin{pmatrix}
r_{11} \mbbm{1}_{d_k} & \dots & r_{1n_k} \mbbm{1}_{d_k}\\
\vdots & \ddots & \vdots \\
r_{n_k1} \mbbm{1}_{d_k} & \dots & r_{n_k n_k} \mbbm{1}_{d_k}
\end{pmatrix}
\end{gather}
$n_k$ is the multiplicity for \textit{k-irrep} and $d_k$ denotes its dimension. We can further define $p^V_{k; ij}$ equals $p^V_k$ with $r_{mn} = \delta_{mi}\delta_{nj}$. Then we can construct our many-body basis $ \bm{\Lambda}_{k; ij}$ as:
\begin{gather}
\bm{\Lambda}_{k; ij} = \frac{1}{\sqrt{d_k}} \begin{pmatrix}
0 & \dots & 0 & \dots & 0\\
\vdots & \ddots & \vdots & & \vdots \\
0 & \dots & p^V_{k; ij} & \dots & 0 \\
\vdots & & \vdots & \ddots & \vdots \\
0 & \dots & 0 & \dots & 0
\end{pmatrix} , \quad (\bm{\Lambda}_{k; ij})_{I I'} = \frac{1}{\sqrt{d_k}} \sum_{t}^{d_k} \delta_{I i_t} \delta_{I' j_t} 
 \\
(\bm{\Lambda}^\intercal_{k; ij})_{I I'} = (\bm{\Lambda}_{k; ij})_{I' I} = \frac{1}{\sqrt{d_k}} \sum_{t}^{d_k} \delta_{I' i_t} \delta_{I j_t} = (\bm{\Lambda}_{k; ji})_{I I'}
\end{gather}
Now we can evaluate elements of $\mbbm{F}_{(k';i'j'), (k;i j)}$:
\begin{gather}
\Tr(\bm{\Lambda}^\intercal_{k'; i'j'} \bm{\Lambda}_{k; ij} \bm{m}^0) = \delta_{k' k} \delta_{i' i} \frac{1}{\sqrt{d_k}} \Tr(\bm{\Lambda}_{k; j' j} \bm{m}^0) = m^0_{k; j'j} \\
\bm{m}^0_k = \begin{pmatrix}
m^0_{k; 11} \mbbm{1}_{d_k} & \dots & m^0_{k; 1n_k} \mbbm{1}_{d_k}\\
\vdots & \ddots & \vdots \\
m^0_{k; n_k1} \mbbm{1}_{d_k} & \dots & m^0_{k; n_k n_k} \mbbm{1}_{d_k}
\end{pmatrix}
\end{gather}
Where we can define a reduced \textit{k-irrep} block of $\bm{m}^0$:
\begin{gather}
\overline{\bm{m}}^0_k = \begin{pmatrix}
m^0_{k; 11} & \dots & m^0_{k; 1n_k}\\
\vdots & \ddots & \vdots \\
m^0_{k; n_k1} & \dots & m^0_{k; n_k n_k}
\end{pmatrix}
\end{gather}
Notice that $i$ can run through $1$ to $n_k$ in the last equation. Therefore,  the \textit{k-irrep} block of $\mbbm{F}$: $\mbbm{F}_k$ is just (upon reindexing $\bm{\Lambda}$) made of $n_k$ copies of $\bar{\bm{m}}^0_{k}$ within each irrep block:
\begin{gather}
\mbbm{F}_k = \begin{pmatrix}
\overline{\bm{m}}^0_k & \dots & 0 \\
\vdots & \ddots & \vdots \\
0 & \dots & \overline{\bm{m}}^0_k
\end{pmatrix}
\end{gather}
Which is positive definite since $\bm{m}^0$ is positive definite. \\

Now we can rewrite \eqref{eq:Bose} into an ordinary eigenvalue equation and define the transformed double-font matrices:
\begin{gather}
    \tilde{\mbbm{H}}^B \ket{\tilde{\bm{a}}} = E^B_g \ket{\tilde{\bm{a}}}, \quad \tilde{\mbbm{H}}^B \defeq \mbbm{F}^{-1/2} \mbbm{H}^B \mbbm{F}^{-1/2}, \quad \ket{\tilde{\bm{a}}} \defeq \mbbm{F}^{1/2} \ket{\bm{a}} \\
    (\tilde{\mbbm{R}}_{\alpha \beta})_{\nu \nu'} \defeq \sum_{\nu_1 \nu_2}\mbbm{F}^{-1/2}_{\nu \nu_1} (\mbbm{R}_{\alpha \beta})_{\nu_1 \nu_2} \mbbm{F}^{-1/2}_{\nu_2 \nu'}, \quad (\tilde{\mbbm{Q}}_{\alpha \beta})_{\nu \nu'} \defeq \sum_{\nu_1 \nu_2}\mbbm{F}^{-1/2}_{\nu \nu_1} (\mbbm{Q}_{\alpha \beta})_{\nu_1 \nu_2} \mbbm{F}^{-1/2}_{\nu_2 \nu'}, \quad \dots \label{def:tilde_dbl_matrix}
\end{gather}
Therefore, $\mathcal{J}^B_{ll'}$ can be solved by a simpler recipe of $\mathcal{J}^F_{ll'}$, the differences are: 1. no $\vb{k}$ summation; 2. local observables are evaluated with a single state $\ket{\tilde{\bm{a}}}$ instead of different energy levels filled up to Fermi energy. We write down the partial partial derivative for a general local observable $\hat{\mathcal{A}}_{\vb{R}}$ with respect to a translational invariant parameter $\gamma$:
\begin{gather}
    \expval{\hat{\mathcal{A}}_{\vb{R}}}_G = \bra{\bm{a}} \mbbm{A} \ket{\bm{a}} = \bra{\tilde{\bm{a}}} \tilde{\mbbm{A}} \ket{\tilde{\bm{a}}}, \quad \mbbm{A}_{\nu \nu'} \defeq \Tr[\bm{\Lambda}^\dagger_{\nu} \bm{A} \bm{\Lambda}_{\nu'} \bm{m}^0], \quad \tilde{\mbbm{A}} \defeq \mbbm{F}^{-1/2} \mbbm{A} \mbbm{F}^{-1/2} \\
    \partial_{\gamma} \expval{\hat{\mathcal{A}}_{\vb{R}}}_G = \sum_{\substack{e \\ (E^B_e > E^B_g)}}\frac{1}{E^B_g - E^B_e} \left(\bra{\tilde{\bm{a}}} \partial_{\gamma} \tilde{\mbbm{H}}^B \ket{\tilde{\bm{e}}}\bra{\tilde{\bm{e}}} \tilde{\mbbm{A}} \ket{\tilde{\bm{a}}} + h.c.\right) \label{eq:ppd_bose}
\end{gather}
In our experience, the ground state of $\tilde{\mbbm{H}}^B$ is generically non-degenerate, and any degeneracy that does arise can be attributed to unphysical parameter choices — though we currently lack a formal proof of this claim. In $\mathcal{J}^B_{ll'}$, the term $\partial_{\gamma} \tilde{\mbbm{H}}^B$ is replaced by $\mbbm{P}_l$ and $\mbbm{A}$ is replaced by $\mbbm{P}_{l'}$.

\section{Analytical Derivatives for $\mathcal{L}$} \label{appdx:dLdrho0}
The derivation of the section follows \cite{PENG2022108348}. Without loss of generality, we assume $\bm{\mathcal{R}}, \bm{\mathcal{Q}}, \bm{\chi}, \bm{\Upsilon}$ are real for simplicity as it is the case in our TBG study. As we briefly discussed in Appendix A, the partial partial derivatives (PPDs: $\ppfpx{}{}$) are defined by fixing $(\bm{\varrho}^0, \ket{\Phi_0})$ or $(\bm{\varrho}^0, \ket{\tilde{\bm{a}}})$ for variables in Fermi part \eqref{eq:Fermi} or Bose part \eqref{eq:Bose}:
\begin{align}
    \text{Fermi:} \quad& E_{\text{kin}}, \ \bm{\chi}, \ \bm{\Upsilon}, \ \bm{\lambda}^F \\
    \text{Bose:} \quad& E_{\text{atom}}, \ \bm{\mathcal{R}}, \ \bm{\mathcal{Q}}, \ \bm{\lambda}^B
\end{align}
For example, $\ppfpx{E_{\text{kin}}}{\bm{\lambda}^F} = \encvert{\pfpx{E_{\text{kin}}}{\bm{\lambda}^F}}{ \bm{\varrho}^0, \ket{\Phi_0}} = \encvert{\pfpx{E_{\text{kin}}}{\bm{\lambda}^F}}{ \bm{\varrho}^0, \bm{\mathcal{R}}, \bm{\mathcal{Q}}}$, where the last equality is because $\hat{H}^F$ is determined by $\bm{\mathcal{R}}$, $\bm{\mathcal{Q}}$ and $\bm{\varrho}^0$. In principle every PPDs can be calculated using \eqref{eq:ppd_bose} and \eqref{eq:ppd_bose}.  When Eq.~\eqref{eq:Fermi} and \eqref{eq:Bose} are solved individually, $\bm{\lambda}^F$ and $\bm{\lambda}^B$ are then determined. We can thus define variables at at \textbf{partial derivatives}(PDs: $\pfpx{}{}$) level as:
\begin{gather}
    \bm{\varrho}^0\ , \ \bm{\mathcal{R}}\ , \ \bm{\mathcal{Q}}\ , \ \bm{\chi}\ , \ \bm{\Upsilon}
\end{gather}
The total derivative (TD: $\dfdx{}{}$) of the energy functional then reads:
\begin{align}
    \dfdx{\mathcal{L}_G}{\bm{\varrho}^0} =& \pfpx{E_{\text{kin}}}{\bm{\varrho}^0} + \pfpx{E_{\text{kin}}}{\bm{\mathcal{R}}} \dfdx{\bm{\mathcal{R}}}{\bm{\varrho}^0} + \pfpx{E_{\text{kin}}}{\bm{\mathcal{Q}}} \dfdx{\bm{\mathcal{Q}}}{\bm{\varrho}^0} \\
        &+ \pfpx{E_{\text{atom}}}{\bm{\varrho}^0} + \pfpx{E_{\text{atom}}}{\bm{\chi}} \dfdx{\bm{\chi}}{\bm{\varrho}^0} + \pfpx{E_{\text{atom}}}{\bm{\Upsilon}} \dfdx{\bm{\Upsilon}}{\bm{\varrho}^0} 
\end{align}
\begin{itemize}
    \item{$\pfpx{E_{\text{kin}}}{\bm{\varrho}^0}$}: 
    \begin{align}
        E_{\text{kin}} =& \frac{1}{N_k} \sum_{\vb{k}} \Tr[\begin{pmatrix}
            \bm{a}_{\vb{k}} & \bm{b}_{\vb{k}} \\ \bm{b}_{\vb{k}} & -\bm{a}_{\vb{k}}
        \end{pmatrix} \bm{\varrho}^F_{\vb{k}}], \quad 
        \begin{pmatrix}
            \bm{a}_{\vb{k}} & \bm{b}_{\vb{k}} \\ \bm{b}_{\vb{k}} & -\bm{a}_{\vb{k}}
        \end{pmatrix} \defeq       
        \begin{pmatrix}
            \mathcal{R} & - \mathcal{Q}^* \\ \mathcal{Q} & \mathcal{R}^*
        \end{pmatrix} 
        \begin{pmatrix}
            \mathcal{H}_{\vb{k}} & 0 \\ 0 & -\mathcal{H}^*_{-\vb{k}}
        \end{pmatrix}
        \begin{pmatrix}
            \mathcal{R}^\dagger & \mathcal{Q}^\dagger \\
            -\mathcal{Q}^\intercal & \mathcal{R}^\intercal
        \end{pmatrix}  \\
        \pfpx{E_{\text{kin}}}{\bm{\varrho}^0} =& \ppfpx{E_{\text{kin}}}{\lambda^F} \pfpx{\lambda^F}{\bm{\varrho}^0}\\
        =& \ppfpx{E_{\text{kin}}}{\lambda^F} \left( \ppfpx{\bm{\varrho}^F}{\lambda^F} \right)^{-1} \\
        \ppfpx{E_{\text{kin}}}{\lambda^F} =& \frac{1}{N_k} \sum_{\vb{k}} \Tr[\begin{pmatrix}
            \bm{a}_{\vb{k}} & \bm{b}_{\vb{k}} \\ \bm{b}_{\vb{k}} & -\bm{a}_{\vb{k}}
        \end{pmatrix} \left(\ppfpx{\bm{\varrho}^F_{\vb{k}}}{\lambda^F} \right)] 
    \end{align}
    \item {$\pfpx{E_{\text{kin}}}{\bm{\mathcal{R}}}$}: 
    \begin{align}
        \pfpx{E_{\text{kin}}}{\bm{\mathcal{R}}} =& \frac{1}{N_k} \sum_{\vb{k}} \Tr[\pfpx{}{\bm{\mathcal{R}}} \begin{pmatrix}
            \bm{a}_{\vb{k}} & \bm{b}_{\vb{k}} \\ \bm{b}_{\vb{k}} & -\bm{a}_{\vb{k}}
        \end{pmatrix} \bm{\varrho}^F_{\vb{k}}] + \Tr[ \begin{pmatrix}
            \bm{a}_{\vb{k}} & \bm{b}_{\vb{k}} \\ \bm{b}_{\vb{k}} & -\bm{a}_{\vb{k}}
        \end{pmatrix} \pfpx{\bm{\varrho}^F_{\vb{k}}}{\bm{\mathcal{R}}}] \\
        =& \frac{1}{N_k} \sum_{\vb{k}} \Tr[\ppfpx{}{\bm{\mathcal{R}}} \begin{pmatrix}
            \bm{a}_{\vb{k}} & \bm{b}_{\vb{k}} \\ \bm{b}_{\vb{k}} & -\bm{a}_{\vb{k}}
        \end{pmatrix} \bm{\varrho}^F_{\vb{k}}] + \Tr[ \begin{pmatrix}
            \bm{a}_{\vb{k}} & \bm{b}_{\vb{k}} \\ \bm{b}_{\vb{k}} & -\bm{a}_{\vb{k}}
        \end{pmatrix} \pfpx{\bm{\varrho}^F_{\vb{k}}}{\bm{\mathcal{R}}}] \\
        =& \bm{\chi} + \frac{1}{N_k} \sum_{\vb{k}} \Tr[ \begin{pmatrix}
            \bm{a}_{\vb{k}} & \bm{b}_{\vb{k}} \\ \bm{b}_{\vb{k}} & -\bm{a}_{\vb{k}}
        \end{pmatrix} \left(\ppfpx{\bm{\varrho}^F_{\vb{k}}}{\bm{\mathcal{R}}} + \ppfpx{\bm{\varrho}^F_{\vb{k}}}{\bm{\eta}} \pfpx{\bm{\eta}}{\bm{\mathcal{R}}} + \ppfpx{\bm{\varrho}^F_{\vb{k}}}{\bm{\xi}}\pfpx{\bm{\xi}}{\bm{\mathcal{R}}} \right)] \\
        \ppfpx{\bm{\varrho_{k}}^F}{\bm{\mathcal{R}}} =& \Tr[ \hat{\bm{\varrho}}^F_{\vb{k}}\ppfpx{}{\bm{\mathcal{R}}}\left(\sum_{p \in \text{all levels}} f_{\vb{k} p} \ket{\vb{k} p}\bra{\vb{k} p}\right)]
    \end{align}
    \item{$\pfpx{E_{\text{kin}}}{\bm{\mathcal{Q}}}$}: The calculation is similar to $\pfpx{E_{\text{kin}}}{\bm{\mathcal{R}}}$, but the result is:
    \begin{align}
        \pfpx{E_{\text{kin}}}{\bm{\mathcal{Q}}}=& \bm{\Upsilon} + \frac{1}{N_k} \sum_{\vb{k}} \Tr[ \begin{pmatrix}
                    \bm{a}_{\vb{k}} & \bm{b}_{\vb{k}} \\ \bm{b}_{\vb{k}} & -\bm{a}_{\vb{k}}
                \end{pmatrix} \left(\ppfpx{\bm{\varrho}^F_{\vb{k}}}{\bm{\mathcal{Q}}} + \ppfpx{\bm{\varrho}^F_{\vb{k}}}{\bm{\eta}} \pfpx{\bm{\eta}}{\bm{\mathcal{Q}}} + \ppfpx{\bm{\varrho}^F_{\vb{k}}}{\bm{\xi}}\pfpx{\bm{\xi}}{\bm{\mathcal{Q}}} \right)]
    \end{align}
    \item{$\dfdx{\bm{\mathcal{R}}}{\bm{\varrho}^0}$}: This step uses the Gutzwiller SCF condition $\delta \bm{\mathcal{R}}_{n+1} = \delta \bm{\mathcal{R}}_{n}$:
    \begin{align}
        \delta \bm{\mathcal{R}} =& \pfpx{\bm{\mathcal{R}}}{\bm{\chi}} \delta \bm{\chi} + \pfpx{\bm{\mathcal{R}}}{\bm{\Upsilon}} \delta \bm{\Upsilon} + \pfpx{\bm{\mathcal{R}}}{\bm{\varrho}^0} \delta \bm{\varrho}^0 \\
        \delta \bm{\chi} =& \pfpx{\bm{\chi}}{\bm{\mathcal{R}}} \delta \bm{\mathcal{R}} + \pfpx{\bm{\chi}}{\bm{\mathcal{Q}}} \delta \bm{\mathcal{Q}} + \pfpx{\bm{\chi}}{\bm{\varrho}^0} \delta \bm{\varrho}^0 \\
        \delta \bm{\mathcal{Q}} =& \pfpx{\bm{\mathcal{Q}}}{\bm{\chi}} \delta \bm{\chi} + \pfpx{\bm{\mathcal{Q}}}{\bm{\Upsilon}} \delta \bm{\Upsilon} + \pfpx{\bm{\mathcal{Q}}}{\bm{\varrho}^0} \delta \bm{\varrho}^0 \\
        \delta \bm{\Upsilon} =& \pfpx{\bm{\Upsilon}}{\bm{\mathcal{R}}} \delta \bm{\mathcal{R}} + \pfpx{\bm{\Upsilon}}{\bm{\mathcal{Q}}} \delta \bm{\mathcal{Q}} + \pfpx{\bm{\Upsilon}}{\bm{\varrho}^0} \delta \bm{\varrho}^0
    \end{align}
    This can be written as set of linear equations:
    \begin{gather}
        \left[\begin{pmatrix} \bm{1} & \\ & \bm{1} \end{pmatrix} - \begin{pmatrix}
        \pfpx{\bm{\mathcal{R}}}{\bm{\chi}} \pfpx{\bm{\chi}}{\bm{\mathcal{R}}} + \pfpx{\bm{\mathcal{R}}}{\bm{\Upsilon}} \pfpx{\bm{\Upsilon}}{\bm{\mathcal{R}}} & \pfpx{\bm{\mathcal{R}}}{\bm{\chi}} \pfpx{\bm{\chi}}{\bm{\mathcal{Q}}} + \pfpx{\bm{\mathcal{R}}}{\bm{\Upsilon}} \pfpx{\bm{\Upsilon}}{\bm{\mathcal{Q}}} \\
        \pfpx{\bm{\mathcal{Q}}}{\bm{\chi}} \pfpx{\bm{\chi}}{\bm{\mathcal{R}}} + \pfpx{\bm{\mathcal{Q}}}{\bm{\Upsilon}} \pfpx{\bm{\Upsilon}}{\bm{\mathcal{R}}} & \pfpx{\bm{\mathcal{Q}}}{\bm{\chi}} \pfpx{\bm{\chi}}{\bm{\mathcal{Q}}} + \pfpx{\bm{\mathcal{Q}}}{\bm{\Upsilon}} \pfpx{\bm{\Upsilon}}{\bm{\mathcal{Q}}}
        \end{pmatrix} \right] \begin{pmatrix}
            \delta \bm{\mathcal{R}} \\ \delta \bm{\mathcal{Q}}
        \end{pmatrix} = \begin{pmatrix}
            \pfpx{\bm{\mathcal{R}}}{\bm{\chi}} \pfpx{\bm{\chi}}{\bm{\varrho}^0} + \pfpx{\bm{\mathcal{R}}}{\bm{\Upsilon}} \pfpx{\bm{\Upsilon}}{\bm{\varrho}^0} + \pfpx{\bm{\mathcal{R}}}{\bm{\varrho}^0} \\
            \pfpx{\bm{\mathcal{Q}}}{\bm{\chi}} \pfpx{\bm{\chi}}{\bm{\varrho}^0} + \pfpx{\bm{\mathcal{Q}}}{\bm{\Upsilon}} \pfpx{\bm{\Upsilon}}{\bm{\varrho}^0} + \pfpx{\bm{\mathcal{Q}}}{\bm{\varrho}^0}
        \end{pmatrix} \delta \bm{\varrho}^0 \\
        \Rightarrow \begin{pmatrix}
            \dfdx{\bm{\mathcal{R}}}{\bm{\varrho}^0} \\ \dfdx{\bm{\mathcal{Q}}}{\bm{\varrho}^0}
        \end{pmatrix} = \left[\begin{pmatrix} \bm{1} & \\ & \bm{1} \end{pmatrix} - \begin{pmatrix}
            \pfpx{\bm{\mathcal{R}}}{\bm{\chi}} \pfpx{\bm{\chi}}{\bm{\mathcal{R}}} + \pfpx{\bm{\mathcal{R}}}{\bm{\Upsilon}} \pfpx{\bm{\Upsilon}}{\bm{\mathcal{R}}} & \pfpx{\bm{\mathcal{R}}}{\bm{\chi}} \pfpx{\bm{\chi}}{\bm{\mathcal{Q}}} + \pfpx{\bm{\mathcal{R}}}{\bm{\Upsilon}} \pfpx{\bm{\Upsilon}}{\bm{\mathcal{Q}}} \\
            \pfpx{\bm{\mathcal{Q}}}{\bm{\chi}} \pfpx{\bm{\chi}}{\bm{\mathcal{R}}} + \pfpx{\bm{\mathcal{Q}}}{\bm{\Upsilon}} \pfpx{\bm{\Upsilon}}{\bm{\mathcal{R}}} & \pfpx{\bm{\mathcal{Q}}}{\bm{\chi}} \pfpx{\bm{\chi}}{\bm{\mathcal{Q}}} + \pfpx{\bm{\mathcal{Q}}}{\bm{\Upsilon}} \pfpx{\bm{\Upsilon}}{\bm{\mathcal{Q}}}
            \end{pmatrix} \right] ^{-1} \begin{pmatrix}
            \pfpx{\bm{\mathcal{R}}}{\bm{\chi}} \pfpx{\bm{\chi}}{\bm{\varrho}^0} + \pfpx{\bm{\mathcal{R}}}{\bm{\Upsilon}} \pfpx{\bm{\Upsilon}}{\bm{\varrho}^0} + \pfpx{\bm{\mathcal{R}}}{\bm{\varrho}^0} \\
            \pfpx{\bm{\mathcal{Q}}}{\bm{\chi}} \pfpx{\bm{\chi}}{\bm{\varrho}^0} + \pfpx{\bm{\mathcal{Q}}}{\bm{\Upsilon}} \pfpx{\bm{\Upsilon}}{\bm{\varrho}^0} + \pfpx{\bm{\mathcal{Q}}}{\bm{\varrho}^0}
        \end{pmatrix} 
    \end{gather} 
    The expression can be simplified by group $(\mathcal{R}, \mathcal{Q}) \Rightarrow \MC{R}$ and $(\chi, \Upsilon) \Rightarrow \chi$:
    \begin{gather*}
        \dfdx{\MC{R}}{\bm{\varrho}^0} = \left(\bm{1} - \pfpx{\MC{R}}{\chi} \pfpx{\chi}{\MC{R}}\right)^{-1} \left(\pfpx{\MC{R}}{\chi}\pfpx{\chi}{\bm{\varrho}^0} + \pfpx{\MC{R}}{\bm{\varrho}^0}\right) \\
        \dfdx{\chi}{\bm{\varrho}^0} = \left(\bm{1} - \pfpx{\chi}{\MC{R}} \pfpx{\MC{R}}{\chi}\right)^{-1} \left(\pfpx{\chi}{\MC{R}}\pfpx{\MC{R}}{\bm{\varrho}^0} + \pfpx{\chi}{\bm{\varrho}^0}\right)
    \end{gather*}

    \item{$\pfpx{\chi}{\MC{R}}$}: 
    \begin{align*}
        \pfpx{\chi}{\MC{R}} =& \ppfpx{\chi}{\MC{R}} + \ppfpx{\chi}{\bm{\lambda}^F} \pfpx{\bm{\lambda}^F}{\MC{R}} \\
        \pfpx{\chi}{\bm{\varrho}^0} =& \ppfpx{\chi}{\lambda^F}\pfpx{\lambda^F}{\bm{\varrho}^0} = \ppfpx{\chi}{\lambda^F}\left( \ppfpx{\bm{\varrho}^F}{\lambda^F} \right)^{-1}
    \end{align*}

    \item{$\pfpx{\MC{R}}{\chi}$}:
    \begin{align*}
        \pfpx{\MC{R}}{\chi} =& \ppfpx{\MC{R}}{\chi} + \ppfpx{\MC{R}}{\lambda^B} \pfpx{\lambda^B}{\chi} \\
        =& \ppfpx{\MC{R}}{\chi} - \ppfpx{\MC{R}}{\lambda^B} \left( \ppfpx{\bm{\varrho}^B}{\bm{\lambda}^B}\right)^{-1} \ppfpx{\bm{\varrho}^B}{\bm{\chi}} \Leftarrow
        \encvert{\delta \bm{\varrho}^B}{\delta \bm{\varrho}^0 = 0} = \ppfpx{\bm{\varrho}^B}{\bm{\lambda}^B} \delta \bm{\lambda}^B + \ppfpx{\bm{\varrho}^B}{\bm{\chi}} \delta \bm{\chi} = 0 \\
        \ppfpx{\MC{R}}{\lambda^B_{\alpha}} =& 2 \sum_{e \neq a} \frac{1}{E_a - E_e} \bra{\tilde{a}} \tilde{\mbbm{N}}_{\alpha} \ket{\tilde{e}} \bra{\tilde{e}} \tilde{\mbbm{R}} \ket{\tilde{a}}
    \end{align*}

    \item{$\pfpx{\MC{R}}{\bm{\varrho}^0}$}:
    \begin{align*}
        \pfpx{\MC{R}}{\bm{\varrho}^0} =& \ppfpx{\MC{R}}{\bm{\varrho}^0} + \ppfpx{\MC{R}}{\lambda^B} \pfpx{\lambda^B}{\bm{\varrho}^0} \\
        \ppfpx{\MC{R}}{\bm{\varrho}^0} =& \bra{\tilde{a}} \ppfpx{\tilde{\mbbm{R}}}{\bm{\varrho}^0} \ket{\tilde{a}} +  \sum_e \frac{1}{E_a - E_e} \left(\bra{\tilde{a}} \ppfpx{\tilde{\mbbm{H}}^B}{\bm{\varrho}^0} \ket{\tilde{e}} \bra{\tilde{e}} \tilde{\mbbm{R}} \ket{\tilde{a}} + h.c. \right)
    \end{align*}

    \item{$\pfpx{E_{\text{atom}}}{\bm{\varrho}^0}$}: 
    \begin{align*}
        \pfpx{E_{\text{atom}}}{\bm{\varrho}^0} =& \ppfpx{E_{\text{atom}}}{\bm{\varrho}^0} + \ppfpx{E_{\text{atom}}}{\bm{\lambda}^B} \pfpx{\bm{\lambda}^B}{\bm{\varrho}^0} \\
        \ppfpx{E_{\text{atom}}}{\bm{\varrho}^0} =& \bra{\tilde{a}} \ppfpx{\tilde{\mbbm{H}}_{\text{atom}}}{\bm{\varrho}^0} \ket{\tilde{a}} + 2 \sum_{e \neq a} \frac{1}{E_a - E_e} \bra{\tilde{a}} \ppfpx{\tilde{\mbbm{H}}^B}{\bm{\varrho}^0} \ket{\tilde{e}} \bra{\tilde{e}} \tilde{\mbbm{H}}_{\text{atom}} \ket{\tilde{a}} \\
        \ppfpx{E_{\text{atom}}}{\bm{\lambda}^B_{\alpha}} =& 2 \sum_{e \neq a} \frac{1}{E_a - E_e} \bra{\tilde{a}} \tilde{\mbbm{N}}_{\alpha} \ket{\tilde{e}} \bra{\tilde{e}} \tilde{\mbbm{H}}_{\text{atom}} \ket{\tilde{a}} \\
        \pfpx{\bm{\lambda}^B}{\bm{\varrho}^0} =& \left(\ppfpx{\bm{\varrho}^B}{\bm{\lambda}^B}\right)^{-1} \left(\bm{1} - \ppfpx{\bm{\varrho}^B}{\bm{\varrho}^0}\right) \Leftarrow
        \encvert{\delta \bm{\varrho}^B}{\delta \bm{\chi} = 0} = \ppfpx{\bm{\varrho}^B}{\bm{\lambda}^B} \delta \bm{\lambda}^B + \ppfpx{\bm{\varrho}^B}{\bm{\varrho}^0} \delta \bm{\varrho}^0 = \delta \bm{\varrho}^0 \\
        \ppfpx{\bm{\varrho}^B}{\bm{\varrho}^0} =& \bra{\tilde{a}} \ppfpx{\tilde{\mbbm{N}}}{\bm{\varrho}^0} \ket{\tilde{a}} +  \sum_e \frac{1}{E_a - E_e} \left(\bra{\tilde{a}} \ppfpx{\tilde{\mbbm{H}}^B}{\bm{\varrho}^0} \ket{\tilde{e}} \bra{\tilde{e}} \tilde{\mbbm{N}} \ket{\tilde{a}} + h.c. \right)
    \end{align*}
    \item{$\pfpx{E_{\text{atom}}}{\bm{\chi}}$}:
    \begin{align*}
        \pfpx{E_{\text{atom}}}{\bm{\chi}} =& \ppfpx{E_{\text{atom}}}{\bm{\chi}} + \ppfpx{E_{\text{atom}}}{\bm{\lambda}^B} \pfpx{\bm{\lambda}^B}{\bm{\chi}} \\
        \ppfpx{E_{\text{atom}}}{\bm{\chi}_{\alpha}} =& 2 \sum_{e \neq a} \frac{1}{E_a - E_e} \bra{\tilde{a}} \tilde{\mbbm{R}}_{\alpha} \ket{\tilde{e}} \bra{\tilde{e}} \tilde{\mbbm{H}}_{\text{atom}} \ket{\tilde{a}} \\
        \pfpx{\bm{\lambda}^B}{\bm{\chi}} =& -\left(\ppfpx{\bm{\varrho}^B}{\bm{\lambda}^B}\right)^{-1} \left(\ppfpx{\bm{\varrho}^B}{\bm{\chi}}\right) \Leftarrow 
        \encvert{\delta \bm{\varrho}^B}{\delta \bm{\varrho}^0 = 0} = \ppfpx{\bm{\varrho}^B}{\bm{\lambda}^B} \delta \bm{\lambda}^B + \ppfpx{\bm{\varrho}^B}{\bm{\chi}} \delta \bm{\chi} =0 \\
        \ppfpx{\bm{\varrho}^B}{\bm{\chi}_{\alpha}} =& 2 \sum_{e \neq a} \frac{1}{E_a - E_e} \bra{\tilde{a}} \tilde{\mbbm{R}}_{\alpha} \ket{\tilde{e}} \bra{\tilde{e}} \tilde{\mbbm{N}} \ket{\tilde{a}}
    \end{align*}
\end{itemize}

\subsection{Bose Part Derivatives}
The PPDs in Bose part needs extra attention as they're a bit more complicated then the \textit{natural basis} version derived in \cite{PENG2022108348}. For example, the PPDs for $\widetilde{\mathbb{O}}$ defined in \eqref{def:tilde_dbl_matrix}:
\begin{gather}
    \ppfpx{\widetilde{\mathbb{O}}}{\bm{\varrho}^0} = \ppfpx{\mathbb{F}^{-1/2}}{\bm{\varrho}^0} \mathbb{O} \mathbb{F}^{-1/2} + \mathbb{F}^{-1/2} \ppfpx{\mathbb{O}}{\bm{\varrho}^0} \mathbb{F}^{-1/2} + \mathbb{F}^{-1/2} \mathbb{O} \ppfpx{\mathbb{F}^{-1/2}}{\bm{\varrho}^0} \\
\end{gather}
Where $\ppfpx{\mathbb{F}^{-1/2}}{\bm{\varrho}^0}$ can be obtained by solving the following Sylvester equation using the scipy package:
\begin{gather}
    \ppfpx{\mathbb{F}^{-1/2}}{\bm{\varrho}^0} \mathbb{F}^{1/2} + \mathbb{F}^{1/2} \ppfpx{\mathbb{F}^{-1/2}}{\bm{\varrho}^0} + \mathbb{F}^{-1/2} \ppfpx{\mathbb{F}}{\bm{\varrho}^0} \mathbb{F}^{-1/2} = 0 \label{eq:dblF_sylvester}
\end{gather}
And the PPDs for any double font matrix $\mathbb{O}$ is calculated from the derivatives of $\bm{m}^0$:
\begin{gather}
    \ppfpx{\mathbb{F}_{\mu \nu}}{\bm{\varrho}^0} = \Tr(\bm{\Gamma}^\intercal_{\mu} \bm{\Gamma}_{\nu} \dfdx{\bm{m}^0}{\bm{\varrho}^0})
\end{gather}
which is symmetric from its construction. The general Sylvester equation :
\begin{gather}
    AX+XB=C
\end{gather}
has a unique solution $X$ when $A$ and $-B$ don't share common eigenvalues. In our case, $A=B=\mbbm{F}^{1/2}$ are both positive definite matrices, so the solution is unique. Transposing and conjugating \eqref{eq:dblF_sylvester} gives us:
\begin{gather}
    \left(\ppfpx{\mathbb{F}^{-1/2}}{\bm{\varrho}^0}\right)^\dagger \mathbb{F}^{1/2} + \mathbb{F}^{1/2} \left(\ppfpx{\mathbb{F}^{-1/2}}{\bm{\varrho}^0}\right)^\dagger + \mathbb{F}^{-1/2} \ppfpx{\mathbb{F}}{\bm{\varrho}^0} \mathbb{F}^{-1/2} = 0.
\end{gather}
With the uniqueness of the solution, we can conclude that:
\begin{gather}
    \left(\ppfpx{\mathbb{F}^{-1/2}}{\bm{\varrho}^0}\right)^\dagger = \ppfpx{\mathbb{F}^{-1/2}}{\bm{\varrho}^0}
\end{gather}
$\bm{m}^0$ is calculated in \eqref{eq:m0_cal} and \eqref{def:M_and_P}:
\begin{gather}
    \bm{m}^0_{I I'} \defeq \bra{\Phi_0} \ket{I'} \bra{I} \ket{\Phi_0} = (-1)^{\abs{\overline{(I \cup I')}}}\det(\mathcal{M}^{I' I}_{\rho}), 
\end{gather}
Where $\mathcal{M}^{I' I}_{\bm{\varrho}^0}$ is the reduced density matrix from $\bm{\varrho}^0$ by only keeping the orbital indices from $I'$ and $I$. $\bm{\varrho}^0$ is expanded by a set of matrix basis $\{ \mathcal{O}_l \}$:
\begin{gather}
    \bm{\varrho}^0 = \sum_{l} \varrho^0_l \mathcal{O}_l
\end{gather}
Then the derivatives of $\mathcal{M}^{I' I}_{\bm{\varrho}^0}$ is simply:
\begin{gather}
    \dfdx{\mathcal{M}^{I' I}_{\bm{\varrho}^0}}{\varrho^0_l} = \begin{pmatrix}
        \mathcal{P}^{I' I}_{\mathcal{O}_l} & \mathcal{P}_{\mathcal{O}_l}^{I' \overline{(I \cup I')}} \\
        \mathcal{P}_{\mathcal{O}_l}^{\overline{(I \cup I')} I} & \mathcal{P}_{\mathcal{O}_l}^{\overline{(I \cup I')} \overline{(I \cup I')}}
        \end{pmatrix}
\end{gather}
The derivatives of a determinant is given by Jacobi's formula:
\begin{align*}
    \dfdx{\det(\mathcal{M}^{I' I}_{\bm{\varrho}^0})}{\varrho^0_l} =& \Tr[\text{adj}(\mathcal{M}^{I' I}_{\bm{\varrho}^0}) \dfdx{\mathcal{M}^{I' I}_{\bm{\varrho}^0}}{\varrho^0_l}] \\
    \text{adj}(\mathcal{M}^{I' I}_{\bm{\varrho}^0})_{\alpha \beta} =& \mathcal{C}(\mathcal{M}^{I' I}_{\bm{\varrho}^0})_{\beta \alpha}
\end{align*}
The final result is:
\begin{gather}
    \dfdx{\bm{m}^0_{I I'}}{\varrho^0_l} = (-1)^{\abs{\overline{\left(I' \cup I\right)}}} \Tr[\text{adj}(\mathcal{M}^{I' I}_{\bm{\varrho}^0}) \dfdx{\mathcal{M}^{I' I}_{\bm{\varrho}^0}}{\varrho^0_l}] 
\end{gather}
Calculation of $\text{adj}(\mathcal{M}^{I' I}_{\bm{\varrho}^0})$ takes a lot of time, we're yet to find the most efficient way to do it. However, we can use the following property to reduce the computational cost by half:
\begin{gather*}
    \mathcal{C}(\mathcal{M}^{I' I}_{\bm{\varrho}^0})_{\beta \alpha} = \mathcal{C}(\mathcal{M}^{I I'}_{\bm{\varrho}^0})^*_{\alpha \beta}
\end{gather*}

\subsubsection{$\pfpx{E_{\text{atom}}}{\varrho^0}$}
Using derivative rules, we have:
\begin{gather}
    \pfpx{E_{\text{atom}}}{\bm{\varrho}^0} = \ppfpx{E_{\text{atom}}}{\bm{\lambda}^B} \pfpx{\bm{\lambda}^B}{\bm{\varrho}^0} + \ppfpx{E_{\text{atom}}}{\bm{\varrho}^0}
\end{gather}
PPDs and special PDs:
\begin{align}
    &\ppfpx{E_{\text{atom}}}{\bm{\lambda}^B} = \sum_e \frac{1}{E_a - E_e} \left(\bra{\tilde{a}} \ppfpx{\tilde{\mbbm{H}}^B}{\bm{\lambda}^B} \ket{\tilde{e}} \bra{\tilde{e}} \tilde{\mbbm{H}}_{\text{atom}} \ket{\tilde{a}} + h.c. \right) \\
    &\pfpx{\bm{\lambda}^B}{\bm{\varrho}^0} = \left(\ppfpx{\bm{\varrho}^B}{\bm{\lambda}^B}\right)^{-1} \left(\bm{1} - \ppfpx{\bm{\varrho}^B}{\bm{\varrho}^0}\right) \\
    &\ppfpx{E_{\text{atom}}}{\bm{\varrho}^0} = \bra{\tilde{a}} \ppfpx{\tilde{\mbbm{H}}_{\text{atom}}}{\bm{\varrho}^0} \ket{\tilde{a}} +  \sum_e \frac{1}{E_a - E_e} \left(\bra{\tilde{a}} \ppfpx{\tilde{\mbbm{H}}^B}{\bm{\varrho}^0} \ket{\tilde{e}} \bra{\tilde{e}} \tilde{\mbbm{H}}_{\text{atom}} \ket{\tilde{a}} + h.c. \right) \\
    &\ppfpx{\bm{\varrho}^B}{\bm{\lambda}^B} = \sum_e \frac{1}{E_a - E_e} \left(\bra{\tilde{a}} \ppfpx{\tilde{\mbbm{H}}^B}{\bm{\lambda}^B} \ket{\tilde{e}} \bra{\tilde{e}} \tilde{\mbbm{N}} \ket{\tilde{a}} + h.c. \right) \\
    &\ppfpx{\bm{\varrho}^B}{\bm{\varrho}^0} = \bra{\tilde{a}} \ppfpx{\tilde{\mbbm{N}}}{\bm{\varrho}^0} \ket{\tilde{a}} +  \sum_e \frac{1}{E_a - E_e} \left(\bra{\tilde{a}} \ppfpx{\tilde{\mbbm{H}}^B}{\bm{\varrho}^0} \ket{\tilde{e}} \bra{\tilde{e}} \tilde{\mbbm{N}} \ket{\tilde{a}} + h.c. \right)
\end{align}
PPDs of double font matrices:
\begin{align}
    & \ppfpx{\tilde{\mbbm{H}}^B}{\bm{\varrho}^0} = \sum_{\alpha} \left(\chi_{\alpha}\ppfpx{\tilde{\mbbm{R}}_{\alpha}}{\bm{\varrho}^0} + \Gamma_{\alpha}\ppfpx{\tilde{\mbbm{Q}}_{\alpha}}{\bm{\varrho}^0} + \lambda^B_{\alpha} \ppfpx{\tilde{\mbbm{N}}_{\alpha}}{\bm{\varrho}^0} \right) + \ppfpx{\tilde{\mbbm{H}}_{\text{atom}}}{\bm{\varrho}^0} \\
    & \ppfpx{\tilde{\mbbm{R}}_{\alpha}}{\bm{\varrho}^0} = \ppfpx{\mathbb{F}^{-1/2}}{\bm{\varrho}^0} \mbbm{R}_{\alpha} \mathbb{F}^{-1/2} + \mathbb{F}^{-1/2} \ppfpx{\mbbm{R}_{\alpha}}{\bm{\varrho}^0} \mathbb{F}^{-1/2} + \mathbb{F}^{-1/2} \mbbm{R}_{\alpha} \ppfpx{\mathbb{F}^{-1/2}}{\bm{\varrho}^0} \\
    & \ppfpx{\tilde{\mbbm{Q}}_{\alpha}}{\bm{\varrho}^0} = \ppfpx{\mathbb{F}^{-1/2}}{\bm{\varrho}^0} \mbbm{Q}_{\alpha} \mathbb{F}^{-1/2} + \mathbb{F}^{-1/2} \ppfpx{\mbbm{Q}_{\alpha}}{\bm{\varrho}^0} \mathbb{F}^{-1/2} + \mathbb{F}^{-1/2} \mbbm{Q}_{\alpha} \ppfpx{\mathbb{F}^{-1/2}}{\bm{\varrho}^0} \\
    & \ppfpx{}{\bm{\varrho}^0} \begin{pmatrix}
        \mbbm{R} \\ \mbbm{Q}
    \end{pmatrix} = \left[-(\bm{\varrho}^0)^{-1} \ppfpx{\bm{\varrho}^0}{\bm{\varrho}^0} (\bm{\varrho}^0)^{-1} + (\bm{\varrho}^0)^{-1} \ppfpx{}{\bm{\varrho}^0}\right] \frac{1}{2} \begin{pmatrix}
        \mbbm{K}_{\uparrow \uparrow} + (\mbbm{K}_{\uparrow \uparrow})^\dagger \\ \mbbm{K}_{\downarrow \uparrow} + (\mbbm{K}_{\downarrow \uparrow})^\dagger \end{pmatrix}\\
    & \ppfpx{\tilde{\mbbm{H}}_{\text{atom}}}{\bm{\varrho}^0} = \ppfpx{\mathbb{F}^{-1/2}}{\bm{\varrho}^0} \mbbm{H}_{\text{atom}} \mathbb{F}^{-1/2} + \mathbb{F}^{-1/2} \ppfpx{\mbbm{H}_{\text{atom}}}{\bm{\varrho}^0} \mathbb{F}^{-1/2} + \mathbb{F}^{-1/2} \mbbm{H}_{\text{atom}} \ppfpx{\mathbb{F}^{-1/2}}{\bm{\varrho}^0} \\
    & \ppfpx{\tilde{\mbbm{N}}_{\alpha}}{\bm{\varrho}^0} = \ppfpx{\mathbb{F}^{-1/2}}{\bm{\varrho}^0} \mbbm{N}_{\alpha} \mathbb{F}^{-1/2} + \mathbb{F}^{-1/2} \ppfpx{\mbbm{N}_{\alpha}}{\bm{\varrho}^0} \mathbb{F}^{-1/2} + \mathbb{F}^{-1/2} \mbbm{N}_{\alpha} \ppfpx{\mathbb{F}^{-1/2}}{\bm{\varrho}^0} \\
    & \ppfpx{\tilde{\mbbm{D}}_{\alpha}}{\bm{\varrho}^0} = \ppfpx{\mathbb{F}^{-1/2}}{\bm{\varrho}^0} \mbbm{D}_{\alpha} \mathbb{F}^{-1/2} + \mathbb{F}^{-1/2} \ppfpx{\mbbm{D}_{\alpha}}{\bm{\varrho}^0} \mathbb{F}^{-1/2} + \mathbb{F}^{-1/2} \mbbm{D}_{\alpha} \ppfpx{\mathbb{F}^{-1/2}}{\bm{\varrho}^0} 
\end{align}

\section{Implementation to TBG}\label{appdx:TBG_implementation}
\providecommand{\cref}[1]{\ref{#1}}
\subsection{8-band model and interactions}
The 8-band tight-binding model of TBG reads:
\begin{gather}
    \hat{H}_0 = \sum_{s \eta \vb{k}}
    \begin{pmatrix}
        \bm{f}^\dagger_{s \eta \vb{k}} & \bm{c}^\dagger_{s \eta \vb{k}}
    \end{pmatrix}
    \begin{pmatrix}
        \hat{H}^{(0, \text{ff})}_{\eta}(\vb{k}) & \hat{H}^{(0, \text{fc})}_{\eta}(\vb{k}) \\
        \hat{H}^{(0, \text{fc})}_{\eta}(\vb{k})^\dagger & \hat{H}^{(0, \text{cc})}_{\eta}(\vb{k}) 
    \end{pmatrix}
    \begin{pmatrix}
        \bm{f}_{s \eta \vb{k}} \\ \bm{c}_{s \eta \vb{k}}
    \end{pmatrix}.
\end{gather}
where $s, \eta$ are spin, valley indices, $\vb{k}$ marks the crystal momentum in the morie Brillouin zone. There 2 orbitals for $f_{s \eta \alpha; \vb{R}}$ ($\alpha$=1, 2: orbital indices for $f$-orbitals) located at the AA-stacking center and 6 orbitals for $\bm{c}_{s \eta \vb{R}}$ located at AA, AB/BA and DW (Domain Wall) regions. The details of the Hamiltonian and symmetry operations can be found in Ref.~\cite{HShi2025_PRB}. We list the symmetry operations of the $f$-orbitals here where we adopt the convention in Ref.~\cite{YJWANG2024_PRL} (The $f$-orbitals in Ref.~\cite{YJWANG2024_PRL} and Ref.~\cite{HShi2025_PRB} are exactly the same $p_{\pm}$ Gaussian orbtials; the site index $\vb{R}$ is omitted for simplicity):
\begin{gather}
    \mathcal{T} f^\dagger_{s \eta \alpha} \mathcal{T}^{-1} = f^\dagger_{s \bar{\eta} \alpha}, \quad C_{2z} f^\dagger_{s \eta \alpha} C_{2z}^{-1} = f^\dagger_{s \bar{\eta} \bar{\alpha}}, \quad C_{3z} f^\dagger_{s \eta \alpha} C_{3z}^{-1} = e^{i \frac{2\pi}{3} \eta (-1)^{(\alpha-1)}} f^\dagger_{s \eta \alpha}, \quad C_{2x} f^\dagger_{s \eta \alpha} C_{2x}^{-1} = f^\dagger_{s \eta \bar{\alpha}}. \label{def:symm_operation_f}
\end{gather}
The interactions are density density $ff$-interactions described by $\hat{H}_U$, the density-density $fc$-interactions described by $\hat{H}_V$ and the density-density $cc$-interactions described by $\hat{H}_W$:
\begin{gather}
    \hat{N}_f(\vb{R}) = \sum_{s \eta \alpha} 
    \hat{n}^f_{\vb{R}; s \eta \alpha} , \quad \hat{N}_c(\vb{R}) = \sum_{s \eta a} \hat{n}^c_{\vb{R}; s \eta a} \\
    \hat{H}_U = \frac{U}{2} \sum_{\vb{R}} \left(\hat{N}_f(\vb{R}) - 4\right)^2, \quad \hat{H}_W = W \sum_{\vb{R}} \left(\hat{N}_f(\vb{R}) - 4\right) \left(\hat{N}_c - 12\right), \quad \hat{H}_V = \frac{V}{2} \sum_{\vb{R}}\left(\hat{N}_c - 12\right)^2
\end{gather}
Morie optical phonon and microscopic carbon-atom Hubbard induce anti-Hund's ($J_A$) and Hund's ($J_H$) splittings \cite{YJWang2025_PRB}, 
\begin{equation}
\begin{aligned}
    \hat{H}_{J_A} + \hat{H}_{J_H} &= -\frac{1}{2} \sum_{\eta s s'} f^\dagger_{\vb{R} \beta_1 \eta s} f^\dagger_{\vb{R}\beta_1'\eta s'}  \begin{pmatrix}
        J_{\rm a} & 0 & 0 & 0 \\
        0 & -J_{\rm a} & J_{\rm b} & 0 \\
        0 & J_{\rm b} & -J_{\rm a} & 0 \\
        0 & 0 & 0 & J_{\rm a} \\
    \end{pmatrix}_{\beta_1'\beta_1, \beta_2'\beta_2} f_{\vb{R}\beta_2'\eta s'} f_{\vb{R}\beta_2\eta s} \\
    & -\frac{1}{2} \sum_{\eta s s'} f^\dagger_{\vb{R}\beta_1 \eta s} f^\dagger_{\vb{R}\beta_1'\overline{\eta} s'}  \begin{pmatrix}
        J_{\rm a} & 0 & 0 & J_{\rm b} \\
        0 & -J_{\rm a} & 0 & 0 \\
        0 & 0 & -J_{\rm a} & 0 \\
        J_{\rm b} & 0 & 0 & J_{\rm a} \\
    \end{pmatrix}_{\beta_1'\beta_1, \beta_2'\beta_2} f_{\vb{R}\beta_2' \overline{\eta} s'} f_{\vb{R}\beta_2\eta s} \\
    & -\frac{1}{2} \sum_{\eta s s'} f^\dagger_{\vb{R}\beta_1 \overline{\eta} s} f^\dagger_{\vb{R}\beta_1'\eta s'}  \begin{pmatrix}
        J_{\rm e} & 0 & 0 & J_{\rm d} \\
        0 & 0 & J_{\rm d} & 0 \\
        0 & J_{\rm d} & 0 & 0 \\
        J_{\rm d} & 0 & 0 & J_{\rm e} \\
    \end{pmatrix}_{\beta_1'\beta_1, \beta_2'\beta_2} f_{\vb{R}\beta_2' \overline{\eta} s'} f_{\vb{R}\beta_2\eta s}
\end{aligned}
\end{equation}
Where we take \cite{YJWANG2024_PRL} $(J_a, J_b, J_d, J_e) = (-1/3, -1/3, -1/3, -1)J_H + (0, 0, 1, 1) J_A$. \\

The model owns translation symmetry, and obeys discrete crystalline symmetries $C_{2z}\mathcal{T}$, $C_{3z}$, $C_{2x}$. The whole Hamiltonian containing both valleys further has $C_{2z}$ and $\mathcal{T}$ independently. Continuous local symmetries include charge $\rm U(1)$, valley $\rm U(1)$, and total spin $\rm SU(2)$. For superconducting phases, the charge U(1) symmetry is broken. For nematic phases, the $C_{3z}$ symmetry is broken. In the following subsections, we will constrain the variational space by these symmetries (except $C_{3z}$ and charge U(1)) both in the single-body as well as many-body levels. \\

We emphasize that in order to correctly describe the Kondo physics which is the interplay between $f$- and $c$-electrons, the occupation number of $c$-orbitals is added as an variational parameter in addition to the reduced Nambu $f$-orbital density matrix $\bm{\varrho}^0$. In addition, we also need to include: 1. chemical potential term for $c$-electron; 2. Coulomb interaction between $f$- and $c$-electrons $\hat{H}_W$ as well as the Coulomb interaction between $c$-electrons $\hat{H}_V$, both of which are treated in Hartree level. The new variational energy functional thus reads:
\begin{gather}
    \mathcal{L}[\bm{\varrho}^0, \nu_c] = \mathcal{L}[\bm{\varrho}^0] - \mu(\nu_c + 12) + \lambda^F_c (\langle \hat{N}_c \rangle_{G} - \nu_c - 12) + W \left(\langle \hat{N}_f \rangle_G - 4\right) \nu_c + \frac{V}{2} \nu_c^2 \label{eq:LG_TBG}
\end{gather}
This requires an additional Lagrange multiplier $\lambda^F_c$ added to $\bm{\lambda}^F$ and $\hat{\mathcal{O}}^F_c$ added to $\{ \hat{\mathcal{O}}^F_l\}$ as well as slight modifications in \ref{sec:FermiPart}. For calculations done in this work, we set the two phenomenological parameters $W = V = 0.7 U$.

\subsection{single-body basis in \textit{irrep} blocks}
Due to valley U(1) symmetry, the reduced Nambu density matrix $\bm{\varrho}^0$ is valley block-diagonal:
\begin{gather}
    \bm{\varrho}^0 = \begin{pmatrix}
        \bm{\rho}^0_{\eta} & \bm{0} & \bm{\Delta}^0_{\eta} & \bm{0} \\
        \bm{0} & \bm{\rho}^0_{\bar{\eta}} & \bm{0} & \bm{\Delta}^0_{\bar{\eta}} \\
        \left(\bm{\Delta}^0_{\eta}\right)^\dagger & \bm{0} & \bm{1} - \left(\bm{\rho}^0_{\bar{\eta}}\right)^T & \bm{0} \\
        \bm{0} & \left(\bm{\Delta}^0_{\bar{\eta}}\right)^\dagger & \bm{0} & \bm{1} - \left(\bm{\rho}^0_{\eta}\right)^T
    \end{pmatrix} \label{eq:TBG_varrho0}
\end{gather}
$C_{2z}$ and $C_{2x}$ relates two valleys:
\begin{gather}
    \bm{\rho}^0 \defeq \bm{\rho}^0_{\eta} = \bm{\rho}^0_{\bar{\eta}}, \quad \bm{\Delta}^0 \defeq \bm{\Delta}^0_{\eta} = \bm{\Delta}^0_{\bar{\eta}}
\end{gather}
$\mathcal{T}$ gives (fixing the gauge choice of $\mathcal{T}$):
\begin{gather}
    (\bm{\rho}^0)^* = \bm{\rho}^0, \quad (\bm{\Delta}^0)^* = \bm{\Delta}^0
\end{gather}
$C_{2x}$ further gives:
\begin{gather}
    \bm{\rho}^0 = \rho^0_{0} \sigma_0 + \rho^0_x \sigma_x, \quad \bm{\Delta}^0 = \Delta^0_0 \sigma_0 + \Delta^0_x \sigma_x, \quad \sigma_0, \sigma_x \text{ are Pauli matrices in orbital space}
\end{gather}
The 4 single-body basis $\{ \mathcal{O}_l \}$ \eqref{def:singlebodybasis} thus are:
\begin{gather}
    \frac{1}{2}\begin{pmatrix}
        \sigma_0 & & & \\
        & \sigma_0 & & \\
         & & -\sigma_0 & \\
        & & & -\sigma_0
    \end{pmatrix} 
    , \quad
    \frac{1}{2} \begin{pmatrix}
        \sigma_x & & & \\
        & \sigma_x & & \\
         & & -\sigma_x & \\
        & & & -\sigma_x
    \end{pmatrix}
    , \quad
    \frac{1}{2} \begin{pmatrix}
         &  & \sigma_0 & \\
        & & & \sigma_0 \\
        \sigma_0 & &  & \\
        & \sigma_0 & & 
    \end{pmatrix} , \quad
    \frac{1}{2} \begin{pmatrix}
            &  & \sigma_x & \\
            & & & \sigma_x \\
            \sigma_x & &  & \\
            & \sigma_x & &
    \end{pmatrix}
\end{gather}

\subsection{many-body basis in \textit{irrep} blocks}
We first need to define the atomic Fock basis convention. For notation simplicity, we also denote $\alpha = (\beta, \eta, s)$, and $\alpha=1,\cdots,8$ is equivalent to $\alpha=(1,+,\uparrow), (1,+,\downarrow), (1,-,\uparrow), \cdots, (2,-,\downarrow)$. 
The $2^8$ Fock states $|\Gamma\rangle$ that span the local Hilbert space are defined as, 
\begin{align}   \label{eq:f_Gamma}
    |\Gamma \rangle = \prod^<_{\alpha \in \Gamma} f^\dagger_{\alpha} |\text{emp}\rangle
\end{align}
where within the product $\prod^<_{\alpha}$ the largest $\alpha$ is created first, \textit{e.g.} $f^\dagger_{1} f^\dagger_{2} \cdots$. We will also denote 
\begin{align}
    |\Gamma\rangle &= |\alpha_1 \cdots \alpha_n\rangle , \quad \text{where}  \alpha_1 < \cdots < \alpha_n \text{and} \alpha_{1,\cdots,n} \in \Gamma \\
    |\text{occ}\rangle &= |12345678\rangle , \quad \quad \quad 
    |\overline{\Gamma}\rangle = \prod^{>}_{\alpha \in \Gamma} f_{\alpha} |\text{occ}\rangle
\end{align}
For example, $|\overline{14}\rangle = f_4 f_1 |\text{occ}\rangle = |235678\rangle$, $|\overline{13}\rangle = f_3 f_1 |\text{occ}\rangle = - |245678\rangle$. 

For all the SC phases that we will consider in this work, valley U(1) charge $N_v$ (generated by $\tau^z$), total SU(2) spin $j$ (generated by $\zeta^{x,y,z}$), and crystalline symmetries $C_{2z}$ ($\sigma^x \tau^x$) and $C_{2x}$ ($\sigma^x$) will always be preserved. 
Therefore, all the $2^8$ atomic configurations can be classified into a series of blocks 
, where each block $B$ consists of a multiple times of the same irrep of $N_v$, $j$, and $C_{2z}$, $C_{2x}$. 
We dub the irreps in each block $B$ as $\Xi=1,\cdots,n_B$, where $n_B$ denotes the total number of irreps in this block, and dub each atomic configuration in each irrep $\Xi$ as $|\Xi, q\rangle$, where $q=1,\cdots,d_B$, and $d_B$ denotes the dimension of this irrep. 
Consequently, any local quantity (or operator) $O$ that respects the above symmetries must take the following form under such an atomic basis, 
\begin{align}   \label{eq:O_form}
    O = \sum_{B} \sum_{\Xi,\Xi'=1}^{n_B} O^{(B)}_{\Xi,\Xi'} \sum_{q=1}^{d_B} |\Xi, q \rangle \langle \Xi', q|
\end{align}
Finally, we can also fix the overall complex phase of $|\Xi, q\rangle$ with $C_{2z}T$ ($\sigma^x K$), so that $(C_{2z}T) |\Xi,q\rangle = |\Xi,q\rangle$. 
By this assumption, if $O$ also commutes with $C_{2z}T$, then $O^{(B)}_{\Xi,\Xi'}$ must be real numbers. 
An example of such quantities would be the on-site interaction $H_U + H_{\rm AH}$, which has already been block-diagonalized in Ref.~\cite{YJWang2025_PRB}.

In sum, there are $20$ such blocks, as summarized in 
For any local quantity $O$ that respects the above symmetries, it takes $\sum_{B} n_B^2 = 513$ real parameters to describe it in the above block-diagonalized form.

The local electron configuration $\{ \ket{I} \}$ can be classfied into 20 symmetry blocks (allow breaking of $C_{3z}$):
\begin{table} 
    \centering
    \caption{\label{tab:B_Xi} $2^8$ atomic configurations classified into $20$ symmetry blocks, $\sum_B n_B d_B = 2^8$. In general, it takes $\sum_{B} n_B^2 = 513$ real parameters to parametrize a symmetric local quantity.}
    \begin{tabular}{cccc||c|c|c}
    \hline\hline
        $N_v$ & $j$ & $C_{2z}$ & $C_{2x}$ & $n_B$ & $d_B$ & Wave-functions \\
        0 & 0 & $+$ & $+$ & 11 & 1 & \cref{tab:irrep_00pp}[0] \\ 
        0 & 0 & $+$ & $-$ & 4 & 1 & \cref{tab:irrep_00pm}[1] \\ 
        0 & 0 & $-$ & $+$ & 1 & 1 & \cref{tab:irrep_00mp}[2] \\ 
        0 & 0 & $-$ & $-$ & 4 & 1 & \cref{tab:irrep_00mm}[3] \\ 
    \hline
        0 & 1 & $+$ & $+$ & 1 & 3 & \cref{tab:irrep_01pp}[4] \\ 
        0 & 1 & $+$ & $-$ & 4 & 3 & \cref{tab:irrep_01pm}[5] \\ 
        0 & 1 & $-$ & $+$ & 6 & 3 & \cref{tab:irrep_01mp}[6] \\ 
        0 & 1 & $-$ & $-$ & 4 & 3 & \cref{tab:irrep_01mm}[7] \\ 
    \hline
        0 & 2 & $+$ & $+$ & 1 & 5 & \cref{tab:irrep_02}[8] \\ 
    \hline
        $\pm 2$ & 0 & $\pm$ & $+$ & 6 & 2 & \cref{tab:irrep_20p}[9] \\
        $\pm 2$ & 0 & $\pm$ & $-$ & 4 & 2 & \cref{tab:irrep_20m}[10] \\
    \hline
        $\pm 2$ & 1 & $\pm$ & $+$ & 2 & 6 & \cref{tab:irrep_21p}[11] \\
        $\pm 2$ & 1 & $\pm$ & $-$ & 4 & 6 & \cref{tab:irrep_21m}[12] \\
    \hline
        $\pm 4$ & 0 & $\pm$ & $+$ & 1 & 2 & \cref{tab:irrep_40}[13] \\
    \hline\hline
        $\pm 1$ & $\frac{1}{2}$ & $\pm$ & $+$ & 10 & 4 & \cref{tab:irrep_11}[14] \\
        $\pm 1$ & $\frac{1}{2}$ & $\pm$ & $-$ & 10 & 4 & \cref{tab:irrep_11}[15] \\
    \hline
        $\pm 1$ & $\frac{3}{2}$ & $\pm$ & $+$ & 2 & 8 & \cref{tab:irrep_13}[16] \\
        $\pm 1$ & $\frac{3}{2}$ & $\pm$ & $-$ & 2 & 8 & \cref{tab:irrep_13}[17] \\
    \hline
        $\pm 3$ & $\frac{1}{2}$ & $\pm$ & $+$ & 2 & 4 & \cref{tab:irrep_31}[18] \\
        $\pm 3$ & $\frac{1}{2}$ & $\pm$ & $-$ & 2 & 4 & \cref{tab:irrep_31}[19] \\
    \hline\hline
    \end{tabular} 
\end{table}
The representative states for total spin number $j>0$ $|\Xi,q=1\rangle$ are chosen as with the highest total $z$-spin. 
\setcounter{table}{0}
\begin{table} 
    \centering
    \caption{\label{tab:irrep_00pp} \textbf{Block 0} ($n_B=11$, $d_B=1$):  $(N_v, j) = (0,0)$, $C_{2z}=+$, $C_{2x}=+$. }
    \begin{tabular}{c|c|c}
    \hline
    	Particle number & Rep. w.f. & Reduced from $\rho$ \\
    \hline
    	$N=0$ & $|\emp\rangle$ & $A_1$  \\
    \hline
	$N=2$ & $ |14\rangle + |58\rangle - |23\rangle - |67\rangle $ & $A_1$   \\
	 & $ |18\rangle - |45\rangle - |27\rangle + |36\rangle $ & $E_2$ \\
    \hline
    $N=4$ & $ |1234\rangle + |5678\rangle $ & $ A_1 $   \\
	 & $ |1368\rangle+|2457\rangle-|1458\rangle-|2367\rangle $ & $ A_1 $   \\
	 & $ |1368\rangle + |2457\rangle + |1458\rangle + |2367\rangle - 2|2358\rangle - 2|1467\rangle $ & $ A_1 $   \\
     & $ |1278\rangle + |3456\rangle$ & $ E_2 $    \\
	 &  $ |1238\rangle - |1247\rangle - |2578\rangle + |1678\rangle - |4567\rangle + |3568\rangle + |1346\rangle - |2345\rangle $  & $ E_2 $  \\
    \hline
    	$N=6$ & $|\ovl{14}\rangle + |\ovl{58}\rangle - |\ovl{23}\rangle - |\ovl{67}\rangle$ & $A_1$ \\
	 & $|\ovl{18}\rangle - |\ovl{45}\rangle - |\ovl{27}\rangle + |\ovl{36}\rangle$ & $E_2$ \\
    \hline
    	$N=8$ &  $|\occ\rangle$ & $A_1$ \\
    \hline
    \end{tabular}
    \\[1cm]
    \centering
    \caption{\label{tab:irrep_00pm} \textbf{Block 1} ($n_B=4$, $d_B=1$): $(N_v, j) = (0,0)$, $C_{2z}=+$, $C_{2x}=-$. }
    \begin{tabular}{c|c|c}
    \hline
    	Particle number & Rep. w.f. & Reduced from $[\rho, j]$ \\
    \hline
	$N=2$ & $i\left(|18\rangle + |45\rangle - |27\rangle - |36\rangle \right)$ & $ E_2$ \\
    \hline
    	$N=4$ & $ i\left( |1278\rangle - |3456\rangle \right)$ & $E_2$  \\
	& $i\left( |1238\rangle - |1247\rangle - |2578\rangle + |1678\rangle + |4567\rangle - |3568\rangle - |1346\rangle + |2345\rangle \right)$ & $E_2$  \\
    \hline
    	$N=6$ & $i\left( |\ovl{18}\rangle + |\ovl{45}\rangle - |\ovl{27}\rangle - |\ovl{36}\rangle \right)$ & $E_2$ \\
    \hline
    \end{tabular}
    \\[1cm]
    \centering
    \caption{\label{tab:irrep_00mp} \textbf{Block 2} ($n_B=1$, $d_B=1$): $(N_v, j) = (0,0)$, $C_{2z}=-$, $C_{2x}=+$. }
    \begin{tabular}{c|c|c}
    \hline
    	Particle number & Rep. w.f. & Reduced from $[\rho, j]$ \\
    \hline
    	$N=4$ & $ |1238\rangle - |1247\rangle + |2578\rangle - |1678\rangle - |4567\rangle + |3568\rangle - |1346\rangle + |2345\rangle $ & $ E_1 $ \\
    \hline
    \end{tabular}
\end{table}


\begin{table} 
    \centering
    \caption{\label{tab:irrep_00mm} \textbf{Block 3} ($n_B=4$, $d_B=1$): $(N_v, j) = (0,0)$, $C_{2z}=-$, $C_{2x}=-$. }
    \begin{tabular}{c|c|c}
    \hline
    	Particle number & Rep. w.f. & Reduced from $[\rho, j]$ \\
    \hline
	$N=2$ & $ i\left( |14\rangle - |58\rangle - |23\rangle + |67\rangle  \right)$ & $ B_2 $ \\
    \hline
    	$N=4$ & $i \left( |1234\rangle - |5678\rangle \right)$ & $ B_2$ \\
    	& $ i\left( |1238\rangle - |1247\rangle + |2578\rangle - |1678\rangle + |4567\rangle - |3568\rangle + |1346\rangle - |2345\rangle  \right)$ & $ E_1 $ \\
    \hline
    	$N=6$ & $i\left( |\ovl{14}\rangle - |\ovl{58}\rangle - |\ovl{23}\rangle + |\ovl{67}\rangle \right)$ & $ B_2$  \\
    \hline
    \end{tabular}
\end{table}

\begin{table} 
    \centering
    \caption{\label{tab:irrep_01pp} \textbf{Block 4} ($n_B=1$, $d_B=3$): $(N_v, j) = (0,1)$, $C_{2z}=+$, $C_{2x}=+$.}
    \begin{tabular}{c|c|c}
    \hline
    	Particle number & Rep. w.f. & Reduced from $ \rho$ \\
    \hline
    	$N=4$ & $|1237\rangle - |1578\rangle - |3567\rangle + |1345\rangle$ & $ E_2$  \\
    \hline
    \end{tabular}
    \\[1cm]
    \centering
    \caption{\label{tab:irrep_01pm} \textbf{Block 5} ($n_B=4$, $d_B=3$): $(N_v, j) = (0,1)$, $C_{2z}=+$, $C_{2x}=-$. }
    \begin{tabular}{c|c|c}
    \hline
    	Particle number & Rep. w.f. & Reduced from $\rho$ \\
    \hline
    	$N=2$ & $i\left( |13\rangle - |57\rangle \right)$ & $A_2$ \\
    \hline
    	$N=4$ & $i\left( |1358\rangle - |1367\rangle - |1457\rangle + |2357\rangle \right)$ & $A_2$  \\
	 & $i\left( |1237\rangle - |1578\rangle + |3567\rangle - |1345\rangle \right)$ & $ E_2 $  \\
    \hline
    	$N=6$ & $i\left( |\ovl{24}\rangle - |\ovl{68}\rangle \right)$ & $ A_2 $ \\ 
    \hline
    \end{tabular}
    \\[1cm]
    \centering
    \caption{\label{tab:irrep_01mp} \textbf{Block 6} ($n_B=6$, $d_B=3$): $(N_v, j) = (0,1)$, $C_{2z}=-$, $C_{2x}=+$. }
    \begin{tabular}{c|c|c}
    \hline
    	Particle number & Rep. w.f. & Reduced from $[\rho, j]$ \\
    \hline
    	$N=2$ & $|13\rangle + |57\rangle$ & $ B_1 $  \\
	& $|17\rangle - |35\rangle $ & $ E_1 $  \\
    \hline
    	$N=4$ & $|1358\rangle - |1367\rangle + |1457\rangle - |2357\rangle$ & $ B_1 $  \\
	& $|1237\rangle + |1578\rangle - |3567\rangle - |1345\rangle$ & $ E_1 $  \\
    \hline
    	$N=6$ & $|\ovl{24}\rangle + |\ovl{68}\rangle$ & $B_1$ \\
         & $|\ovl{28}\rangle - |\ovl{46}\rangle$ & $E_1$ \\
    \hline
    \end{tabular}
    \\[1cm]
    \centering
    \caption{\label{tab:irrep_01mm} \textbf{Block 7} ($n_B=4$, $d_B=3$): $(N_v, j) = (0,1)$, $C_{2z}=-$, $C_{2x}=-$. }
    \begin{tabular}{c|c|c}
    \hline
    	Particle number & Rep. w.f. & Reduced from $[\rho, j]$ \\
    \hline
    	$N=2$ & $i\left( |17\rangle + |35\rangle \right)$ & $ E_1 $  \\
    \hline
    	$N=4$ & $i\left( |1358\rangle + |1367\rangle - |1457\rangle - |2357\rangle \right)$ & $ B_2 $  \\
	& $i\left( |1237\rangle + |1578\rangle + |3567\rangle + |1345\rangle \right)$ & $ E_1 $  \\
    \hline
    	$N=6$ & $i\left( |\ovl{28}\rangle + |\ovl{46}\rangle \right)$ & $ E_1 $  \\
    \hline
    \end{tabular}
\end{table}

\begin{table} 
    \centering
    \caption{\label{tab:irrep_02} \textbf{Block 8} ($n_B=1$, $d_B=5$): $(N_v, j) = (0,2)$. }
    \begin{tabular}{c|c|c}
    \hline
    	Particle number & Rep. w.f. & Reduced from $[\rho, j]$ \\
    \hline
    	$N=4$ & $|1357\rangle$ & $A_1$  \\
    \hline
    \end{tabular}
\end{table}

\begin{table} 
    \centering
    \caption{\label{tab:irrep_20p} \textbf{Block 9} ($n_B=6$, $d_B=2$): $(|N_v|, j) = (2, 0)$, $C_{2x}=+$. }
    \begin{tabular}{c|c|c}
    \hline
    	Particle number & Rep. w.f. & Reduced from $\rho$ \\
    \hline
    	$N=2$ & $|16\rangle - |25\rangle$ & $A_1 + B_1 $  \\
         & $|12\rangle + |56\rangle$ & $E_1 + E_2$  \\
    \hline
        $N=4$ & $|1236\rangle - |1245\rangle - |2567\rangle + |1568\rangle$ & $A_1 + B_1$  \\
        & $|1258\rangle - |1267\rangle + |1456\rangle - |2356\rangle$ & $E_1 + E_2$  \\
    \hline
        $N=6$ & $|\ovl{38}\rangle - |\ovl{47}\rangle$ & $A_1 + B_1$  \\
         & $|\ovl{34}\rangle + |\ovl{78}\rangle$ & $E_1 + E_2$  \\
    \hline
    \end{tabular}
    \\[1cm]
    \centering
    \caption{\label{tab:irrep_20m} \textbf{Block 10} ($n_B=4$, $d_B=2$): $(|N_v|, j) = (2, 0)$, $C_{2x}=-$. }
    \begin{tabular}{c|c|c}
    \hline
    	Particle number & Rep. w.f. & Reduced from $\rho$ \\
    \hline
    	$N=2$ & $i\left( |12\rangle - |56\rangle \right)$ & $ E_1 + E_2$  \\
    \hline
        $N=4$ & $i\left( |1236\rangle - |1245\rangle + |2567\rangle - |1568\rangle \right)$ & $A_2 + B_2$ \\
        & $i\left( |1258\rangle - |1267\rangle - |1456\rangle + |2356\rangle \right)$ & $E_1+E_2$ \\
    \hline
        $N=6$ & $i\left( |\ovl{34}\rangle - |\ovl{78}\rangle \right)$ & $E_1 + E_2$ \\
    \hline
    \end{tabular}
\end{table}

\begin{table} 
    \centering
    \caption{\label{tab:irrep_21p} \textbf{Block 11} ($n_B=2$, $d_B=6$): $(|N_v|, j) = (2, 1)$, $C_{2x}=+$. }
    \begin{tabular}{c|c|c}
    \hline
    	Particle number & Rep. w.f. & Reduced from $\rho$ \\
    \hline
        $N=4$ & $|1235\rangle - |1567\rangle$ & $[A_1+B_1, 1]$ \\
         & $|1257\rangle + |1356\rangle$ & $[E_1+E_2, 1]$ \\
    \hline
    \end{tabular}
    \\[1cm]
    \centering
    \caption{\label{tab:irrep_21m} \textbf{Block 12} ($n_B=4$, $d_B=6$): $(|N_v|, j) = (2, 1)$, $C_{2x}=-$. }
    \begin{tabular}{c|c|c}
    \hline
    	Particle number & Rep. w.f. & Reduced from $\rho$ \\
    \hline
        $N=2$ & $i |15\rangle$ & $A_2+B_2$ \\
    \hline
        $N=4$ & $i\left( |1235\rangle + |1567\rangle \right)$ & $A_2+B_2$ \\
        & $i\left( |1257\rangle - |1356\rangle \right)$ & $E_1+E_2$ \\
    \hline
        $N=6$ & $i |\ovl{48}\rangle$ & $A_2+B_2$ \\
    \hline
    \end{tabular}
\end{table}

\begin{table} 
    \centering
    \caption{\label{tab:irrep_40} \textbf{Block 13} ($n_B=1$, $d_B=2$): $(|N_v|, j) = (4, 0)$. }
    \begin{tabular}{c|c|c}
    \hline
    	Particle number & Rep. w.f. & Reduced from $[\rho, j]$ \\
    \hline
    	$N=4$ & $|1256\rangle$ & $A_1+B_1$  \\
    \hline
    \end{tabular}
\end{table}

\begin{table} 
    \centering
    \caption{\label{tab:irrep_11} \textbf{Block 14 \& 15} ($n_B=10$, $d_B=4$): $(|N_v|, j) = (1, 1/2)$, $C_{2x}=\xi =\pm$. }
    \begin{tabular}{c|c|c}
    \hline
    	Particle number & Rep. w.f. & Reduced from $\rho$ \\
    \hline
    	$N=1$ & $i^{\frac{1-\xi}{2}} \left( |1\rangle + \xi |5\rangle \right)$ & $E_1 + E_2$  \\
    \hline
        $N=3$ & $i^{\frac{1-\xi}{2}} \left( |127\rangle + \xi |356\rangle \right)$ & $A_1 + B_1$  \\
         & $i^{\frac{1-\xi}{2}} \left(|123\rangle + \xi |567\rangle \right)$ & $E_1+E_2$ \\
         & $i^{\frac{1-\xi}{2}} \left(|167\rangle - |257\rangle + \xi |235\rangle - \xi |136\rangle\right)$ & $E_1+E_2$ \\
         & $i^{\frac{1-\xi}{2}} \left(|167\rangle + |257\rangle - 2|158\rangle + \xi |235\rangle + \xi |136\rangle - 2\xi |145\rangle \right)$ & $E_1+E_2$ \\
    \hline
        $N=5$ & $i^{\frac{1-\xi}{2}} \left(|\ovl{278}\rangle + \xi |\ovl{346}\rangle \right)$ & $A_1 + B_1$  \\
         & $i^{\frac{1-\xi}{2}} \left(|\ovl{234}\rangle + \xi  |\ovl{678}\rangle \right)$ & $E_1+E_2$ \\
         & $i^{\frac{1-\xi}{2}} \left(|\ovl{238}\rangle - |\ovl{247}\rangle + \xi |\ovl{467}\rangle - \xi |\ovl{368}\rangle \right)$ & $E_1+E_2$ \\
         & $i^{\frac{1-\xi}{2}} \left(|\ovl{238}\rangle + |\ovl{247}\rangle - 2|\ovl{148}\rangle + \xi |\ovl{467}\rangle + \xi |\ovl{368}\rangle - 2\xi |\ovl{458}\rangle \right)$ & $E_1+E_2$ \\
    \hline
        $N=7$ & $i^{\frac{1-\xi}{2}} \left(|\ovl{4}\rangle + \xi |\ovl{8}\rangle \right)$ & $E_1+E_2$ \\
    \hline
    \end{tabular}
\end{table}

\begin{table} 
    \centering
    \caption{\label{tab:irrep_13} \textbf{Block 16 \& 17} ($n_B=2$, $d_B=8$): $(|N_v|, j) = (1, 3/2)$, $C_{2x}=\xi =\pm$. }
    \begin{tabular}{c|c|c}
    \hline
    	Particle number & Rep. w.f. & Reduced from $\rho$ \\
    \hline
        $N=3$ & $i^{\frac{1-\xi}{2}} \left( |135\rangle + \xi |157\rangle \right)$ & $A_1 + B_1$ \\
    \hline
        $N=5$ & $i^{\frac{1-\xi}{2}} \left(|\ovl{248}\rangle + \xi |\ovl{468}\rangle \right)$ & $A_1 + B_1$  \\
    \hline
    \end{tabular}
\end{table}

\begin{table} 
    \centering
    \caption{\label{tab:irrep_31} \textbf{Block 18 \& 19} ($n_B=2$, $d_B=4$): $(|N_v|, j) = (3, 1/2)$, $C_{2x}=\xi =\pm$. }
    \begin{tabular}{c|c|c}
    \hline
    	Particle number & Rep. w.f. & Reduced from $\rho$ \\
    \hline
        $N=3$ & $i^{\frac{1-\xi}{2}} \left( |125\rangle + \xi |156\rangle \right)$ & $A_1 + B_1$ \\
    \hline
        $N=5$ & $i^{\frac{1-\xi}{2}} \left(|\ovl{348}\rangle + \xi |\ovl{478}\rangle \right)$ & $A_1 + B_1$  \\
    \hline
    \end{tabular}
\end{table}

\clearpage

\section{Parametrization of reduced (Nambu) density matrix} \label{appndx:prmtrz_rho0}

A valid parametrization of $\bm{\rho}^0$ (or Nambu density matrix $\bm{\varrho}^0$) is indispensable for the variational method, as it must ensure that all eigenvalues of $\bm{\rho}^0$ lie within $[0,1]$, regardless of whether the system is superconducting or not. We present two such parametrizations: an unconstrained one and a constrained one. In the unconstrained parametrization, we employ a Fermi-Dirac-like distribution function, which is straightforward to implement but can become numerically inaccurate when some eigenvalues of $\bm{\rho}^0$ approach 0 or 1, as occurs in, e.g., orbital-selective Mott insulators or strongly correlated superconductors. In the constrained parametrization, we express $\bm{\rho}^0$ directly in the matrix basis \eqref{def:singlebodybasis} and derive explicit forms of the constraints, yielding a scheme that is numerically stable in all cases.

\subsection{Parameterizing $\rho^0$ ($\varrho^0$) (unconstrained)}
It's obvious that for any Hermitian matrix $\mathcal{X}$, the following construction:
\begin{gather}
    \bm{\rho}^0(\mathcal{X}) = \frac{1}{e^{\mathcal{X}} + 1} \label{rho0_pmtrz:exp}
    \end{gather}
has all of its eigenvalues lies between 0 and 1. Conversely, we can show that any legit density matrix $\bm{\rho}^0$ can be parametrized in this way by realizing that the exponential function $e^x$ is monotone such that we can solve $x_i = \log(1/d_i - 1)$ for each eigenvalue $d_i$ of $\bm{\rho}^0$:
\begin{gather}
\bm{\rho}^0 = \mathcal{U} \mathcal{D} \mathcal{U}^\dagger = \,
\mathcal{U} \begin{pmatrix} \dmat{d_1, d_2, \ddots, d_N } \end{pmatrix} \mathcal{U}^\dagger =\,
\mathcal{U} \begin{pmatrix} \dmat{\frac{1}{e^{x_1} + 1}, \frac{1}{e^{x_2} + 1}, \ddots, \frac{1}{e^{x_N} + 1} } \end{pmatrix} \mathcal{U}^\dagger =\,
\frac{1}{e^{\mathcal{X}} + 1}
\end{gather}
where $\mathcal{X}$ commutes with $\frac{1}{e^{\mathcal{X}} + 1}$ and is defined as:
\begin{gather}
\mathcal{X} = \mathcal{U} \begin{pmatrix} \dmat{x_1, x_2, \ddots, x_N } \end{pmatrix} \mathcal{U}^\dagger 
\end{gather}
We now have shown there's one-to-one correspondence between $\bm{\rho}^0$ and $\mathcal{X}$. Moreover, one can check that $\bm{\rho}^0$ is subjected to some symmetries $[\bm{\rho}^0, P]=0$ if and only if $\mathcal{X}$ satisfies $[\mathcal{X}, P]=0$. This means one can choose the matrix basis \eqref{def:singlebodybasis} to expand $\mathcal{X}=\sum_{l}x_l \mathcal{O}_l$. 

In practice, the Gutzwiller method does not permit the eigenvalues of $\bm{\rho}^0$ to be exactly 0 or 1, as either case would render the matrix $\bm{m}^0$ singular. A simple modification to \eqref{rho0_pmtrz:exp} would make the eigenvalues of $\bm{\rho}^0$ lie within $[\delta_{\text{lb}}, 1 - \delta_{\text{ub}}]$:
\begin{gather}
    \bm{\rho}^0(\mathcal{X}) = \frac{1-(\delta_{\text{lb}} + \delta_{\text{ub}})}{e^{\mathcal{X}} + 1} + \delta_{\text{lb}}\label{rho0_pmtrz:exp}
\end{gather}

\subsubsection{Derivatives of $\rho^0$($\varrho^0$)}
The derivative of $\rho^0$ taken with respect to $x_l$ can be calculated using second-order perturbation theory \eqref{formula:Drvt_2ndPert}:
\begin{gather}
    \pfpx{\bm{\rho}^0_{\alpha \beta}}{x_l} = \sum_{m \neq n} \sum_{i j} \bra{\alpha} \ket{m^{l}_i} \mathcal{F}_{mn; ij} \bra{m^{l}_i} \hat{\mathcal{O}}_{l} \ket{n^{l}_j} \bra{n^{l}_j} \ket{\beta} \\
    \mathcal{F}_{mn; ij} = \delta_{mn} \delta_{ij} \left[-\frac{e^{d_n}}{(e^{d_n} + 1)^2}\right] + (1 - \delta_{mn}) \frac{\frac{1}{e^{d_m}+1} - \frac{1}{e^{d_n}+1}}{d_m - d_n}
\end{gather}
Where $\{ \ket{m^{l}_i} \}$ is the set of re-gauged eigenvectors of $\mathcal{X}$ with eigenvalues $\{ d_m \}$. They're called re-gauged because they also diagonalize the m-degenerate subspace matrix $\mathcal{O}^{l}_{m; ij} \defeq \bra{m^{l}_i} \hat{\mathcal{O}}_{l} \ket{m^{l}_j} = \delta_{ij} \mathcal{O}^{l}_{m; ii}$.

\subsection{Parametrising $\rho^0$($\varrho^0$) (polynomial constraints)}
Although the exponential map method gives a simple and elegant parameterization, it does not provide an accurate derivative modulation when the minimization is close to the boundary (the eigenvalues of $\bm{\rho}^0$ is close to 0 or 1) because $\dfdx{\bm{\rho}^0}{\bm{x}}$ becomes exponentially very small, which leads the minimization to stay at the boundary since $\dfdx{\mathcal{L}}{\bm{x}} = \dfdx{\mathcal{L}}{\bm{\rho}^0} \dfdx{\bm{\rho}^0}{\bm{x}}$ would be small at the boundary. On the light of this, we parameterize $\bm{\rho}^0$ directly as a linear combination of generators of SU(N) (denoted by $\{ \hat{\lambda}_i \}$, which the matrix basis defined in \eqref{def:singlebodybasis} are a special case of) and derive explicit forms of the constraints on the coefficients to ensure that all eigenvalues of $\bm{\rho}^0$ lie within $[\delta_{\text{lb}},1 - \delta_{\text{ub}}]$. $\bm{\rho}^0$ is expressed by $\{ \hat{\lambda}_i \}$ as:
\begin{gather}
    \bm{\rho}^0 = \frac{n_e}{N} \mbbm{1} + \frac{1}{2} \sum_{i=1}^{N^2-1} \lambda_i \hat{\lambda}_i,
\end{gather}
where $n_e$ is a constant which equals the occupation number per site (or differ by an integral number for special Fermi liquid state) for charge U(1) preserving systems and equals $\frac{N}{2}$ for superconducting systems (as $\bm{\rho}^0$ becomes a Nambu density matrix). The coefficients $\{ \lambda_i \}$ are the variational parameters and the constraints on $\{ \lambda_i \}$ can be derived from the characteristic polynomial of $\bm{\rho}^0$ by generalising Kimura's method \cite{KIMURA2003339}:
\begin{theorem}[Boundary for Polynomial Roots]
    The (N-i)-th derivative of the characteristic polynomial $P_{\rho}(x) = \prod_i^N (x-x_i)$ is calculated as ($x_1, x_2, \dots, x_N$ are the eigenvalues of $\bm{\rho}^0$):
    \begin{gather}
        P^{(N-i)}_{\rho}(x) = (N-i)! \sum_{1 \leq j_1 < j_2 < \dots < j_i \leq N}\  \prod_{k=j_1}^{j_i}(x-x_k) \ \ (i=1, \dots , N) \label{def:P_drvt}
    \end{gather}
    Then for any real number $r$, we have:
    \begin{align}
        P_{\rho}^{(j)}(r) \geq 0 \text{,  for all $j \in \{ 0, \dots, N-1 \}$)} &\Leftrightarrow  x_i \leq r \text{,  for all $i \in \{ 1, \dots, N \}$}  \label{cnstr:upperbound} \\
        (-1)^{N-j} P_{\rho}^{(j)}(r) \geq 0 \text{,  for all $j \in \{ 0, \dots, N-1 \}$} &\Leftrightarrow  x_i \geq r \text{,  for all $i \in \{ 1, \dots, N \}$} \label{cnstr:lowerbound}
    \end{align}
\end{theorem}
\emph{Proof:} We'll start with \eqref{cnstr:upperbound} and the proof of \eqref{cnstr:lowerbound} follows naturally. We'll only prove sufficiency ($\Rightarrow$) since the necessity ($\Leftarrow$) is trivial. We use proof by contradiction. Define $\delta_k \defeq r-x_k$. Suppose at least one $\delta_i<0$ and we can take $\delta_N < 0$  without loss of generality. Let:
\begin{align}
\tilde{P}^{(N-1-i)}_{\rho}(r) =& (N-1-i)! \sum_{1 \leq j_1 < j_2 < \dots < j_i<N}\delta_{j_1}\delta_{j_2} \dots \delta_{j_i} \ \ (i=1, \dots , N-1) \\
\tilde{P}^{(-1)}_{\rho} \defeq& 0 \\
\tilde{P}^{(N-1)}_{\rho} =& 1
\end{align}
We rewrite \eqref{def:P_drvt}:
\begin{align}
P_{\rho}(r) =& \tilde{P}_{\rho}(r) \delta_N \label{prf:0}\\
P_{\rho}^{(1)}(r) =&  \tilde{P}_{\rho}(r) + \tilde{P}^{(1)}_{\rho}(r) \delta_N \label{prf:1}\\
\dots & \nonumber \\
P_{\rho}^{(N-k)}(r) =& (N-k) \tilde{P}^{(N-k-1)}_{\rho}(r) + \tilde{P}^{(N-k)}_{\rho}(r) \delta_N \\
\dots & \nonumber \\
P_{\rho}^{(N-1)}(r) =&  (N-1) \tilde{P}^{(N-2)}_{\rho}(r) +  (N-1)!\delta_N
\end{align}
From \eqref{prf:0} we conclude $\tilde{P}_{\rho}(r) \leq 0$ since we've required $P^{(j)}_{\rho}(r) \geq 0$; Following \eqref{prf:1} we conclude $\tilde{P}'_{\rho}(r) \leq 0$. Continuing this deduction and we eventually get $P^{(N-1)}_{\rho}(r) < 0$ which contradicts to our assumption. To prove \eqref{cnstr:lowerbound}, we can define:
\begin{align}
    Q^{(N-i)}_{\rho}(r) \defeq& (-1)^i P^{(N-i)}_{\rho}(r) \\
    =& \sum_{1 \leq j_1 < j_2 < \dots < j_i \leq N} \eta_{j_1} \eta_{j_2} \dots \eta_{j_i} \ \ (i=1, \dots , N)\\[1em]
    \eta_k \defeq& x_k - r \\[1em]
    \tilde{Q}^{(N-1-i)}_{\rho}(r) \defeq& (-1)^i \tilde{P}^{(N-1-i)}_{\rho}(r) \\
    =& \sum_{1 \leq j_1 < j_2 < \dots < j_i<N}\eta_{j_1}\eta_{j_2} \dots \eta_{j_i} \ \ (i=1, \dots , N-1)
\end{align}
To prove sufficiency in \eqref{cnstr:lowerbound}, we can assume $\eta_N<0$. We can write down the same recursive equations for $Q_{\rho}(r)$:
\begin{align}
    Q_{\rho}(r) =& \tilde{Q}_{\rho}(r) \eta_N \\
    \dots & \nonumber \\
    Q_{\rho}^{(N-k)}(r) =& (N-k) \tilde{Q}^{(N-k-1)}_{\rho}(r) + \tilde{Q}^{(N-k)}_{\rho}(r) \eta_N \\
    \dots & \nonumber \\
    Q_{\rho}^{(N-1)}(r) =& (N-1) \tilde{Q}^{(N-2)}_{\rho}(r) +  (N-1)!\eta_N
\end{align}
QED.\\

Now we can write down the constraints. $P^{(N-i)}_{\rho}(x)$ can be expressed as:
\begin{gather}
    P^{(N-i)}_{\rho}(x) = \sum_{j=0}^{i} (-1)^j a_j \frac{(N-j)!}{(i-j)!} x^{i-j}, \label{eq:P_drvt_aj}
\end{gather}
where the coefficients $\{ a_j \}$ can be solved using Faddeev–LeVerrier algorithm:
\begin{equation}
    \begin{aligned}
        M_0 &=0  &&a_0=1 \\
        M_k &=\hat{\rho} M_{k-1} + (-1)^{k-1} a_{k-1}\mbbm{I}_N && a_k = (-1)^{k-1} \frac{1}{k}\Tr(\hat{\rho} M_k) 
    \end{aligned} \label{method:Faddev-LeVerrier}
\end{equation}

The constraints are then calculated by plugging $r=\delta_{\text{lb}}$ and $r=1 - \delta_{\text{ub}}$ in \eqref{cnstr:lowerbound} and \eqref{cnstr:upperbound} respectively. In the following, we will explicitly calculate the constraints for $N=4$ case, which is the case for TBG. 

\subsubsection{Constraints for SC state of TBG}
The spectrum of Nambu density matrix is symmetric about $1/2$, which implies $\delta_{\text{lb}} = \delta_{\text{ub}} = \delta$. The density matrix for `one valley' is a $4 \times 4$ Hermitian matrix \eqref{eq:TBG_varrho0}:
\begin{gather}
    \bm{\varrho}^0_{\eta} = \begin{pmatrix}
        \bm{\rho}^0_{\eta} & \bm{\Delta}^0_{\eta} \\
        \bm{\Delta}^0_{\eta} & 1 - \bm{\rho}^0_{\bar{\eta}}
    \end{pmatrix},
\end{gather}
which is spanned by (both $s$ and $\sigma$ denotes Pauli matrices):\textbf{}
\begin{gather}
    \{ \hat{\lambda}_{ia} \} = \{ s_z \otimes \sigma_0, s_z \otimes \sigma_x, s_x \otimes \sigma_0, s_x \otimes \sigma_x \} \label{basis:symmrho}
\end{gather}
For the convenience of later algebras, we enlarge the basis set to:
\begin{gather}
        \{ \hat{\lambda}_{ia} \}_{\text{aux}} = \{ s_z \otimes \sigma_0, s_z \otimes \sigma_x, s_x \otimes \sigma_0, s_x \otimes \sigma_x, s_y \otimes \sigma_0, s_y \otimes \sigma_x, s_0 \otimes \sigma_x \}
\end{gather}
From which we can calculate its structrual constants (The left side belongs to the original basis set \eqref{basis:symmrho} and the right side belongs to the auxiliary basis set):
\begin{gather}
    [\hat{\lambda}_{ia}, \hat{\lambda}_{jb}] = 2 i \epsilon_{ijk} \delta_{ab} \hat{\lambda}_{k0} + 2 i \epsilon_{ijk} (1 - \delta_{ab}) \hat{\lambda}_{kx} \\
    \{ \hat{\lambda}_{ia}, \hat{\lambda}_{jb} \} = 2 \delta_{ij} \delta_{ab} \mbbm{1}_4 + 2 \delta_{ij} (1 - \delta_{ab}) \hat{\lambda}_{0x}
\end{gather}
We can then calculate the traces of products of $\hat{\lambda}$ which will be useful in calculating powers of $\bm{\rho}^0$ (up to power of 4):
\begin{align}
    \Tr(\hat{\lambda}_{ia} \hat{\lambda}_{jb}) =& 4\delta_{ij} \delta_{ab} 
        \\[2em]
    \Tr(\hat{\lambda}_{ia} \hat{\lambda}_{jb} \hat{\lambda}_{kc}) =& \Tr(\frac{1}{2} \left( [\hat{\lambda}_{ia}, \hat{\lambda}_{jb}] + \{ \hat{\lambda}_{ia}, \hat{\lambda}_{jb} \} \right) \hat{\lambda}_{kc}) \\
        =& 2 i\epsilon_{ijk} \delta_{ab} \delta_{0c} + 2 i\epsilon_{ijk} (1 - \delta_{ab}) \delta_{xc} + 2 \delta_{ij} (1-\delta_{ab}) \delta_{k0} \delta_{xc} 
        \\[2em]
    \Tr(\hat{\lambda}_{ia} \hat{\lambda}_{jb} \hat{\lambda}_{kc} \hat{\lambda}_{ld}) =& \Tr(\frac{1}{2} \left( [\hat{\lambda}_{ia}, \hat{\lambda}_{jb}] + \{ \hat{\lambda}_{ia}, \hat{\lambda}_{jb} \} \right) \frac{1}{2} \left( [\hat{\lambda}_{kc}, \hat{\lambda}_{ld}] + \{ \hat{\lambda}_{kc}, \hat{\lambda}_{ld} \} \right)) \\
        =& \Tr\left[\left( i\epsilon_{ijp} \delta_{ab} \hat{\lambda}_{p0} +  i\epsilon_{ijp} (1 - \delta_{ab}) \hat{\lambda}_{px} + \delta_{ij} \delta_{ab} \mbbm{1}_4 +  \delta_{ij} (1 - \delta_{ab}) \hat{\lambda}_{0x} \right) \right. \nonumber \\
        & \left. \times \left( i\epsilon_{klq} \delta_{cd} \hat{\lambda}_{q0} +  i\epsilon_{klq} (1 - \delta_{cd}) \hat{\lambda}_{qx} + \delta_{kl} \delta_{cd} \mbbm{1}_4 +  \delta_{kl} (1 - \delta_{cd}) \hat{\lambda}_{0x} \right)\right] \\[1em]
        =& - 4 \epsilon_{ijp} \epsilon_{klq} \delta_{cd} \delta_{ab} \delta_{pq} - 4\epsilon_{ijp} \epsilon_{klq} (1-\delta_{ab}) (1-\delta_{cd}) \delta_{pq} + 4 i \epsilon_{ijp} \delta_{kl}(1-\delta_{ab}) (1-\delta_{cd}) \delta_{p0} \nonumber \\
        &+ 4 i \delta_{ij} \epsilon_{klq}(1-\delta_{ab}) (1-\delta_{cd}) \delta_{q0} + 4 \delta_{ij} \delta_{kl}(1- \delta_{ab}) (1-\delta_{cd}) + 4 \delta_{ij} \delta_{kl} \delta_{ab} \delta_{cd} \\[1em]
        =& - 4 \epsilon_{ijp} \epsilon_{klp} \delta_{cd} \delta_{ab} - 4\epsilon_{ijp} \epsilon_{klp}(1- \delta_{cd}) (1 - \delta_{ab}) + 4i \epsilon_{ij0}\delta_{kl}(1-\delta_{ab})(1-\delta_{cd}) \nonumber \\
        &+ 4 i \delta_{ij} \epsilon_{kl0} (1-\delta_{ab})(1-\delta_{cd}) + 4 \delta_{ij} \delta_{kl}(1- \delta_{ab}) (1-\delta_{cd}) + 4 \delta_{ij} \delta_{kl} \delta_{ab} \delta_{cd} \\[1em]
        =& - 4 \epsilon_{ijp} \epsilon_{klp} \delta_{cd} \delta_{ab} - 4\epsilon_{ijp} \epsilon_{klp}(1- \delta_{cd}) (1 - \delta_{ab}) + 4 \delta_{ij} \delta_{kl}(1- \delta_{ab}) (1-\delta_{cd}) + 4 \delta_{ij} \delta_{kl} \delta_{ab} \delta_{cd} \label{eq:trace4prod}
\end{align}
The last equal sign is because the anti-symmetric tensor of Pauli matrices does not include the identity matrix, such that $\epsilon_{ij0}=0$. If we restrict the basis set to \eqref{basis:symmrho}, then $\Tr(\hat{\lambda}_{ia} \hat{\lambda}_{jb} \hat{\lambda}_{kc}) = 0$ since $k \neq y \text{ or } 0$. 
We use the above trace formulas to calculate $\Tr(\hat{\rho}^n)$
\begin{align}
    \Tr\hat{\rho} =& 2 \label{eq:Tr_rho1}\\
    \Tr(\hat{\rho}^2) =& \frac{1}{2^2} \left(4 +4 \sum_{\substack{i=z,x \\ a=0, x}} \lambda^2_{ia} \right) = 1+ \Gamma \label{eq:Tr_rho2} \\
    \Tr(\hat{\rho}^3) =& \frac{1}{2^3} \left(4 + 12 \sum_{\substack{i=z,x \\ a=0, x}} \lambda^2_{ia} \right) = \frac{1}{2} + \frac{3}{2} \Gamma \label{eq:Tr_rho3}\\
    \Tr(\hat{\rho}^4) =& \frac{1}{2^4} \left(4+ 24 \sum_{\substack{i=z,x \\ a=0, x}} \lambda^2_{ia} + 4\sum_{\substack{i,k=z, x\\a,c=0, x}}\lambda_{ia} \lambda_{i \bar{a}} \lambda_{kc} \lambda_{k \bar{c}} + 4 \sum_{\substack{i,k=z, x\\a,c=0, x}} \lambda^2_{ia} \lambda^2_{kc} \right) \label{eq:Tr_rho4} \\
    =& \frac{1}{4} + \frac{3}{2} \Gamma + \frac{1}{4} \mathcal{Z}^2 + \frac{1}{4} \Gamma^2 \\[2em]
    \Gamma(\bm{\lambda}) =& \sum_{\substack{i=z,x \\ a=0, x}} \lambda^2_{ia} , \quad \mathcal{Z}(\bm{\lambda}) = \sum_{\substack{i=z, x\\a=0, x}}\lambda_{ia} \lambda_{i \bar{a}}
\end{align}
Where the last 2 terms in \eqref{eq:Tr_rho4} are calculated from the 4-product trace \eqref{eq:trace4prod}. Here we show how the first 2 terms (anti-symmetric) from \eqref{eq:trace4prod} vanish after summing over the indices:
\begin{align}
    \sum_{\substack{i,j,k,l=z,x\\a,b,c,d=0,x}} \lambda_{ia} \lambda_{jb} \lambda_{kc} \lambda_{ld} \left(-2 \epsilon_{ijp} \epsilon_{klp} \delta_{cd} \delta_{ab} \right) =& -2 \sum_{\substack{i,j,k,l=z,x\\a,c=0,x}} \lambda_{ia} \lambda_{ja} \lambda_{kc} \lambda_{lc} \epsilon_{ijy} \epsilon_{kly} \nonumber \\
    =& -2 \sum_{\substack{k,l=z,x\\a,c=0,x}} \lambda_{za} \lambda_{xa} \lambda_{kc} \lambda_{lc} \epsilon_{zxy} \epsilon_{kly} + \lambda_{xa} \lambda_{za} \lambda_{kc} \lambda_{lc} \epsilon_{xzy} \epsilon_{kly} \nonumber \\
    =& -2 \sum_{\substack{k,l=z,x\\a,c=0,x}} \lambda_{za} \lambda_{xa} \lambda_{kc} \lambda_{lc} (1-1)\epsilon_{kly} =0 \\[2em]
    \sum_{\substack{i,j,k,l=z,x\\a,b,c,d=0,x}} \lambda_{ia} \lambda_{jb} \lambda_{kc} \lambda_{ld} \left(-2 \epsilon_{ijp} \epsilon_{klp} (1-\delta_{cd}) (1-\delta_{ab}) \right) =& -2 \sum_{\substack{i,j,k,l=z,x\\a,c=0,x}} \lambda_{ia} \lambda_{j\bar{a}} \lambda_{kc} \lambda_{l\bar{c}} \epsilon_{ijy} \epsilon_{kly} \nonumber\\
    =& -2 \sum_{\substack{k,l=z,x\\a,c=0,x}} \lambda_{za} \lambda_{x\bar{a}} \lambda_{kc} \lambda_{l\bar{c}} \epsilon_{zxy} \epsilon_{kly} + \lambda_{xa} \lambda_{z\bar{a}} \lambda_{kc} \lambda_{l\bar{c}} \epsilon_{xzy} \epsilon_{kly} \nonumber \\
    &= -2 \sum_{\substack{k,l=z,x\\a,c=0,x}} \lambda_{za} \lambda_{x\bar{a}} \lambda_{kc} \lambda_{l\bar{c}} \epsilon_{zxy} \epsilon_{kly} + \lambda_{x\bar{a}} \lambda_{za} \lambda_{k\bar{c}} \lambda_{lc} \epsilon_{xzy} \epsilon_{kly} \nonumber\\
    &= -2 \sum_{\substack{k,l=z,x\\a,c=0,x}} \lambda_{za} \lambda_{x\bar{a}} \lambda_{kc} \lambda_{l\bar{c}} (1-1)\epsilon_{kly} =0
\end{align}

Now we can calculate the coefficients $\{ a_j\}$ in \eqref{eq:P_drvt_aj} using Faddeev–LeVerrier algorithm \eqref{method:Faddev-LeVerrier}:
\begin{align}
a_1 &= \Tr\hat{\rho} \\
a_2 &= \frac{1}{2}(\Tr\hat{\rho})^2 - \frac{1}{2}\Tr\hat{\rho}^2 \\
a_3 &= \frac{1}{3} \Tr \hat{\rho}^3 - \frac{1}{2}\Tr\hat{\rho} \Tr\hat{\rho}^2 + \frac{1}{6} (\Tr\hat{\rho})^3 \\
a_4 &= -\frac{1}{4}\Tr \hat{\rho}^4 + \frac{1}{3}\Tr\hat{\rho} \Tr\hat{\rho}^3 - \frac{1}{4} (\Tr\hat{\rho})^2\Tr\hat{\rho}^2 + \frac{1}{8}(\Tr\hat{\rho}^2)^2 + \frac{1}{24} (\Tr\hat{\rho})^4
\end{align}
Using ~Eq.\eqref{eq:Tr_rho1}-\eqref{eq:Tr_rho4}, we get:
\begin{align}
    a_1 =& 2\\
    a_2 =& \frac{3}{2} - \frac{1}{2} \Gamma(\bm{\lambda}) \\
    a_3 =& \frac{1}{2} - \frac{1}{2} \Gamma(\bm{\lambda}) \\
    a_4 =& \frac{1}{16}\left(1 - 2\Gamma(\bm{\lambda}) + \Gamma(\bm{\lambda})^2 - \mathcal{Z}(\bm{\lambda})^2 \right)
\end{align}
For $N=4$, we have explicitly from \eqref{eq:P_drvt_aj}:
\begin{align}
    P^{(3)}_{\rho}(x) =& 24x - 6a_1 \\
    P^{(2)}_{\rho}(x) =& 12x^2 - 6a_1 x + 2 a_2 \\
    P^{(1)}_{\rho}(x) =& 4x^3 - 3a_1 x^2 + 2a_2 x - a_3 \\
    P^{(0)}_{\rho}(x) =& x^4 - a_1 x^3 + a_2 x^2 - a_3 x + a_4
\end{align}
And the inequalities \eqref{cnstr:lowerbound} and \eqref{cnstr:upperbound} are:
\begin{align}
    P^{(j)}_{\rho} (1-\delta) \geq& 0 \\
    (-1)^{N-j}  P^{(j)}_{\rho} (\delta) \geq& 0
\end{align}
\begin{itemize}
    \item $j=3$:
        \begin{gather}
            \delta \leq \frac{1}{2}
        \end{gather}
    \item $j=2$:
        \begin{gather}
            \Gamma(\bm{\lambda}) \leq 12 \delta (\delta - 1) + 3
        \end{gather}
    \item $j=1$:
        $P^{(1)}_{\rho}(x)$ is calculated as:
        \begin{gather}
            P^{(1)}_{\rho}(x) = 4x^3 - 6x^2+3x - \frac{1}{2} + (\frac{1}{2} - x) \Gamma
        \end{gather}
        let $f(x)=4x^3 - 6x^2+3x - \frac{1}{2}$, we can check that:
        \begin{gather}
            f(x) = - f(1-x)
        \end{gather}
        Therefore, the two inequalities are equivalent and we only get one constraint:
        \begin{gather}
            \Gamma(\bm{\lambda}) \leq \frac{-2(4\delta^3 - 6\delta^2+3\delta - \frac{1}{2})}{1 - 2\delta} = (1-2\delta)^2
        \end{gather}
    \item $j=0$:
        \begin{align}
            \left[x^4 - 2x^3+\frac{3}{2}x^2 - \frac{1}{2}x + \frac{1}{16}\right] - \left(\frac{1}{2}x^2 - \frac{1}{2}x +\frac{1}{8}\right) \Gamma + \frac{1}{16}\Gamma^2 - \frac{1}{16} \mathcal{Z} \geq 0 \ \ (x=\delta \text{ and } 1-\delta) \label{ineq:j=0}
        \end{align}
        Notice that the constant term as well as the coefficient of $\Gamma$ are invariant under the transformation $x \rightarrow 1-x$. We only get one constraint:
        \begin{gather}
            \left(\frac{1}{2}\delta^2 - \frac{1}{2}\delta +\frac{1}{8}\right) \Gamma(\bm{\lambda}) - \frac{1}{16}\Gamma(\bm{\lambda})^2 + \frac{1}{16} \mathcal{Z}(\bm{\lambda})^2 \leq \left[\delta^4 - 2\delta^3+\frac{3}{2}\delta^2 - \frac{1}{2}\delta + \frac{1}{16}\right]
        \end{gather}
        multiplying both sides by 16 gives:
        \begin{gather}
            \left(8\delta^2 - 8\delta + 2\right) \Gamma(\bm{\lambda}) - \Gamma(\bm{\lambda})^2 + \mathcal{Z}(\bm{\lambda})^2 \leq \left[16 \delta^4 - 32\delta^3+24\delta^2 - 8\delta + 1\right]
        \end{gather}
\end{itemize}
To summarise, the constraints on $\bm{\lambda}$ are:
\begin{align}
    \delta \leq& \frac{1}{2} \\[1em]
    \Gamma(\bm{\lambda}) \leq& (1-2\delta)^2 \\[1em]
    \left(8\delta^2 - 8\delta + 2\right) \Gamma(\bm{\lambda}) - \Gamma(\bm{\lambda})^2 + \mathcal{Z}(\bm{\lambda})^2 \leq& \left[16 \delta^4 - 32\delta^3+24\delta^2 - 8\delta + 1\right]
\end{align}
\section{Gauge Freedom in SC Gutzwiller Wavefunction}\label{appdx:gauge}
It was first appeared to us that the numerical results of our variational approach exist certain gauge freedom, where different $\bm{\varrho}^0$ up to a unitary transformation can give the same solution $\ket{\Psi_G}$. We later realized that such gauge freedom is a unitary transformation $\mathcal{U}$ in the quasi-particle space as a consequence of charge-U(1) symmetry breaking in both $\ket{\Phi_0}$ and $\hat{P}_G$:
\begin{gather}
    \ket{\Psi_G} = \left[\hat{P}_G \hat{\mathcal{U}}^\dagger \right] \left[\hat{\mathcal{U}} \ket{\Phi_0}\right]. \label{gauge_qp}
\end{gather}
Notice that $\mathcal{U}$ \textbf{must commute with every symmetry operation of the system}. We first demonstrate it with s-wave SC state of TBG with $C_{3z}$ symmetry, whose effective $\bm{\varrho}^0$ is a $2\times2$ matrix (per valley per orbital):
\begin{gather}
    \bm{\varrho}^0_{\eta \alpha} = \bra{\Phi_0} \begin{pmatrix}
        \hat{f}^\dagger_{\alpha \eta \uparrow} \hat{f}_{\alpha \eta \uparrow} & \hat{f}_{\alpha \bar{\eta} \downarrow} \hat{f}_{\alpha \eta \uparrow} \\
          \hat{f}^\dagger_{\alpha \eta \uparrow} \hat{f}^\dagger_{\alpha \bar{\eta} \downarrow} & \hat{f}_{\alpha \bar{\eta} \downarrow} \hat{f}^\dagger_{\alpha \bar{\eta} \downarrow}
    \end{pmatrix} \ket{\Phi_0},
\end{gather}
which can be rotated in the iso-spin space formed by Nambu-spinors:
\begin{gather}
    \bm{\varrho}^0_{\eta \alpha} \to \mathcal{D}_{y}(\theta) \bm{\varrho}^0_{\eta \alpha} \mathcal{D}_{y}^\intercal (\theta) \\
    \mathcal{D}_y (\theta) = \begin{pmatrix}
        \cos\left(\frac{\theta}{2}\right) & - \sin\left(\frac{\theta}{2}\right) \\
        \sin\left(\frac{\theta}{2}\right) & \cos\left(\frac{\theta}{2}\right)
    \end{pmatrix} \label{rep:Dy}
\end{gather}
It's easy to check that $\hat{\mathcal{D}}_y(\theta)$ commutes with all the symmetry operations \eqref{def:symm_operation_f}. We can then solve a $\theta$ that diagonalizes $\bm{\varrho}^0_{\eta \alpha}$ where the anomalous terms are zero. This implies that the real variational freedom for $\mathcal{L}_{\text{SCF}}[\mu, \bm{\varrho}^0]$ of SC state equals to that of a normal state, which cuts the computational cost by half if one uses numerical derivative on $\mathcal{L}_{\text{SCF}}[\mu, \bm{\varrho}^0]$. We propose the above statement is true for any SC state calculation in general, though we don't have a rigorous proof yet. As a support of our conjecture, we solve the gauge $\mathcal{U}$ that kills the anomalous term in the s+d-wave case in the remaining part of this section.

When $C_{3z}$ is broken, the effective Nambu density matrix (per valley) is (we abbreviate $\uparrow \eta \rightarrow \Uparrow$ and $\downarrow \bar{\eta} \rightarrow \Downarrow$ for simplicity):
\begin{gather*}
    \bm{\varrho}^0_{\eta} = \langle \Phi_0|
    \begin{pmatrix}
        \hat{f}^\dagger_{\alpha \Uparrow} \hat{f}_{\alpha \Uparrow} & \hat{f}^\dagger_{\bar{\alpha} \Uparrow} \hat{f}_{\alpha \Uparrow} & \hat{f}_{\alpha \Downarrow} \hat{f}_{\alpha \Uparrow} & \hat{f}_{\bar{\alpha} \Downarrow} \hat{f}_{\alpha \Uparrow} \\
        \hat{f}^\dagger_{\alpha \Uparrow} \hat{f}_{\bar{\alpha} \Uparrow} & \hat{f}^\dagger_{\bar{\alpha} \Uparrow} \hat{f}_{\bar{\alpha} \Uparrow} & \hat{f}_{\alpha \Downarrow} \hat{f}_{\bar{\alpha} \Uparrow} & \hat{f}_{\bar{\alpha} \Downarrow} \hat{f}_{\bar{\alpha} \Uparrow} \\
        \hat{f}^\dagger_{\alpha \Uparrow} \hat{f}^\dagger_{\alpha \Downarrow} & \hat{f}^\dagger_{\bar{\alpha} \Uparrow} \hat{f}^\dagger_{\alpha \Downarrow} & \hat{f}_{\alpha \Downarrow} \hat{f}^\dagger_{\alpha \Downarrow} & \hat{f}_{\bar{\alpha} \Downarrow} \hat{f}^\dagger_{\alpha \Downarrow} \\
        \hat{f}^\dagger_{\alpha \Uparrow} \hat{f}^\dagger_{\bar{\alpha} \Downarrow} & \hat{f}^\dagger_{\bar{\alpha} \Uparrow} \hat{f}^\dagger_{\bar{\alpha} \Downarrow} & \hat{f}_{\alpha \Downarrow} \hat{f}^\dagger_{\bar{\alpha} \Downarrow} & \hat{f}_{\bar{\alpha} \Downarrow} \hat{f}^\dagger_{\bar{\alpha} \Downarrow}
    \end{pmatrix} |\Phi_0 \rangle = 
    \begin{pmatrix}
        n^0_s & n^0_d & \Delta^0_s & \Delta^0_d \\
        n^0_d & n^0_s & \Delta^0_d & \Delta^0_s \\
        \Delta^0_s & \Delta^0_d  & 1- n^0_s & -n^0_d\\
        \Delta^0_d & \Delta^0_s & -n^0_d & 1- n^0_s 
    \end{pmatrix}
\end{gather*}
Given the following unitary transform:
\begin{align*}
    \hat{\mathcal{G}}_s \hat{f}^\dagger_{\alpha \Uparrow} \hat{\mathcal{G}}_s^{-1} =& \hat{f}^\dagger_{\alpha \Uparrow} \cos(\frac{\theta_s}{2}) + \hat{f}_{\alpha  \Downarrow} \sin(\frac{\theta_s}{2}) ,\quad & \hat{\mathcal{G}}_d \hat{f}^\dagger_{\alpha \Uparrow} \hat{\mathcal{G}}_d^{-1} =& \hat{f}^\dagger_{\alpha \Uparrow} \cos(\frac{\theta_d}{2}) + \hat{f}_{\bar{\alpha}  \Downarrow} \sin(\frac{\theta_d}{2}) \\
    \hat{\mathcal{G}}_s \hat{f}_{\alpha \Downarrow} \hat{\mathcal{G}}^{-1}_s =& - \hat{f}^\dagger_{\alpha \Uparrow} \sin(\frac{\theta_s}{2}) + \hat{f}_{\alpha \Downarrow} \cos(\frac{\theta_s}{2}) ,\quad & \hat{\mathcal{G}}_d \hat{f}_{\bar{\alpha} \Downarrow} \hat{\mathcal{G}}^{-1}_d =& - \hat{f}^\dagger_{\alpha \Uparrow} \sin(\frac{\theta_d}{2}) + \hat{f}_{\bar{\alpha} \Downarrow} \cos(\frac{\theta_d}{2})
\end{align*}
Which transforms $\bm{\varrho}^0_{\eta}$ as:
\begin{gather*}
    \bm{\varrho}^0_{\eta} \xlongrightarrow{\hat{\mathcal{G}}_{s(d)}} \mathcal{D}_{s(d)}^\dagger(\theta / 2) \bm{\varrho}^0_{\eta} \mathcal{D}_{s(d)}(\theta / 2) \\[1.5em]
    \mathcal{D}_{s}(\theta_s /2) = \begin{pmatrix}
    \cos(\frac{\theta_s}{2}) & 0 & -\sin(\frac{\theta_s}{2}) & 0 \\
    0 & \cos(\frac{\theta_s}{2}) & 0 & -\sin(\frac{\theta_s}{2}) \\
    \sin(\frac{\theta_s}{2}) & 0 & \cos(\frac{\theta_s}{2}) & 0 \\
    0 & \sin(\frac{\theta_s}{2}) & 0 & \cos(\frac{\theta_s}{2})
    \end{pmatrix} \\[1em]
    \mathcal{D}_{d}(\theta_d /2) = \begin{pmatrix}
    \cos(\frac{\theta_d}{2}) & 0 & 0 & -\sin(\frac{\theta_d}{2}) \\
    0 & \cos(\frac{\theta_d}{2}) & -\sin(\frac{\theta_d}{2}) & 0 \\
    0 & \sin(\frac{\theta_d}{2}) & \cos(\frac{\theta_d}{2}) & 0 \\
    \sin(\frac{\theta_d}{2}) & 0 & 0 & \cos(\frac{\theta_d}{2})
    \end{pmatrix}   
\end{gather*}

Our goal is to solve for $\theta_s$ and $\theta_d$ such that the transformed anamolous order parameters vanishes: $\expval{\hat{\mathcal{G}}_d \hat{\mathcal{G}}_s \hat{f}_{\alpha \Uparrow} \hat{f}_{\alpha \Downarrow} \hat{\mathcal{G}}^{-1}_s \hat{\mathcal{G}}^{-1}_d}_0 = 0$ and $\expval{\hat{\mathcal{G}}_d \hat{\mathcal{G}}_s \hat{f}_{\alpha \Uparrow} \hat{f}_{\bar{\alpha} \Downarrow} \hat{\mathcal{G}}^{-1}_s \hat{\mathcal{G}}^{-1}_d}_0 = 0$. Without loss of generality, we first apply $\hat{\mathcal{G}}_s$ and the transformed density matrix is:
\begin{equation}
    \label{transG:step1}
    \begin{aligned}
        \Delta^0_s \xlongrightarrow{\hat{\mathcal{G}}_s} & \Delta^{\text{I}}_s = \frac{1 - 2 n^0_s}{2} \sin(\theta_s) + \Delta^0_s \cos(\theta_s) \\
        \Delta^{0}_d \xlongrightarrow{\hat{\mathcal{G}}_s} & \Delta^{\text{I}}_d = - n^0_d \sin(\theta_s) + \Delta^0_d \cos(\theta_s) \\
        n^0_s \xlongrightarrow{\hat{\mathcal{G}}_s} & n^{\text{I}}_s = n^0_s \cos(\theta_s) + \Delta^0_s \sin(\theta_s) + \frac{1}{2}(1 - \cos(\theta_s)) \\
        n^0_d \xlongrightarrow{\hat{\mathcal{G}}_s} & n^{\text{I}}_d = n^0_d \cos(\theta_s) + \Delta^0_d \sin(\theta_s)
    \end{aligned}
\end{equation}
Next, we apply $\hat{\mathcal{G}}_d$ to $\bm{\varrho}^{\text{I}}$:
\begin{equation}
    \label{transG:step2}
    \boxed{\begin{aligned}
        \Delta^{\text{I}}_s \xlongrightarrow{\hat{\mathcal{G}}_d} & \Delta^{\text{II}}_s = - n^{\text{I}}_d \sin(\theta_d) + \Delta^{\text{I}}_s \cos(\theta_d) =0 \\
        \Delta^{\text{I}} _d \xlongrightarrow{\hat{\mathcal{G}}_d} & \Delta^{\text{II}}_d = \frac{1 - 2 n^{\text{I}}_s}{2} \sin(\theta_d) + \Delta^{\text{I}}_d \cos(\theta_d) =0 
    \end{aligned}} 
\end{equation}
which solves $\theta_d$. The necessary and sufficient condition for \eqref{transG:step2} to be solvable is:
\begin{gather}
    \Delta^{\text{I}}_s \left(2 n^{\text{I}}_s - 1\right) = 2 \Delta^{\text{I}}_d n^{\text{I}}_d
\end{gather}
from which we can solve for $\theta_s$: 
\begin{align}
    \Delta^{\text{I}}_s \left(2 n^{\text{I}}_s - 1\right) =& \left[\left(\Delta^0_s\right)^2 - \left(n^0_s - 1/2\right)^2\right] \sin(2 \theta_s) + \Delta^0_s \left(2 n^0_s -1\right) \cos(2 \theta_s) \\
    2 \Delta^{\text{I}}_d n^{\text{I}}_d =& \left[\left(\Delta^0_d\right)^2 - \left(n^0_d\right)^2\right] \sin(2 \theta_d) + 2 n^0_d \Delta^0_d \cos(2 \theta_d)
\end{align}
Thus, $\theta_s$ is given by:
\begin{equation}
    \boxed{\left[\left(\Delta^0_s\right)^2 - \left(n^0_s - 1/2\right)^2 - \left(\Delta^0_d\right)^2 + \left(n^0_d\right)^2\right] \sin(2 \theta_s) = \left[2 n^0_d \Delta^0_d - \Delta^0_s \left(2 n^0_s -1\right)\right] \cos(2 \theta_s)} \label{sol:theta_s}
\end{equation}
We can do a sanity check by setting $\Delta^0_d = n^0_d = 0$, and the solution is the same as the lone s-wave case. The gauge makes anomalous terms vanish are solved by \eqref{transG:step2} and \eqref{sol:theta_s}.

\subsection{Projector weights in SC state}\label{appdx:project_weight}
To better understand the local correlations of SC-SC state in $\nu=2+x$ TBG and its comparison with normal states (FL and sFL states), one can factorize the projector $\hat{P}$ into components having different commutation relations with the $f$-electron number operator: 
\begin{equation}
    \begin{gathered}
        \hat{P} = \sum_{n} \hat{P}_n, \quad [\hat{P}_n, \hat{N}_f] = n \hat{P}_n, \quad \text{$n$ is even integer}
    \end{gathered}
\end{equation}
From \eqref{GA:normalization}, we have:
\begin{gather}
    1 = \sum_{mn} \bra{\Phi_0} \hat{P}^\dagger_m \hat{P}_n \ket{\Phi_0} = \sum_{mn} \sum_{\substack{J_1 J_2 I \\ \abs{J_1} - \abs{I} = m \\ \abs{J_2} - \abs{I} = n}} \bm{\Lambda}^\dagger_{m; J_1 I} \bm{\Lambda}_{n; I J_2} \bra{\Phi_0} \ket{J_1} \bra{J_2} \ket{\Phi_0} 
\end{gather}
Under the natural gauge where $\bm{\Delta}^0 = 0$ in $\bm{\varrho}^0$, the charge U(1) breaking off-diagonal terms in $\bm{m}^0$ vanish, thus $\bra{\Phi_0} \ket{J_1} \bra{J_2} \ket{\Phi_0} \propto \delta_{\abs{J_1}, \abs{J_2}}$. Therefore, the cross terms with $m \neq n$ vanish and can define the projector weights as:
\begin{gather}
    1 = \sum_{n} \bra{\Phi_0} \hat{P}^\dagger_n \hat{P}_n \ket{\Phi_0}, \quad W_n \defeq \bra{\Phi_0} \hat{P}^\dagger_n \hat{P}_n \ket{\Phi_0} \label{def:Wn}
\end{gather} 

\subsection{Gauge transformation of $\MC{R}$ and $\MC{Q}$}
The gauge transformation on Nambu basis operators can be conveniently defined as:
\begin{equation}
    \label{eq:nambu_gauge}
   \begin{aligned}
    \hat{\mathcal{U}}^\dagger \hat{f}_{\uparrow \beta} \hat{\mathcal{U}} =& \sum_{\beta'}U_{\uparrow \beta, \uparrow \beta'} \hat{f}_{\uparrow \beta'} + U_{\uparrow \beta, \downarrow \beta'} \hat{f}^\dagger_{\downarrow \beta'} \\
    \hat{\mathcal{U}}^\dagger \hat{f}^\dagger_{\downarrow \beta} \hat{\mathcal{U}} =& \sum_{\beta'}U_{\downarrow \beta, \uparrow \beta'} \hat{f}_{\uparrow \beta'} + U_{\downarrow \beta, \downarrow \beta'} \hat{f}^\dagger_{\downarrow \beta'} 
    \end{aligned} 
\end{equation}

Such that under $\ket{\Phi_0} \to \hat{\mathcal{U}} \ket{\Phi_0}$, $\bm{\varrho}^0$ transforms as:
\begin{gather}
    \bm{\varrho}^0 \to \mathcal{U} \bm{\varrho}^0 \mathcal{U}^\dagger
\end{gather}
Correspondinly, under $\hat{P} \to \hat{P} \mathcal{U}^\dagger$ and $\ket{\Phi_0} \to \hat{\mathcal{U}} \ket{\Phi_0}$, the equation \eqref{eq:RQ_solution} transforms as:
\begin{align}
    \begin{pmatrix}
        \tilde{\MC{R}} \\ \tilde{\MC{Q}}
    \end{pmatrix} =& \mathcal{U} (\bm{\varrho}^0)^{-1} \mathcal{U}^\intercal \begin{pmatrix}
        \bra{\Phi_0} \hat{\mathcal{U}}^\intercal \hat{\mathcal{U}} \hat{P}^\dagger_{i} \hat{f}^\dagger_{i \uparrow \alpha } \hat{P}_{i} \hat{\mathcal{U}}^\intercal \hat{f}_{i \uparrow \beta} \hat{\mathcal{U}} \ket{\Phi_0} \\
        \bra{\Phi_0} \hat{\mathcal{U}}^\intercal \hat{\mathcal{U}} \hat{P}^\dagger_{i} \hat{f}^\dagger_{i \uparrow \alpha } \hat{P}_{i} \hat{\mathcal{U}}^\intercal \hat{f}^\dagger_{i \downarrow \beta} \hat{\mathcal{U}} \ket{\Phi_0}
    \end{pmatrix} \\
    =& \mathcal{U} (\bm{\varrho}^0)^{-1} \mathcal{U}^\intercal \mathcal{U} \begin{pmatrix}
        \mathcal{K}_{\uparrow \uparrow} \\ \mathcal{K}_{\downarrow \uparrow}
    \end{pmatrix} \\
    =& \mathcal{U} \begin{pmatrix}
        \MC{R} \\ \MC{Q}
    \end{pmatrix} \label{eq:sqrtZ_gauge}
\end{align}
A immediate consequence is that the quasi-particle weight matrix:
\begin{gather}
    Z = \begin{pmatrix}
        \MC{R}^\dagger & \MC{Q}^\dagger
    \end{pmatrix} \begin{pmatrix}
        \MC{R} \\ \MC{Q}
    \end{pmatrix}
\end{gather}
is invariant under $\mathcal{U}$.

\section{Physical observables under $\ket{\Psi_G}$}

\subsection{Momentum dependent occupation number: $n_{\vb{k}}$}
For uncorrelated $c$-electrons, $n_{c; \vb{k}}$ equals the uncorrelated occupation number $n^0_{c; \vb{k}}$ as the Gutzwiller projection does not act on $c$-electrons:
\begin{equation}
    n_{c; \vb{k} a \sigma} = \bra{\Psi_G} \hat{c}^\dagger_{\vb{k} a \sigma} \hat{c}_{\vb{k} a \sigma} \ket{\Psi_G} = \bra{\Phi_0} \hat{c}^\dagger_{\vb{k} a \sigma}\hat{c}_{\vb{k} a \sigma} \ket{\Phi_0} = n^0_{c; \vb{k} a \sigma}
\end{equation}
The momentum dependent occupation number of $f$-electrons is given by:
\begin{equation}
    \begin{aligned}
        n_{f; \vb{k} \sigma \alpha} =&  \bra{\Psi_G} \hat{f}^\dagger_{\vb{k} \sigma \alpha} \hat{f}_{\vb{k} \sigma \alpha} \ket{\Psi_G} \\
        =& \frac{1}{N_k} \left[\sum_i \bra{\Phi_0} \hat{P}^\dagger_i \hat{f}^\dagger_{i \sigma \alpha} \hat{f}_{i \sigma \alpha} \hat{P}_i \ket{\Phi_0} + \sum_{i \neq j} e^{i \vb{k} \cdot (\vb{R}_i - \vb{R}_j)}\bra{\Phi_0}\hat{P}^\dagger_i \hat{f}^\dagger_{i \sigma \alpha} \hat{P}_i \hat{P}^\dagger_j \hat{f}_{j \sigma \alpha} \hat{P}_j \ket{\Phi_0} \right] \\
        =& \frac{1}{N_k} \left[N_{\vb{k}} n_{f \sigma \alpha} + \sum_{i \neq j} e^{i \vb{k} \cdot (\vb{R}_i - \vb{R}_j)}\bra{\Phi_0}\hat{P}^\dagger_i \hat{f}^\dagger_{i \sigma \alpha} \hat{P}_i \hat{P}^\dagger_j \hat{f}_{j \sigma \alpha} \hat{P}_j \ket{\Phi_0} \right]
    \end{aligned}
\end{equation}
To get the sense of how $n_{f; \vb{k}}$ is related to $n^0_{f; \vb{k}}$, we can first consider a spin-orbital-less case where the quasi-particle is given by:
\begin{gather}
    \hat{P}^\dagger_i \hat{f}^\dagger_{\vb{R}_i} \hat{P}_i = \sqrt{Z} \hat{f}^\dagger_{\vb{R}_i}
\end{gather}
such that the intersite contribution to $n_{f; \vb{k}}$ is given by:
\begin{equation}
    \begin{aligned}
        \bra{\Phi_0}\hat{P}^\dagger_i \hat{f}^\dagger_{i \sigma \alpha} \hat{P}_i \hat{P}^\dagger_j \hat{f}_{j \sigma \alpha} \hat{P}_j \ket{\Phi_0} = & Z \bra{\Phi_0} \hat{f}^\dagger_{i} \hat{f}_{j} \ket{\Phi_0} \\
        =& Z \frac{1}{N_k} \sum_{\vb{k}'} e^{i \vb{k}' \cdot (\vb{R}_j - \vb{R}_i)} n^0_{f; \vb{k}'}
    \end{aligned}
\end{equation}
Therefore, we have:
\begin{equation}
    \begin{aligned}
        n_{f; \vb{k}} =& n_{f} + \left(\sum_{ij} e^{i \vb{k} \cdot (\vb{R}_i - \vb{R}_j)} - \sum_{ij} \delta_{ij}\right) \left( Z \frac{1}{N_k} \sum_{\vb{k}} e^{i \vb{k} \cdot (\vb{R}_j - \vb{R}_i)} n^0_{f; \vb{k}} \right) \\
        =& n_f + Z \left(n^0_{f; \vb{k}} - \frac{1}{N_{k}} \sum_{\vb{k}'} n^0_{f; \vb{k}'}\right) \\
        =& Z n^0_{f; \vb{k}} + \left(n_f - Z n^0_{f}\right)
    \end{aligned}
\end{equation}
Which satisfy the sum rule: $$\frac{1}{N_k} \sum_{\vb{k}} n_{f; \vb{k}} = n_f.$$
Now we work on the general case where the quasi-particle operator is given in \eqref{eq:renrml_substitution}. The intersite contribution is given in \eqref{cal:non-local pairing and hopping} and can be compactly expressed in terms of renormalised single particle density matrix in momentum space $\tilde{\bm{\varrho}}^0_{\vb{k}}$:
\begin{gather}
    \bra{\Psi_G} \hat{f}^\dagger_{i \uparrow \alpha} \hat{f}_{j \uparrow \alpha} \ket{\Psi_G} = \frac{1}{N_{\vb{k}}} \sum_{\vb{k}} e^{i \vb{k} \cdot (\vb{R}_j - \vb{R}_i)} \tilde{\bm{\varrho}}^0_{\vb{k} \uparrow \uparrow, \alpha \alpha} 
\end{gather}
Therefore, the momentum resolved occupation number $n_{f; \vb{k} \sigma \alpha}$ is:
\begin{align}
    n_{f; \vb{k} \sigma \alpha} =& n_{f \sigma \alpha} + \tilde{\bm{\varrho}}^0_{\vb{k} \sigma \sigma, \alpha \alpha} - \frac{1}{N_k} \sum_{\vb{k}'} \tilde{\bm{\varrho}}^0_{\vb{k}' \sigma \sigma, \alpha \alpha} \\
    =& \tilde{\bm{\varrho}}^0_{\vb{k} \sigma \sigma, \alpha \alpha} + \left(n_{f \sigma \alpha} - n^0_{f \sigma \alpha}\right)
\end{align}
The total momentum resolved occupation number for each spin/valley is given by:
\begin{gather}
    n_{\vb{k} \sigma} = \sum_{\alpha} n_{f; \vb{k} \sigma \alpha} + \sum_{a} n_{c; \vb{k} \sigma a}
\end{gather}
For normal state, $n_{\vb{k} \sigma}$ is discontinous at the Fermi surface. In the superconducting state, such discontinuities of $n_{\vb{k} \sigma}$ is smoothened out by non-zero order parameter $\Delta_{\vb{k} \sigma}$, while nodal lines are indicated by discontinuities of $n_{\vb{k} \sigma}$.

\subsection{Intersite pairing amplitudes} \label{appdx:NN_pairing}
The $f$-electron pairing amplitude are defined as:
\begin{equation}\label{pairing:fifj}
    \begin{aligned}
        \langle \hat{f}_{\downarrow \alpha \vb{R}_i} \hat{f}_{\uparrow \beta \vb{R}_j} \rangle_G = & \delta_{ij} \langle \hat{f}_{\downarrow \alpha \vb{R}_i} \hat{f}_{\uparrow \beta \vb{R}_i} \rangle_G + (1 - \delta_{ij}) \langle \hat{f}_{\downarrow \alpha \vb{R}_i} \hat{f}_{\uparrow \beta \vb{R}_j} \rangle_G \\[1.5em]
        \delta_{ij} \langle \hat{f}_{\downarrow \alpha \vb{R}_i} \hat{f}_{\uparrow \beta \vb{R}_i} \rangle_G =& \delta_{ij} \sum_{\Gamma \Gamma'} \Lambda_{\Gamma} \Lambda_{\Gamma'} \Tr(\mbbm{m}^\intercal_{\Gamma} \mbbm{D}_{\downarrow \alpha, \uparrow \beta} \mbbm{m}_{\Gamma'} \bm{m}^0) \\[1em]
        (1 - \delta_{ij}) \langle \hat{f}_{\downarrow \alpha \vb{R}_i} \hat{f}_{\uparrow \beta \vb{R}_j} \rangle_G =& (1 - \delta_{ij}) \sum_{\gamma \delta} \langle \left(\hat{f}_{\downarrow \gamma \vb{R}_i} \MC{R}^*_{\gamma \alpha} - \hat{f}^\dagger_{\uparrow \gamma \vb{R}_i} \MC{Q}^*_{\gamma \alpha} \right) \left( \hat{f}_{\uparrow \delta \vb{R}_j} \MC{R}^*_{\delta \beta} + \hat{f}^\dagger_{\downarrow \delta \vb{R}_j} \MC{Q}^*_{\delta \beta} \right) \rangle_0 \\
        =& \begin{aligned}[t]
            (1 - \delta_{ij}) \frac{1}{N} \sum_{\vb{k}} e^{i \vb{k} \cdot (\vb{R}_j - \vb{R}_i)} \Big[&\MC{R}^\dagger_{\beta \delta} \langle \hat{f}^\dagger_{\uparrow \gamma \vb{k}} \hat{f}_{\uparrow \delta \vb{k}} \rangle_0 \left(-\MC{Q}^*_{\gamma \alpha}\right) + \MC{Q}^\dagger_{\beta \delta} \langle \hat{f}^\dagger_{\uparrow \gamma \vb{k}} \hat{f}^\dagger_{\downarrow \delta -\vb{k}} \rangle_0 \left(-\MC{Q}^*_{\gamma \alpha}\right) \\
            &+\MC{R}^\dagger_{\beta \delta} \langle \hat{f}_{\downarrow \gamma - \vb{k}} \hat{f}_{\uparrow \delta \vb{k}} \rangle_0 \MC{R}^*_{\gamma \alpha} + \MC{Q}^\dagger_{\beta \delta} \langle \hat{f}_{\downarrow \gamma - \vb{k}} \hat{f}^\dagger_{\downarrow \delta -\vb{k}} \rangle_0 \MC{R}^*_{\gamma \alpha} \Big]
        \end{aligned}
    \end{aligned}
\end{equation}
Where the $\vb{k}$ dependent terms inside the square bracket from the last line of intersite pairing amplitude $\langle \hat{f}_{\downarrow \alpha \vb{R}_i} \hat{f}_{\uparrow \beta \vb{R}_j} \rangle_G$ can be identified as the anomalous part of the renormalised single particle density matrix in momentum space:

\begin{gather}
    \bm{\varrho}^{G}_{ff}(\vb{k}) \defeq \begin{pmatrix}
        \bm{\varrho}^{G}_{ff}(\vb{k})_{\uparrow \uparrow} & \bm{\varrho}^{G}_{ff}(\vb{k})_{\uparrow \downarrow} \\[0.75em]
        \bm{\varrho}^{G}_{ff}(\vb{k})_{\downarrow \uparrow} & \bm{\varrho}^{G}_{ff}(\vb{k})_{\downarrow \downarrow}
    \end{pmatrix}
     = \begin{pmatrix}
        \MC{R}^\dagger & \MC{Q}^\dagger \\[0.75em]
        -\MC{Q}^\intercal & \MC{R}^\intercal
    \end{pmatrix}
    \begin{pmatrix}
        \bm{\rho}^0_{ff}(\vb{k}) & \bm{\Delta}^0_{ff}(\vb{k}) \\[0.75em]
        \bm{\Delta}^0_{ff}(\vb{k})^\dagger & \bm{1} - \bm{\rho}^0_{ff}(-\vb{k})
    \end{pmatrix}
    \begin{pmatrix}
        \MC{R} & - \MC{Q}^* \\[0.75em]
        \MC{Q} & \MC{R}^*
    \end{pmatrix}.
\end{gather}
Therefore, $\langle \hat{f}_{\downarrow \alpha \vb{R}_i} \hat{f}_{\uparrow \beta \vb{R}_j} \rangle_G$ can be compactly written as:
\begin{gather}
    (1 - \delta_{ij}) \langle \hat{f}_{\downarrow \alpha \vb{R}_i} \hat{f}_{\uparrow \beta \vb{R}_j} \rangle_G = (1 - \delta_{ij}) \frac{1}{N} \sum_{\vb{k}} e^{i \vb{k} \cdot (\vb{R}_j - \vb{R}_i)} (\bm{\varrho}^{G}_{ff}(\vb{k})_{\uparrow \downarrow})_{\beta \alpha}
\end{gather}
From here we work on TBG and we abbreviate: $\Uparrow \defeq \uparrow \eta$, $\Downarrow \defeq \downarrow \bar{\eta}$. Using the spin-independent $C_{2z} \mathcal{T}$ symmetry, one can relate the $f$- pairing amplitudes with their complex conjugates as following (for $i \neq j$):
\begin{align}
    \langle \hat{f}_{\Downarrow \alpha \vb{R}_i} \hat{f}_{\Uparrow \beta \vb{R}_j} \rangle_G =&  \frac{1}{N} \sum_{\vb{k}} e^{i \vb{k} \cdot (\vb{R}_j - \vb{R}_i)} \langle \hat{f}_{\Downarrow - \vb{k} \alpha} \hat{f}_{\Uparrow \vb{k} \beta} \rangle_G \\
    =& \frac{1}{N} \sum_{\vb{k}} e^{i \vb{k} \cdot (\vb{R}_j - \vb{R}_i)} \langle C_{2z}\mathcal{T} \hat{f}_{\Downarrow - \vb{k} \alpha} \hat{f}_{\Uparrow \vb{k} \beta} (C_{2z}\mathcal{T})^{-1} \rangle^*_G \\
    =& \frac{1}{N} \sum_{\vb{k}} e^{i \vb{k} \cdot (\vb{R}_j - \vb{R}_i)} \langle \hat{f}_{\Downarrow - \vb{k} \bar{\alpha}} \hat{f}_{\Uparrow \vb{k} \bar{\beta}}  \rangle^*_G \\
    =& \langle \hat{f}_{\Downarrow \bar{\alpha} \vb{R}_j} \hat{f}_{\Uparrow \bar{\beta} \vb{R}_i} \rangle_G^*,
\end{align}
such that one could define real non-local pairing amplitudes as:
\begin{gather}
    \Delta^s_{f;\alpha \alpha ij} = \frac{1}{2} \left(\langle \hat{f}_{\Downarrow \alpha \vb{R}_i} \hat{f}_{\Uparrow \alpha \vb{R}_j} \rangle_G + \langle \hat{f}_{\Downarrow \bar{\alpha} \vb{R}_j} \hat{f}_{\Uparrow \bar{\alpha} \vb{R}_i} \rangle_G\right), \\
    \Delta^d_{f;\alpha \bar{\alpha} ij} = \frac{1}{2} \left(\langle \hat{f}_{\Downarrow \alpha \vb{R}_i} \hat{f}_{\Uparrow \bar{\alpha} \vb{R}_j} \rangle_G + \langle \hat{f}_{\Downarrow \bar{\alpha} \vb{R}_j} \hat{f}_{\Uparrow \alpha \vb{R}_i} \rangle_G\right).
\end{gather}
Under $C_{2x}$, we have:
\begin{align}
    \langle \hat{f}_{\Downarrow \alpha \vb{R}_i} \hat{f}_{\Uparrow \beta \vb{R}_j} \rangle_G =& \frac{1}{N} \sum_{\vb{k}} e^{i [\mathcal{D}(C_{2x})\vb{k}] \cdot [\mathcal{D}(C_{2x})(\vb{R}_j - \vb{R}_i)]} \langle C_{2x} \hat{f}_{\Downarrow - \vb{k} \alpha} \hat{f}_{\Uparrow \vb{k} \beta} (C_{2x})^{-1} \rangle_G \\
    =& \frac{1}{N} \sum_{\vb{k}} e^{i [\mathcal{D}(C_{2x})\vb{k}] \cdot [\mathcal{D}(C_{2x})(\vb{R}_j - \vb{R}_i)]} \langle \hat{f}_{\Downarrow - \mathcal{D}(C_{2x}) \vb{k} \bar{\alpha}} \hat{f}_{\Uparrow \mathcal{D}(C_{2x}) \vb{k} \bar{\beta}} \rangle_G \\
    =& \frac{1}{N} \sum_{\vb{k}} e^{i \vb{k} \cdot [\mathcal{D}(C_{2x})(\vb{R}_j - \vb{R}_i)]} \langle \hat{f}_{\Downarrow - \vb{k} \bar{\alpha}} \hat{f}_{\Uparrow \vb{k} \bar{\beta}} \rangle_G \\
    =& \langle \hat{f}_{\Downarrow \bar{\alpha} \mathcal{D}(C_{2x})\vb{R}_i} \hat{f}_{\Uparrow \bar{\beta} \mathcal{D}(C_{2x})\vb{R}_j} \rangle_G.
\end{align}
A combination of SU(2) spin rotation (along y-axis by $\pi$) and the global $C_{2z}$ or $\mathcal{T}$ symmetry leads to (We need to explicitly write out the valley indices here since $C_{2z}$ switches the valleys and levaes the spin unchanged while $\mathcal{S}_y(\pi)$ switches the spins and leaves the valleys unchanged):
\begin{align}
    \langle \hat{f}_{\downarrow \bar{\eta} \alpha \vb{R}_i} \hat{f}_{\uparrow \eta \beta \vb{R}_j} \rangle_G =& \langle \mathcal{S}_y(\pi) \hat{f}_{\downarrow \bar{\eta} \alpha \vb{R}_i} \hat{f}_{\uparrow \eta \beta \vb{R}_j} \mathcal{S}^{-1}_y (\pi) \rangle_G \\
    =& - \langle \hat{f}_{\uparrow \bar{\eta} \alpha \vb{R}_i} \hat{f}_{\downarrow \eta \beta \vb{R}_j} \rangle_G \\
    =& \frac{1}{N} \sum_{\vb{k}} e^{i \vb{k} \cdot (\vb{R}_i - \vb{R}_j)} \langle \hat{f}_{\downarrow \eta \beta -\vb{k}} \hat{f}_{\uparrow \bar{\eta} \alpha \vb{k}} \rangle_G \\
    =& \frac{1}{N} \sum_{\vb{k}} e^{i \vb{k} \cdot (\vb{R}_i - \vb{R}_j)} \langle C_{2z} \hat{f}_{\downarrow \eta \beta -\vb{k}} \hat{f}_{\uparrow \bar{\eta} \alpha \vb{k}} C_{2z}^{-1} \rangle_G \\
    =& \frac{1}{N} \sum_{\vb{k}} e^{i \vb{k} \cdot (\vb{R}_i - \vb{R}_j)} \langle \hat{f}_{\downarrow \bar{\eta} \bar{\beta} \vb{k}} \hat{f}_{\uparrow \eta \bar{\alpha} -\vb{k}} \rangle_G \\
    =& \langle \hat{f}_{\downarrow \bar{\eta} \bar{\beta} \vb{R}_i} \hat{f}_{\uparrow \eta \bar{\alpha} \vb{R}_j} \rangle_G, \\
    \langle \hat{f}_{\downarrow \bar{\eta} \alpha \vb{R}_i} \hat{f}_{\uparrow \eta \beta \vb{R}_j} \rangle_G =& \frac{1}{N} \sum_{\vb{k}} e^{i \vb{k} \cdot (\vb{R}_i - \vb{R}_j)} \langle \mathcal{T} \hat{f}_{\downarrow \eta \beta -\vb{k}} \hat{f}_{\uparrow \bar{\eta} \alpha \vb{k}} \mathcal{T}^{-1} \rangle^*_G \\
    =& \frac{1}{N} \sum_{\vb{k}} e^{i \vb{k} \cdot (\vb{R}_i - \vb{R}_j)} \langle \hat{f}_{\downarrow \bar{\eta} \beta \vb{k}} \hat{f}_{\uparrow \eta \alpha -\vb{k}} \rangle^*_G \\
    =& \langle \hat{f}_{\downarrow \bar{\eta} \beta \vb{R}_j} \hat{f}_{\uparrow \eta \alpha \vb{R}_i} \rangle^*_G
\end{align}
which implies that:
\begin{gather}
    \Delta^s_{f;\alpha \alpha ij} = \Delta^s_{f;\bar{\alpha} \bar{\alpha} ij}, \quad \Delta^d_{f;\alpha \bar{\alpha} ij} = \Delta^d_{f;\bar{\alpha} \alpha ji}.
\end{gather}
In summary, the nearest neighbor $f$-electron pairing amplitudes satisfy the following relations with $C_{3z}$ breaking: (we abbreviate $\langle \hat{f}_{\Downarrow \alpha \vb{R}_i} \hat{f}_{\Uparrow \beta \vb{R}_j} \rangle_G \defeq \Delta^f_{\Uparrow\beta; \Downarrow \alpha}(\bm{\delta}), \quad \bm{\delta} \defeq \vb{R}_j - \vb{R}_i$)
\begin{align}
    \Delta^f_{\Uparrow\beta; \Downarrow \alpha}(\bm{\delta}) =& \Delta^f_{\Uparrow \bar{\beta}; \Downarrow \bar{\alpha}}(\bm{-\delta})^*, \quad (C_{2z} \mathcal{T})\\
    \Delta^f_{\Uparrow\beta; \Downarrow \alpha}(\bm{\delta}) =& \Delta^f_{\Uparrow \bar{\beta}; \Downarrow \bar{\alpha}}(\mathcal{D}(C_{2x})\bm{\delta}), \quad (C_{2x}) \\
    \Delta^f_{\Uparrow\beta; \Downarrow \alpha}(\bm{\delta}) =& \Delta^f_{\Uparrow \bar{\alpha}; \Downarrow \bar{\beta}}(\bm{\delta}), \quad (\mathcal{S}_y(\pi) C_{2z}) \\
    \Delta^f_{\Uparrow\beta; \Downarrow \alpha}(\bm{\delta}) =& \Delta^f_{\Uparrow \alpha; \Downarrow \beta}(-\bm{\delta})^*, \quad (\mathcal{S}_y(\pi) \mathcal{T})
\end{align}
\begin{figure}[h]
    \centering
    \includegraphics[width=0.48\textwidth]{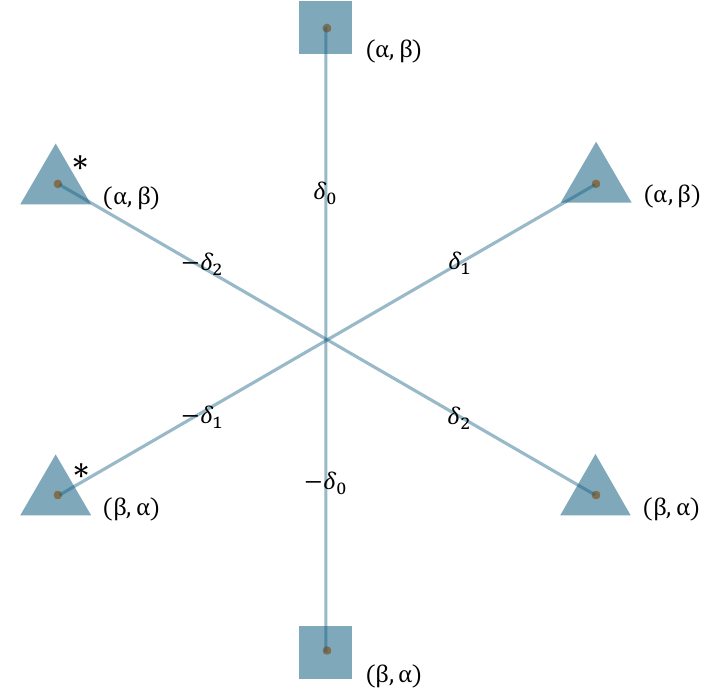}
    \caption{Structure of nearest neighboring pairing amplituides among $f$-electrons on the triangular Moir\'e lattice. }
    \label{fig:nnpairing_structure}
\end{figure}
The independent nearest pairing amplitudes and bonds $\bm{\delta}_0, \bm{\delta}_1, \bm{\delta}_2$ are illustrated in Fig.~\ref{fig:nnpairing_structure}. 


\subsection{Fermi surface volume of sFL}\label{appdx:FS_sFL}
We first postulate that the quasi-particle dispersion $\epsilon_n(\vb{k})$ in sFL is still given by the dispersion of the Fermi Hamiltonian $\hat{H}^F = \sum_{\vb{k} n} \epsilon_n(\vb{k}) \hat{d}^\dagger_{\vb{k} n} \hat{d}_{\vb{k} n}$ \cite{Bünemann2003_PRB}. The Fermi surface volume equates to the quasi-particle (uncorrelated) occupation number: 
\begin{gather}
    V_{FS} = \frac{\Omega_d}{(2\pi)^d} \sum_n \int_{BZ} d\vb{k}^d \Theta(\mu -\epsilon_n(\vb{k})) = \frac{1}{N_k} \sum_{\vb{k} n} \Theta(\mu - \epsilon_n(\vb{k})) = \bra{\Phi_0} \sum_{\alpha} \hat{n}_{\alpha} \ket{\Phi_0} = \sum_{\alpha} n^0_{\alpha}
\end{gather}
Assuming all orbitals are correlated so that their quasi-particle operators are modified by $\hat{P}_G$, the physical occupation number 
is given by:
\begin{align}
    \sum_{\alpha} n_{\alpha} =& \bra{\Phi_0} \hat{P}_G^\dagger \left( \sum_{\alpha} \hat{n}_{i; \alpha} \right) \hat{P}_G \ket{\Phi_0} = \bra{\Phi_0} \hat{P}_i^\dagger \left( \sum_{\alpha} \hat{n}_{i; \alpha} \right) \hat{P}_i \ket{\Phi_0} \\
    =& \bra{\Phi_0} \hat{P}^\dagger_i \hat{P}_i \left(\sum_{\alpha} \hat{n}_{i; \alpha} - 2\right) \ket{\Phi_0} \\
    =& \sum_{\alpha} n^0_{\alpha} - 2
\end{align}
where the last equality is just the Gutzwiller constraint. We show that the sFL state's occupation number differs from the Fermi surface volume by 2 due to the quenching of a local singlet per site.
\section{More numerical results}
\subsection{$s$+$d$-wave state}
Most of the SC area in the phase diagram Fig.~\ref{fig:phase_diagram} is occupied by $s$+$d$ state and it only appears at finite U. Here we study one point in the phase diagram $U=10 \, \mathrm{meV}$ and $J_A = 3\, \mathrm{meV}$ where all three SC phases can be stablized as an example to explain it. We see from the band structure Fig.~\ref{fig:s+d_band} that $s$+$d$-wave state opens a gap of the flat bands at moir\'e $\vb{K}$ points because the $s$+$d$-wave contains mixing components of both $\ket{\psi_s}$ and $\ket{\psi_d}$ which leads to a finite inter-orbital order $n_d$ (Notice that in Table.~\ref{table:phases} we also state $d$-wave has finite $n_d$, because the transition from $d$-wave to $s$+$d$-wave is second order at large U and we use $\abs{\Delta_s} / \abs{\Delta}_d >0.05$ instead of $n_d$ as a threshold, but $n_d$ is still very small for $d$-wave. For small $U$ like $10 \,\mathrm{meV}$, $n_d \approx0$ in $d$-wave case). Such `emergent order' in $s$+$d$-wave makes it lower U more efficient than $s$-wave or $d$-wave, as it is shown in the data under band structures in Fig.~\ref{fig:s+d_band}.
\begin{figure}
    \centering
    \includegraphics[width=1\textwidth]{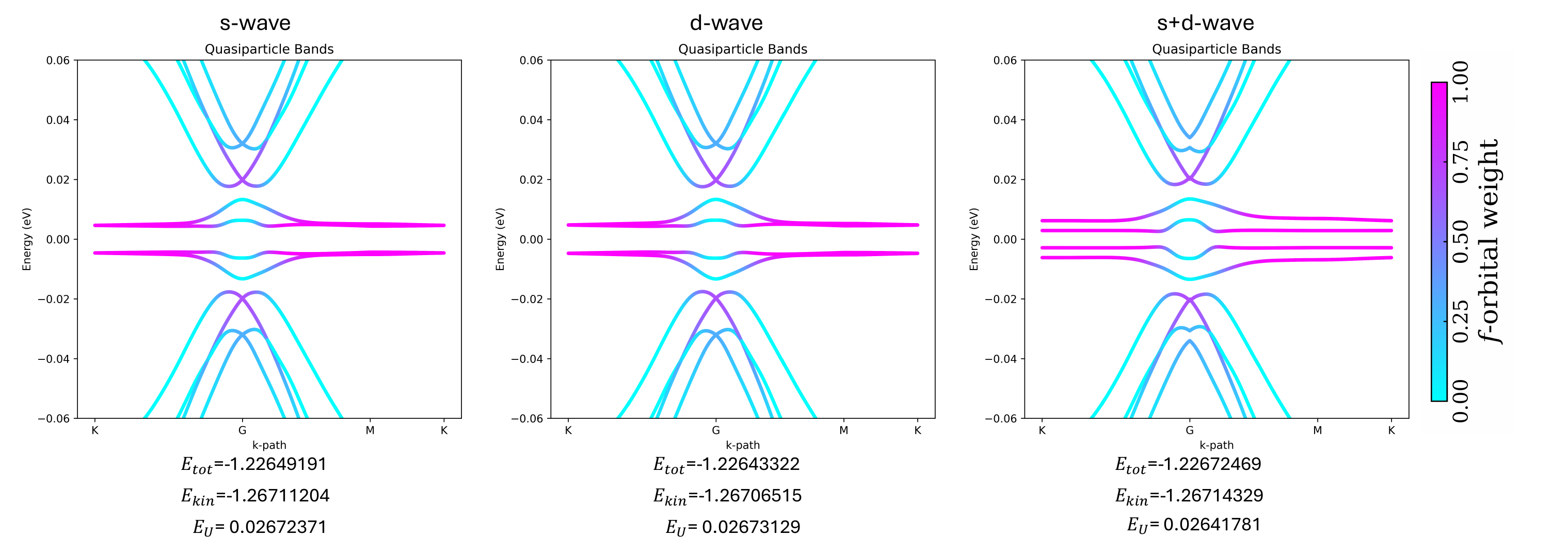}
    \caption{BdG quasiparticle band structure of the $s$-, $d$- and $s$+$d$-wave states respectively at $\nu=2.5$, $U=10 \, \mathrm{meV}$, $J_A = 3.0 \, \mathrm{meV}$ and $J_H = 1.5 \, \mathrm{meV}$. The $s$+$d$-wave state opens a gap of the flat bands at moir\'e $\vb{K}$ points by having finite inter-orbital normal order $n_d$, thus lowering the total energy compared with the pure $s$- and $d$-wave states. The color of the bands indicates the weight of $f$-orbitals. The energy units are $\mathrm{meV}$}
    \label{fig:s+d_band}
\end{figure}

\subsection{SC gap reconstruction driven by \(U\)} \label{appdx:gap_reconstruction}
We plot density of states (DOS) Fig.~\ref{fig:DOS_U} at different $U$ values with $(J_A, J_H)=(3\mathrm{meV}, 1.5\mathrm{meV})$ and find that there's a reconstruction of the SC gap structure at $U=51 \mathrm{meV}$ which corresponds to the kink in the quasi-particle weight Fig.~\ref{fig:phase_diagram}(b).
\begin{figure}[H]
    \centering
    \includegraphics[width=1\textwidth]{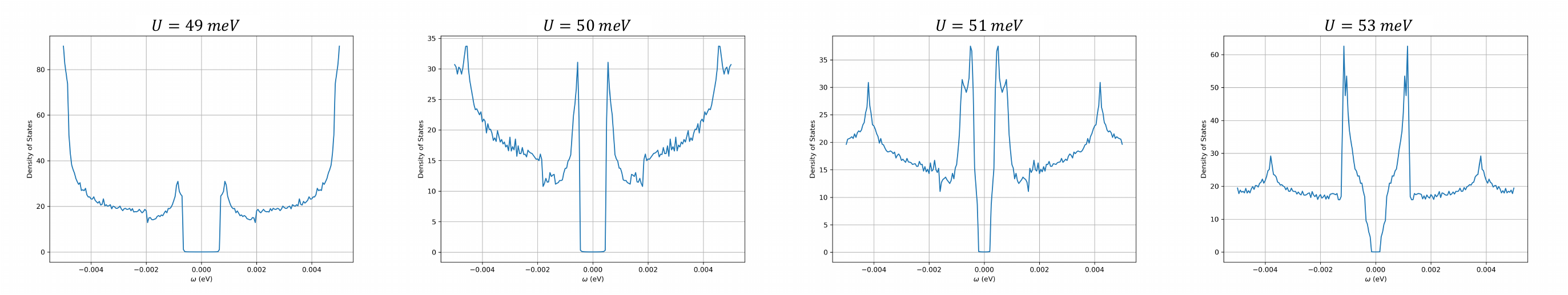}
    \caption{Density of states (DOS) at different $U$ values with $(J_A, J_H)=(3\mathrm{meV}, 1.5\mathrm{meV})$.}
    \label{fig:DOS_U}
\end{figure}
The occurance of nodal V-shape DoS relates to the second SC gap reconstruction at $U=55 \mathrm{meV}$ shown in Fig.~\ref{fig:DOS_Vshape}
\begin{figure}[H]
    \centering
    \includegraphics[width=0.9\textwidth]{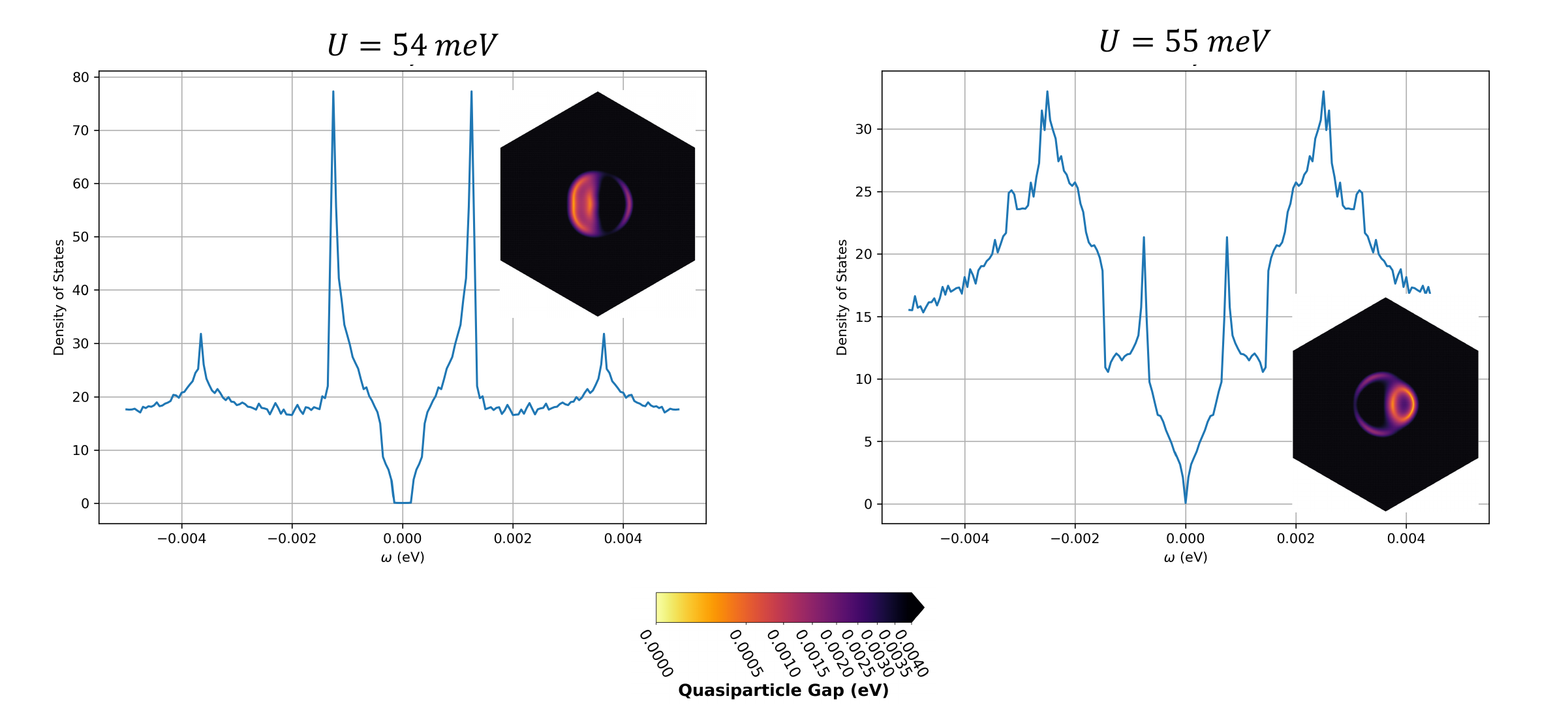}
    \caption{Density of states (DOS) at different $U$ values with $(J_A, J_H)=(3\mathrm{meV}, 1.5\mathrm{meV})$. The inset shows the SC gap structure at corresponding $U$ values.}
    \label{fig:DOS_Vshape}
\end{figure}

\subsection{SC gap structure and normal state Fermi surface} \label{appdx:Uvar_JAvar}
Here we present the SC gap contour plots and normal state Fermi surfaces at different $U$ or $J_A$ values.
\begin{figure}[H]
\centering
\includegraphics[width=0.88\textwidth]{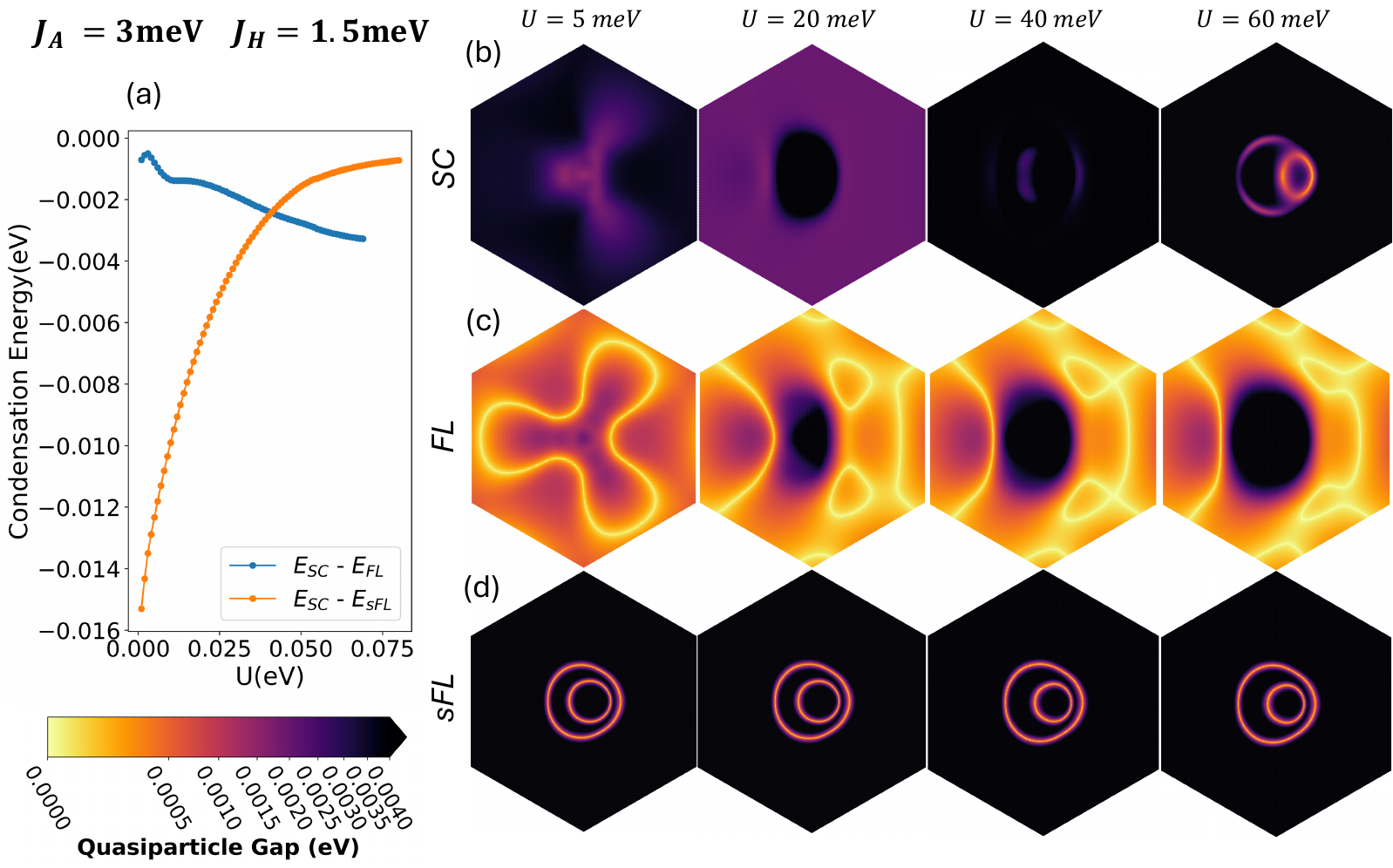}
\caption{Calculations at $\nu=2.5$, $(J_A, J_H)=(3\mathrm{meV}, 1.5\mathrm{meV})$: (a) Condensation energy versus $U$ of the SC state against the FL and sFL states. (b), (c) and (d) are the quasiparticle gap contour plots near (Bogoliubov) Fermi surface of SC, FL and sFL states respectively at $U=5, 20, 40, 60 \ \mathrm{meV}$. The color bar is normalised to show the gap structure within $4\mathrm{meV}$ where the bright yellow lines indicates Fermi surface/nodal lines.
}
\label{fig:gap_contour_Uvary}
\end{figure}

\begin{figure}[H]
\centering
\includegraphics[width=0.88\textwidth]{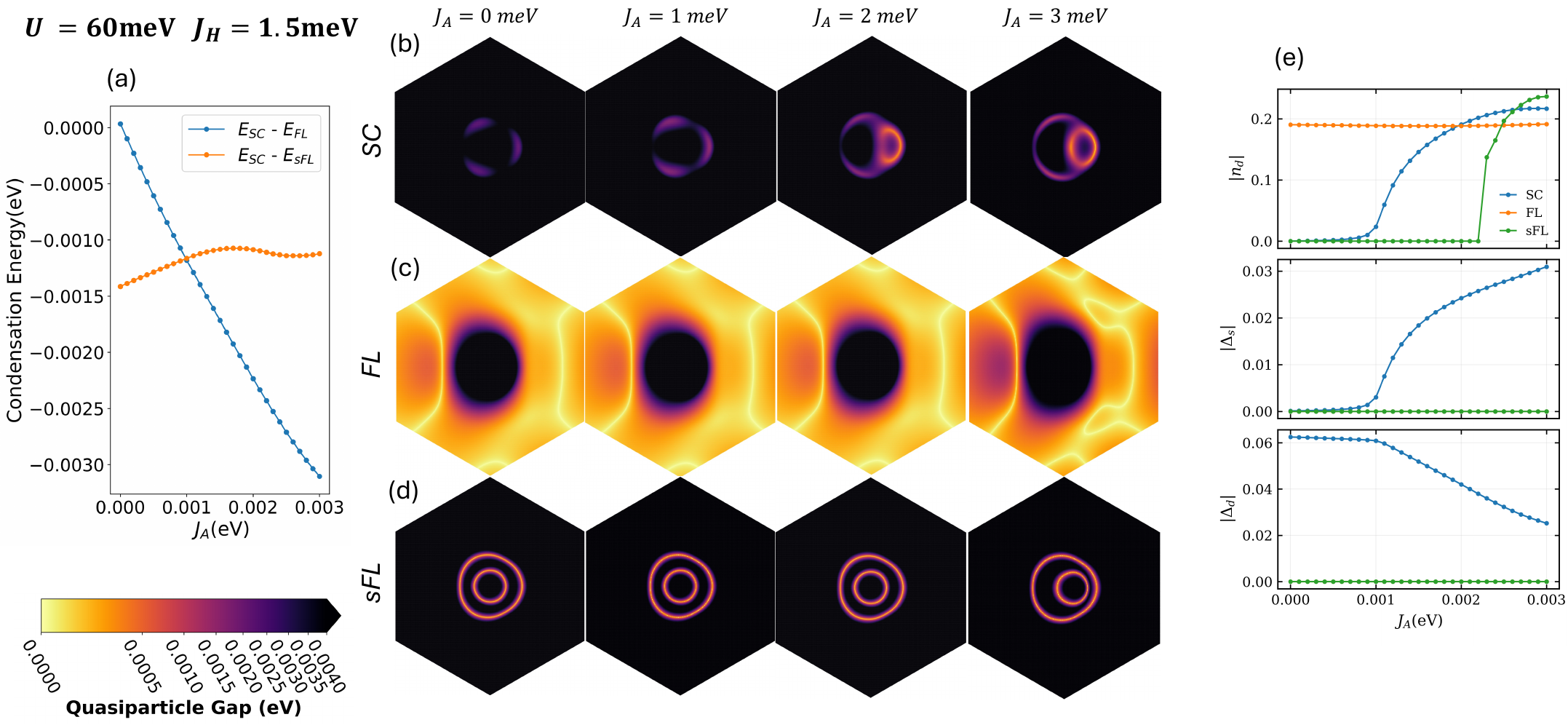}
\caption{Calculations at $\nu=2.5$, $(U, J_H)=(60\mathrm{meV}, 1.5\mathrm{meV})$: (a) Condensation energy versus $U$ of the SC state against the FL and sFL states. (b), (c) and (d) are the quasiparticle gap contour plots near (Bogoliubov) Fermi surface of SC, FL and sFL states respectively at $J_A=0, 1, 2, 3 \ \mathrm{meV}$. The color bar is normalised to show the gap structure within $4\mathrm{meV}$ where the bright yellow lines indicates Fermi surface/nodal lines. (e) Local order parameters: nematic normal order $\abs{n_d}$, anomalous s-wave order $\abs{\Delta_s}$ and anomalous d-wave order $\abs{\Delta_d}$. The finite $\abs{n_d}$ corresponds to mixing of $\ket{\psi_s}$ and $\ket{\psi_d}$.
}
\label{fig:gap_contour_JAvary}
\end{figure}

\subsection{Intersite pairing amplitudes}\label{appdx:NN_pairing_res}
We plot the nearest neighbor pairing amplitudes with respect to the change of $J_A$ (fixing $U=60\mathrm{meV}$) and $U$ (fixing $J_A=3\mathrm{meV}$) respectively. Both plots in Fig.~\ref{fig:nn_pairing_amplitudes} show that the pairing amplitudes in all directions decrease when either the coupling strength $J_A$ or the correlation strength $U$ is large, indicating that either factor alone is sufficient to drive the system into the BEC limit. The sudden rise of $\Delta^s_{0(1)}$ from $J_A=1\mathrm{meV}$ in Fig.~\ref{fig:nn_pairing_amplitudes}(a) corresponds to the second order phase transition from d-wave to s+d-wave shown in Fig.~\ref{fig:phase_diagram}(a). The discontinuity in Fig.~\ref{fig:nn_pairing_amplitudes}(a) at $J_A = 4.9 \mathrm{meV}$ links to a first order transition from s+d-wave to s-wave as the pairing amplitudes become spatially uniform for $J_A \geq 5 \mathrm{meV}$. The sudden drop of pairing amplitudes at $U\approx 5 \mathrm{meV}$ in Fig.~\ref{fig:nn_pairing_amplitudes} matches the crossover from BCS to SC-SC shown in Fig.~\ref{fig:phase_diagram}(a). The discontinuity at $U=55\mathrm{meV}$ corresponds to the gap reconstruction shown in Fig.~\ref{fig:DOS_Vshape}.
\begin{figure}[h]
    \centering
    \includegraphics[width=0.95\linewidth]{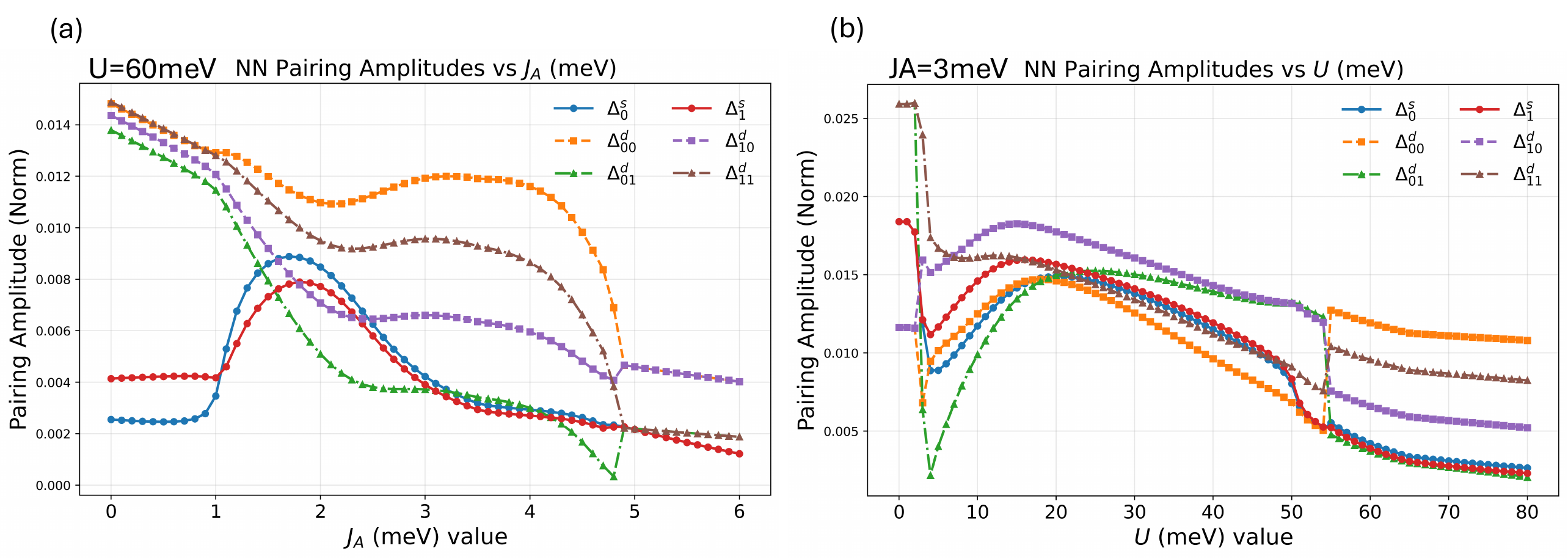}
    \caption{The legends are defined as $\Delta^s_0 \defeq \Delta^f_{\Uparrow 0; \Downarrow 0}(\bm{\delta}_0)$, $\Delta^s_1 \defeq \Delta^f_{\Uparrow 0; \Downarrow 0}(\bm{\delta}_{1(2)})$, $\Delta^d_{0 0} \defeq \Delta^f_{\Uparrow 0; \Downarrow 1}(\bm{\delta}_0)$, $\Delta^d_{0 1} \defeq \Delta^f_{\Uparrow 1; \Downarrow 0}(\bm{\delta}_0)$, $\Delta^d_{1 0} \defeq \Delta^f_{\Uparrow 0; \Downarrow 1}(\bm{\delta}_1)$, $\Delta^d_{1 1} \defeq \Delta^f_{\Uparrow 1; \Downarrow 0}(\bm{\delta}_1)$. Please refer to Appendix.~\ref{appdx:NN_pairing}}
    \label{fig:nn_pairing_amplitudes}
\end{figure}

\subsection{Benchmark: Dimer-Lattice model} \label{appdx:benchmark}
The benchmark of our code using Fabrizio \cite{PhysRevB.76.165110}'s Dimer-Lattice model with flat density of states Fig.~\ref{fig:Benchmark}. All physical observables agree well with Fabrizio's work (orange dashed line) especially in $U\rightarrow0$ and large $U$ limit where He's projector ansatz works the best. The normal and anomalous renormalization factor $Z$ ($\MC{R}$ in this work) and $\Delta$ ($\MC{Q}$ in this work) differs quite a bit due to different gauge choices in our code.

\begin{figure}
    \centering
    \includegraphics[width=0.7\linewidth]{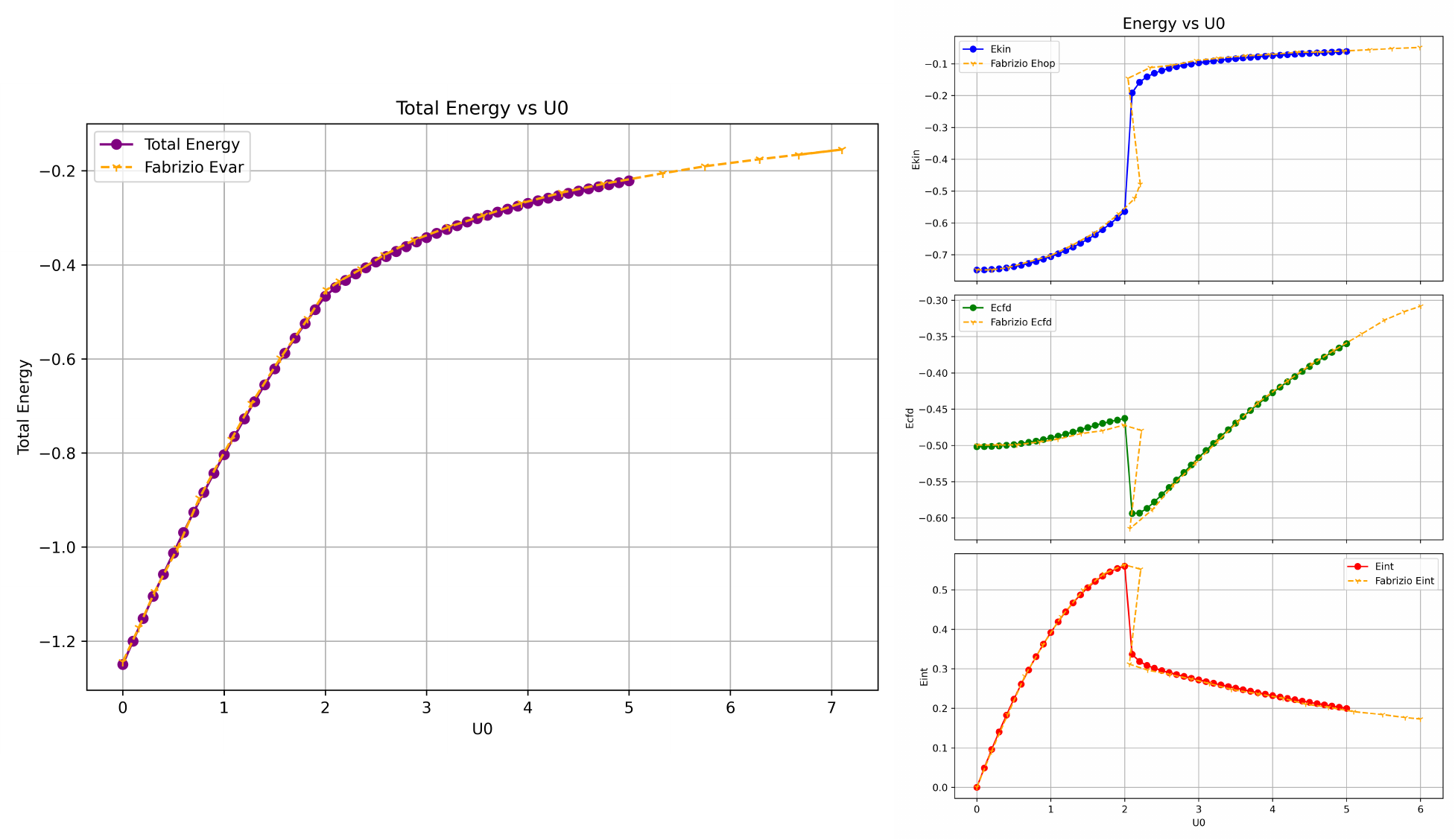}
    \includegraphics[width=0.7\linewidth]{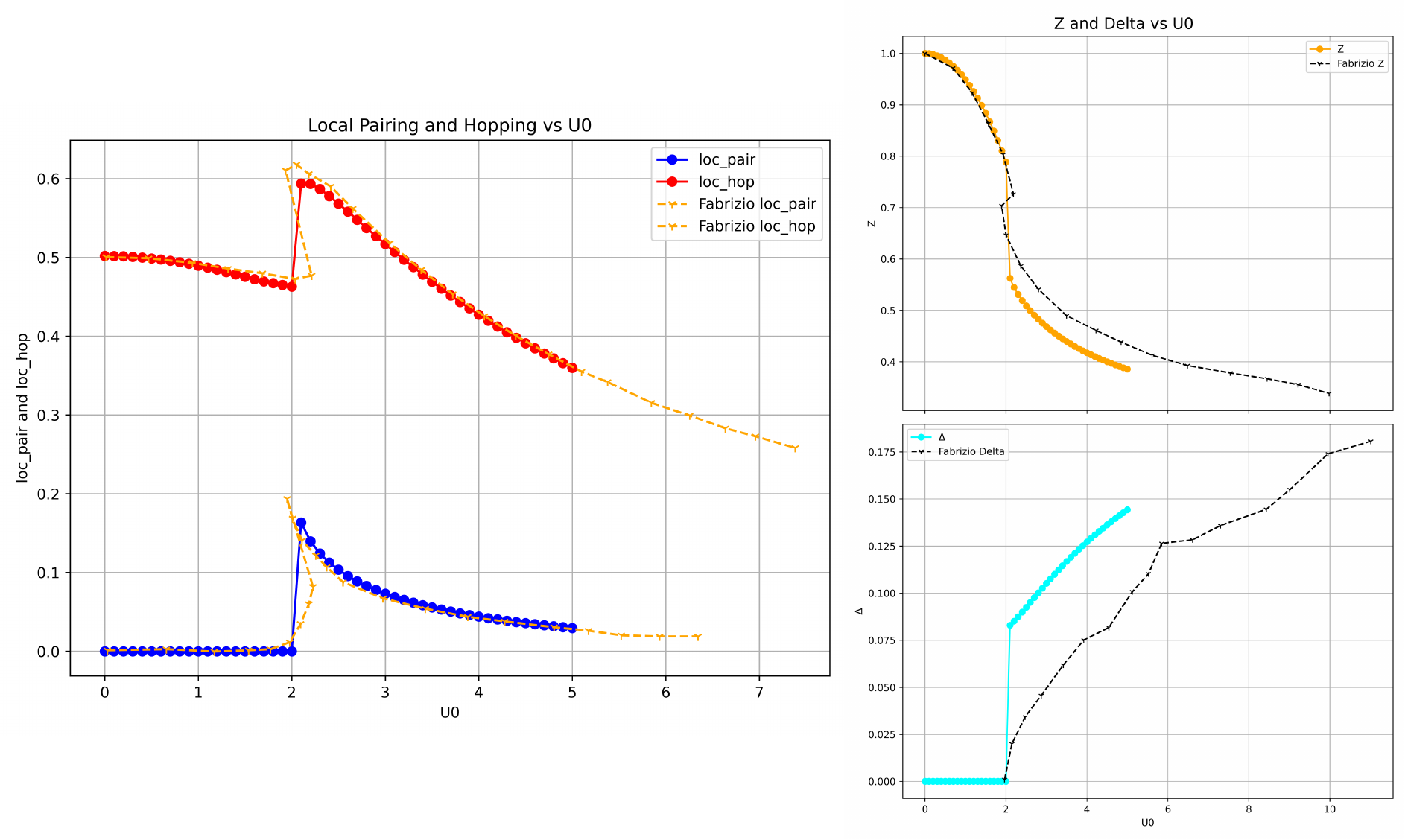}
    \caption{Benchmark with Dimer-Lattice model with flat density of states}
    \label{fig:Benchmark}
\end{figure}

\clearpage






\end{document}